\documentclass[a4paper,UKenglish]{lipics}
 
\usepackage{microtype}%
\usepackage{alltt,paralist,bcprules,amsmath,amssymb,proof,myproof,graphicx,url,color,enumitem,stmaryrd}
\usepackage[svgnames]{xcolor}
\theoremstyle{plain}
\newtheorem{prop}[theorem]{Proposition}
\newtheorem{conj}[theorem]{Conjecture}

\newenvironment{defn}{\vspace*{1ex}\begin{definition}}{\end{definition}\vspace*{1ex}}

\setlist{  
  listparindent=\parindent,
}

\allowdisplaybreaks[1]

\def\qed{\(\Box\)}

\newcommand\Labs{\mathbf{Labs}}
\newcommand\Lab{\psi}
\newcommand\LAB{\Psi}

\newcommand\mult{\mathbf{mul}}

\newcommand\LTE{\Gamma}
\newcommand\rb{\raisebox}
\newcommand\pdt{\models}
\newcommand\mcup{\cup_{\mathtt{mul}}} %

\newcommand\toLty[1]{{#1}^!}
\newcommand\markers{\textit{Markers}}

\newcommand\lty{\delta}
\newcommand\lrty{\beta}

\newcommand\pity{\xi}
\newcommand\prty{\rho}
\newcommand\PTE{\Theta}
\newcommand\subPTE{\leq}
\newcommand\subptyk[1]{\leq_{#1}}
\newcommand\PTEcup{+}
\newcommand\PTEsigma{\Sigma}
\newcommand\term{t}
\newcommand\restrict[1]{\mathbin{\upharpoonright_{#1}}}
\newcommand\Markers{\textit{Markers}}

\newcommand\PMD{\PM^{\mathtt{dir}}}

\newcommand\PMm[1]{\p_{#1}}
\newcommand\eorder{\mathtt{eorder}}

\newcommand\Comp{\mathit{Comp}}

\newcommand\directedT[2]{#2^{\langle #1\rangle}}
\newcommand\plen[1]{|#1|_\mathrm{p}}
\newcommand\CanoTerms[2]{\mathbf{CTerms}_{#1,#2}}
\newcommand\contraction{homeomorphic embedding}
\newcommand\eqbetaeta{=_{\beta\eta}}
\newcommand\eval[1]{\tree(#1)}
\newcommand\tlt{\prec} %
\newcommand\tleq{\preceq} %
\newcommand\tle{\tleq} %
\newcommand\stlambda{\lambda^{\to}}

\newcommand\balanced{\mathbf{bal}}
\newcommand\Wlang{\mathcal{L}_{\mathtt{w}}}
\newcommand\uty{\delta}

\newcommand\rname[1]{(\rn{#1})}

\newcommand\ttytosty[1]{\mathbin{[\![}#1\mathbin{]\!]}}
\newcommand\utytosty[1]{\mathbin{[\![}#1\mathbin{]\!]}}
\newcommand\remeps[1]{#1{\uparrow_{\Te}}}

\newcommand\TTE{\Xi}
\newcommand\tty{\xi}
\newcommand\Tempty{\T_{\epsilon}}
\newcommand\Tplus{\T_{+}}

\newcommand\tropt[1]{}
\newcommand\TREE[1]{\mathbf{Tree}_{#1}}

\newcommand\pty{\gamma}

\newcommand\emphd[1]{\emph{#1}}
\newcommand\pK{\p_{\mathtt{ST}}}
\newcommand\SE{\mathcal{K}}

\newcommand\Lang{\mathcal{L}}

\newcommand\LLang{\mathcal{L}_{\mathtt{leaf}}}
\newcommand\leaves{\textbf{leaves}}

\newcommand\NONTERMS{\mathcal{N}}

\newcommand\tr{\Rightarrow}

\newcommand\order{\mathtt{order}}

\newcommand\app[1]{{#1}\,}
\newcommand\Hole{[\,]}

\newcommand\size[1]{|#1|}

\newcommand\FV{\mathbf{FV}}

\newcommand\nk[1]{{\footnotesize \color{blue}{[#1 -nk]}}}
\newcommand\T{\mathtt{o}}

\newcommand\DCOL{\mathbin{::}}

\newcommand\TE{\Gamma}

\newcommand\sty{\kappa}

\newcommand\ity{\sigma}

\newcommand\COL{\mathbin{:}}
\newcommand\p{\vdash}
\newcommand\PM{\p_n}

\newcommand\pLin{\p_{\mathtt{lin}}}

\newcommand\tree{\mathcal{T}}

\newcommand\GRAM{\mathcal{G}}
\newcommand\TERMS{\Sigma}

\newcommand\Ta{\mathtt{a}}
\newcommand\Tb{\mathtt{b}}
\newcommand\Tc{\mathtt{c}}
\newcommand\Te{\mathtt{e}}
\newcommand\Teps{\Te}

\newcommand\TT[1]{\mathtt{#1}}

\newcommand\RULES{\mathcal{R}}

\newcommand\red{\longrightarrow}
\newcommand\reds{\longrightarrow^*}
\newcommand\redwith[1]{\longrightarrow_{#1}}

\newcommand\redswith[1]{\longrightarrow^*_{#1}}

\newcommand\arity{\mathtt{ar}}

\newcommand\expn[2]{\mathbf{exp}_{#1}(#2)}
\newcommand\set[1]{\{#1\}}
\newcommand\Hra{\rightarrow}
\newcommand\ra{\rightarrow}

\newcommand\dom{\mathit{dom}}

\newcommand\seq[1]{\widetilde{#1}}

\newcommand\IFF{\Leftrightarrow}

\newcommand\InfruleS[3]{\begin{minipage}{#1\textwidth}\infrule{#2}{#3}\end{minipage}}
\newcommand\InfruleSR[4]{\begin{minipage}{#1\textwidth}\infrule[#2]{#3}{#4}\end{minipage}}

\newcommand\hascost{\triangleright}

\newcommand\UE{\Gamma}

\newcommand\myemph[1]{\emph{#1}}

\newif\ifa\afalse

\newcommand{\anp}{}

\newcommand{\urr}[2]{#1 \DCOL_{\mathrm{u}} #2}
\newcommand{\brr}[2]{#1 \DCOL_{\mathrm{b}} #2}
\newcommand{\envdom}[1]{\dom(#1)}

\newcommand{\ibr}{\mathbin{*}}

\newcommand{\lambdap}{{\ensuremath{\lambda^{\to,+}}}}

\newcommand{\defe}{\mathrel{:=}}

\newcommand{\rt}{\mathtt{oR}}
\newcommand{\appseq}[2]{{\overrightarrow{#1}}^{#2}}
\newcommand{\ind}[2]{#1 \le #2}

\newcommand{\br}{\TT{br}}

\newcommand{\empword}{\varepsilon}

\newcommand\wtt[1]{#1^{\star}} %

\newcommand{\LLangE}{\LLang^{\empword}}

\newcommand\G{\GRAM}

\newcommand\he{\preceq}

\newcommand\she{\prec}

\newcommand\sle{\inplus} %

\newcommand\pf{\p_{\mathrm{fst}}} %
\newcommand\ps{\p_{\mathrm{snd}}} %
\newcommand\pof{\p_{\mathrm{ofst}}} %
\newcommand\pofd{\p^{\mathrm{det}}_{\mathrm{ofst}}} %
\newcommand\pos{\p_{\mathrm{osnd}}} %

\newcommand\lab{\dagger}

\newcommand\tp[1]{{{#1}_{\mathrm{p}}}} %
\newcommand\cp{\mathbf{path}} %
\newcommand\wl[1]{{{#1}_{\mathit{l}}}} %
\newcommand\rmd{\mathbf{rmdir}} %

\newcommand\word{\mathbf{word}}

\newcommand\dt{\mathscr{D}}

\newcommand\ored[1]{\red_{#1}}

\newcommand\oreds[1]{\reds_{#1}}

\newcommand\cred{\red_{\mathrm{c}}}
\newcommand\creds{\reds_{\mathrm{c}}}

\newcommand\eset{\emptyset}

\newcommand\E[1]{E^{#1}}

\newcommand\dred[1]{\red_{#1}}

\newcommand\lo{\le_{\mathrm{lg}}}
\newcommand\loo{\ge_{\mathrm{lg}}}
\newcommand\slo{<_{\mathrm{lg}}}

\newcommand\ac[2]{\mspace{-0mu}\raisebox{-1pt}{\rm @}_{#1}\mspace{-1mu}(#2)} %
\newcommand\ntc[1]{\mathtt{NT}(#1)} %

\newcommand\asd[1]{{\footnotesize \color[RGB]{105,10,130}[#1 -asd]}}

\renewcommand\nk[1]{}
\renewcommand\asd[1]{}

\title{Pumping Lemma for Higher-order Languages}
\author[1]{Kazuyuki Asada}
\author[1]{Naoki Kobayashi}
\affil[1]{The University of Tokyo}
\authorrunning{K. Asada and N. Kobayashi}
\Copyright{Kazuyuki Asada and Naoki Kobayashi}
\serieslogo{}%
\volumeinfo%
  {Billy Editor and Bill Editors}%
  {2}%
  {Conference title on which this volume is based on}%
  {1}%
  {1}%
  {1}%
\EventShortName{}
\DOI{10.4230/LIPIcs.xxx.yyy.p}%
\begin{document}
\maketitle
\setcounter{footnote}{0}
\begin{abstract}
We study a pumping lemma for the word/tree languages generated by
higher-order grammars. Pumping lemmas are known up to order-2 word languages
(i.e., for regular/context-free/indexed languages), and have been used to
show that a given language does not belong to the classes of regular/context-free/indexed languages.
We prove a pumping lemma for word/tree languages of arbitrary orders, modulo a conjecture
that a higher-order version of Kruskal's tree theorem holds. 
We also show that
the conjecture indeed holds for the order-2 case, which yields a pumping lemma for
order-2 tree languages and order-3 word languages. 
\end{abstract}
\section{Introduction}
\label{sec:intro}
We study a pumping lemma for higher-order languages, i.e., the languages generated
by higher-order word/tree grammars where non-terminals can take higher-order
functions as parameters. The classes of higher-order languages~\cite{wand74,Mas74,DammTCS,Engelfriet91,EngelfrietHTT}
form an infinite hierarchy, where the classes of order-0, order-1, and order-2 languages are
those of regular, context-free and indexed languages.
Higher-order grammars and languages have been extensively studied by Damm~\cite{DammTCS} and 
Engelfriet~\cite{Engelfriet91,EngelfrietHTT} and recently re-investigated 
in the context of model checking and program 
verification~\cite{Knapik01TLCA,Ong06LICS,KO09LICS,Salvati11ICALP,Kobayashi13JACM,Salvati15OI,Kobayashi13LICS,Parys16ITRS}.

Pumping lemmas~\cite{uvwxy,Hayashi73} are known up to order-2 word languages, and have been used to show that
a given language does not belong to the classes of regular/context-free/indexed languages.
To our knowledge, however, little is known about languages of order-3 or higher. 
Pumping lemmas~\cite{Parys12STACS,Kobayashi13LICS} are also known 
for higher-order \emph{deterministic} grammars (as generators of infinite \emph{trees}, rather than
tree languages), but they cannot be applied to non-deterministic grammars.

In the present paper, we state and prove a pumping lemma for 
unsafe\footnote{%
See, e.g.,
\cite{Salvati15OI} for the distinction between safe vs unsafe languages;
the class of unsafe languages subsumes that of safe languages.}
languages of arbitrary orders
modulo an assumption 
that a ``higher-order version'' of Kruskal's tree theorem~\cite{10.2307/1993287,KruskalSimpleProof}
holds. Let \(\tle\) be the homeomorphic embedding on finite ranked trees\footnote{I.e.,
\(T_1\tle T_2\) if there exists an injective map from the nodes of \(T_1\) to those of \(T_2\)
that preserves the labels of nodes and the ancestor/descendant-relation of nodes; see Section~\ref{sec:pre}
for the precise definition.}, and \(\tlt\) be the strict version of \(\tle\).
The statement of our pumping lemma\footnote{This should perhaps be called a pumping ``conjecture''
since it relies on the conjecture of the higher-order Kruskal's tree theorem.} 
is that for any order-\(n\) infinite tree language \(L\), 
there exist a constant \(c\) and 
a strictly increasing  infinite sequence of trees \(T_0\tlt T_1 \tlt T_2 \tlt \cdots\) in \(L\)
such that \(|T_i| \leq \expn{n}{ci}\) for every \(i\geq 0\),
where \(\expn{0}{x}=x\) and \(\expn{n+1}{x}=2^{\expn{n}{x}}\). 
Due to the correspondence between word/tree languages~\cite{DammTCS,asada_et_al:LIPIcs:2016:6246},
it also implies that for any order-\(n\) infinite word language \(L\) (where \(n\geq 1\)),
there exist a constant \(c\) and 
a strictly increasing  infinite sequence of words \(w_0\tlt w_1 \tlt w_2 \tlt \cdots\) in \(L\)
such that \(|w_i| \leq \expn{n-1}{ci}\) for every \(i\geq 0\),
where \(\tlt\) is the subsequence relation.
The pumping lemma can be used, for example, to show (modulo the conjecture)
that the order-(\(n+1\)) language \(\set{ a^{\expn{n}{k}} \mid k\geq 0}\)
does not belong to the class of order-\(n\) word languages, for \(n>0\).
Thus the lemma would also provide an alternative proof of the strictness of the hierarchy
of the classes of higher-order languages.\footnote{The strictness of the hierarchy
of higher-order \emph{safe} languages has 
 been shown by Engelfriet~\cite{Engelfriet91} using a complexity argument, 
and Kartzow~\cite{KartzowPC} observed that essentially the same argument is applicable to
obtain the strictness of the hierarchy of unsafe languages as well. Their argument cannot
be used for showing that a particular language does not belong to the class of order-\(n\) languages.}

We now informally explain the assumption of ``higher-order Kruskal's tree theorem''
(see Section~\ref{sec:pre} for details).
Kruskal's tree theorem~\cite{10.2307/1993287,KruskalSimpleProof} states that 
the homeomorphic embedding \(\tle\) is a well-quasi order, i.e., that for any infinite sequence
of trees \(T_0,T_1,T_2,\ldots\), there exist \(i<j\) such that \(T_i\tle T_j\).
The homeomorphic embedding \(\tle\) can be naturally lifted (e.g. via the logical relation)
to a family of relations \((\tle_\sty)_\sty\) on higher-order tree functions of type \(\sty\).
Our conjecture of ``higher-order Kruskal's theorem'' states that, for every simple type
\(\sty\), \(\tle_\sty\) is also a well-quasi
order on the functions expressed by the simply-typed \(\lambda\)-terms. 
We prove that the conjecture indeed holds up to order-2 functions, if we take \(\tle_\sty\) 
as the logical relation induced from the homeomorphic embedding \(\tle\).
Thus, our pumping ``lemma'' is indeed true for order-2 tree languages and order-3 word languages.
To our knowledge, the pumping lemma for those languages is novel.
The conjecture remains open for order-3 or higher, which should be of independent interest.

Our proof of the pumping lemma (modulo the conjecture) uses the recent work of Parys~\cite{Parys16ITRS}
on an intersection type system for deciding the infiniteness of the language generated
by a given higher-order grammar, and our previous work on the relationship between higher-order 
word/tree languages~\cite{asada_et_al:LIPIcs:2016:6246}. 

The rest of this paper is organized as follows.
Section~\ref{sec:pre} prepares several definitions and states our pumping lemma and
the conjecture more formally. 
Section~\ref{sec:Parys} derives some corollaries of Parys' result~\cite{Parys16ITRS}.
Section~\ref{sec:wordleafbody} prepares a simplified and specialized 
version of our previous result~\cite{asada_et_al:LIPIcs:2016:6246}. 
Using the results in Sections~\ref{sec:Parys} and~\ref{sec:wordleafbody}, we prove
our pumping lemma (modulo the conjecture) in Section~\ref{sec:indOrder}. Section~\ref{sec:Kruskal} proves the conjecture on
higher-order Kruskal's tree theorem for the order-2 case, by which we obtain the (unconditional)
pumping lemma for order-2 tree languages and order-3 word languages.
Section~\ref{sec:related} discusses related work and Section~\ref{sec:conc} concludes.

\section{Preliminaries}
\label{sec:pre}

We first give basic definitions needed for explaining our main theorem.
We then state the main theorem and provide an overview of its proof.

\subsection{\(\lambda\)-terms and Higher-order Grammars}
\label{sec:termANDgrammar}

This section gives basic definitions for terms and higher-order grammars.
\begin{defn}[types and terms]
The set of \emphd{simple types}, %
ranged over by \(\sty\), is given by:
\(
\sty %
::= \T \mid \sty_1\ra \sty_2
\).
The order %
 of a simple type \(\sty\), written
\(\order(\sty)\) %
is defined by
\(\order(\T)=0\) and \(\order(\sty_1\ra\sty_2)=\max(\order(\sty_1)+1,\order(\sty_2))\).
The type \(\T\) describes trees, and \(\sty_1\to\sty_2\) describes functions
from \(\sty_1\) to \(\sty_2\).
The set of \emph{\lambdap-terms} (or \emphd{terms}), ranged over by \(s,t,u,v\), is defined by:
\[
t %
::= x %
\mid a\,t_1\,\cdots\, t_k \mid \app{t_1}{t_2}\mid \lambda x\COL\sty.t
\mid t_1 + t_2
\]
Here, \(x\) ranges over variables, and \(a\) over constants (which represent tree constructors).
Variables are also called \emph{non-terminals}, ranged over by \(x,y,z,f,g,A,B\);
and constants are also called \emph{terminals}.
A ranked alphabet \(\TERMS\) is a map from a finite set of terminals to natural numbers called \emph{arities};
we implicitly assume a ranked alphabet whose domain contains all terminals discussed,
unless explicitly described.
\(+\) is non-deterministic choice.
As seen below, our simple type system forces that a terminal must be fully applied;
this does not restrict the expressive power, as \(\lambda x_1,\dots,x_k.a\,x_1\cdots x_k\) is available.
We often omit the type \(\sty\) of \(\lambda x : \sty. t\). %
A term is called an \emphd{applicative term} %
if it does not contain \(\lambda\)-abstractions %
nor \(+\), and
called a \emph{\(\stlambda\)-term} if it does not contain %
\(+\). As usual, we identify terms up to the \(\alpha\)-equivalence,
and implicitly apply \(\alpha\)-conversions.

A (simple) type environment \(\SE\) is a
sequence of type bindings of the form \(x\COL\sty\)
such that if \(\SE\) contains \(x\COL\sty\) and \(x'\COL\sty'\) in different positions
then \(x \neq x'\).
In type environments, {non-terminals} are also treated as variables. 
A term \(t\) has type \(\sty\) under \(\SE\) if \(\SE\pK t:\sty\)
is derivable from the following typing rules.\\[6pt]
\InfruleS{0.33}{}{\SE,\,x\COL\sty,\,\SE' \pK x:\sty}
\InfruleS{0.65}{
\TERMS(a)=k \andalso
\SE \pK t_i :\T \ \ (\text{for each \(i \in\set{1,\dots,k}\)})
}{\SE \pK a\,t_1\cdots t_k:\T}
\\[2pt]
\hspace{-19pt}
\InfruleS{0.38}{\SE\pK t_1:\sty_2\ra\sty\andalso \SE\pK t_2:\sty_2}
  {\SE\pK \app{t_1}{t_2}:\sty}
\InfruleS{0.3}{\SE,\,x\COL\sty_1 \pK t:\sty_2}
  {\SE\pK \lambda x\COL\sty_1.t:\sty_1\ra\sty_2} %
\InfruleS{0.3}{\SE\pK t_1:\T \andalso \SE\pK t_2:\T}
  {\SE\pK t_1 + t_2 :\T}\\[6pt]
We consider below only well-typed terms. 
Note that given \(\SE\) and \(t\), there exists at most one type \(\sty\) such that \(\SE\pK t:\sty\).
We call \(\sty\) the type of \(t\) (with respect to \(\SE\)). We often omit ``with respect to \(\SE\)''
if \(\SE\) is clear from context.
The (internal) \emph{order} of \(t\), written \(\order_\SE(t)\), is the largest order of the types of subterms of \(t\),
and the \emph{external order} of \(t\), written \(\eorder_\SE(t)\), is the order of the type of \(t\)
(both with respect to \(\SE\)). We often omit \(\SE\) when it is clear from context.
For example, for \(t = (\lambda x:\T.x)\Te\), \(\order_\emptyset(t)=1\) and \(\eorder_\emptyset(t)=0\).

We call a term \(t\) \emph{ground} (with respect to \(\SE\)) if \(\SE \pK t : \T\).
We call \(t\) a (finite, \(\TERMS\)-ranked) \emphd{tree} if \(t\) is a closed ground applicative
term consisting of only terminals.
We write \(\TREE{\TERMS}\) for the set of \(\TERMS\)-ranked trees,
and use the meta-variable \(\pi\) for trees.

The set of \emph{contexts}, ranged over by \(C,\,D,\,G,\,H\), is defined by \(C::= \Hole \mid C\,t \mid t\,C \mid \lambda x. C\). 
We write \(C[t]\) for the term obtained from \(C\) by replacing \(\Hole\) with \(t\).
Note that the replacement may capture variables; e.g., \((\lambda
x.\Hole)[x]\) is \(\lambda x.x\). 
We call \(C\) a \emph{\((\SE',\sty')\)-\((\SE,\sty)\)-context} if 
\(\SE \pK C : \sty\) is derived by using axiom \(\SE'\pK \Hole : \sty'\).
We also call a \((\emptyset,\sty')\)-\((\emptyset,\sty)\)-context a \emph{\(\sty'\)-\(\sty\)-context}.
The (internal) \emph{order} of a \((\SE',\sty')\)-\((\SE,\sty)\)-context, is the largest order of
the types occurring in the derivation of \(\SE \pK C : \sty\).
A context is called a \emph{\(\stlambda\)-context} if it does not contain %
\(+\).

We define the \emph{size} \(|t|\) of a term \(t\) by:
\(|x|\defe1\), \(|a\,t_1\,\cdots,t_k|\defe 1+|t_1|+\cdots+|t_k|\), \(|s\,t|\defe|s|+|t|+1\), \(|\lambda x.t| \defe |t|+1\), and \(|s+t|\defe|s|+|t|+1\).
The size %
\(|C|\) of a context \(C\) %
is defined similarly, with \(|\Hole|\defe0\).
\end{defn}

\begin{defn}[reduction and language]
The set of \emph{(call-by-name) evaluation contexts} is defined by:
\[
E ::= \Hole\,t_1\cdots t_k %
\mid a\,\pi_1 \cdots \pi_i\,E\,t_1\cdots t_k
\]
and the \emph{call-by-name reduction} for (possibly open) ground terms is defined by:
\begin{align*}&
E[(\lambda x . t)t'] \red E[[t'/x]t]
\qquad\qquad
E[t_1+t_2] \red E[t_i] \quad (i=1,2)
\end{align*}
where \([t'/x]t\) is the usual capture-avoiding substitution.
We write \(\reds\) for the reflexive transitive closure of \(\red\).
A \emph{call-by-name normal form} is a ground term \(t\) such that \(t \ \not\!\!\red t'\) for any \(t'\).
For a closed ground term \(t\), we define 
the \emphd{tree language \(\Lang(t)\) generated by} \(t\) by
\(\Lang(t) \defe \set{\pi \mid t \reds \pi}\). 
For a closed ground \(\stlambda\)-term \(t\), \(\Lang(t)\) is a singleton set \(\set{\pi}\);
we write \(\tree(t)\) for such \(\pi\) and call it \emph{the tree of \(t\)}.
\end{defn}
Note 
that \(t \reds t'\) implies \([s/x]t \reds [s/x]t'\), and
that the set of call-by-name normal forms equals the set of trees and ground terms of the form \(E[x]\).
For \(x : \sty \pK t : \T\) where \(t\) does not contain the non-deterministic choice, 
\(t\) is called \emph{linear} (with respect to \(x\))
if \(x\) occurs exactly once in the call-by-name normal form of \(t\).
A pair of contexts \(\Hole : \sty \pK C: \T\) and \(\Hole : \sty \pK D: \sty\) is called \emph{%
linear} %
if \(x : \sty \pK C[D^i[x]] : \T\) is linear for any \(i \ge 0\) where
\(x\) is a fresh variable that is not captured by the context applications.

\begin{defn}[higher-order grammar]
A \emphd{higher-order grammar} (or \emph{grammar} for short) is a quadruple \((\TERMS,\NONTERMS,\RULES,S)\), where
\begin{inparaenum}[(i)]
\item \(\TERMS\) is a ranked alphabet;
\item \(\NONTERMS\) is a map from a finite set of non-terminals 
to their types;
\item \(\RULES\) is a finite set of \emphd{rewriting rules} of the form
\( A\,\to \lambda x_1.\cdots \lambda x_\ell.t\),
where \(\NONTERMS(A) = \kappa_1\ra\cdots\ra\kappa_\ell\ra\T\),
\(t\) is an applicative term, and 
\(\NONTERMS,x_1\COL\kappa_1,\ldots,x_\ell\COL\kappa_\ell\pK t:\T\) holds;
\item \(S\) is a non-terminal called the \emphd{start symbol},
and \(\NONTERMS(S)=\T\). 
\end{inparaenum}
The \emphd{order} of a grammar \(\GRAM\)
is the largest order 
of the types of non-terminals.
We sometimes write \(\TERMS_\GRAM,\NONTERMS_\GRAM,\RULES_\GRAM,S_\GRAM\) for the four components
of \(\GRAM\).
We often write 
\(A\,x_1\,\cdots\,x_k\to t\) for the rule \(A\to \lambda x_1.\cdots \lambda x_k.t\).

For a grammar \(\GRAM=(\TERMS,\NONTERMS,\RULES,S)\), 
the rewriting relation \(\redwith{\GRAM}\) is defined by:\\[1ex]
\InfruleS{0.46}{(A\to \lambda x_1.\cdots \lambda x_k.t)\in \RULES}
{A\,t_1\,\cdots\,t_k \redwith{\GRAM} [t_1/x_1,\ldots,t_k/x_k]t}
\InfruleS{0.48}{t_i\redwith{\GRAM}t_i'\andalso i\in\set{1,\ldots,k}\andalso \TERMS(a)=k}
 {a\,t_1\,\cdots\,t_k \redwith{\GRAM}
  a\,t_1\,\cdots\,t_{i-1}\,{t'_{i}}\,t_{i+1}\,\cdots\,t_k
}\\[1ex]
We write \(\redswith{\GRAM}\) for the reflexive transitive closure of \(\redwith{\GRAM}\).
The \emphd{tree language generated by} \(\GRAM\), written \(\Lang(\GRAM)\), is
the set \(\set{\pi %
\mid S\redswith{\GRAM} \pi}\).
\end{defn}

\begin{remark}
An order-\(n\) grammar can also be represented as a ground closed order-\(n\) \(\lambdap\)-term extended with
the Y-combinator such that \(Y_\sty x.t\red [Y_\sty x.t/x]t\). Conversely, any ground closed order-\(n\) \(\lambdap\)-term
(extended with \(Y\)) can be represented as an equivalent order-\(n\) grammar. 
We shall make use of this correspondence in Appendix.
\end{remark}

The grammars defined above may also be viewed as generators of word languages.
\begin{defn}[word alphabet / \(\br\)-alphabet]
We call a ranked alphabet \(\TERMS\) a \emph{word alphabet} if
it has a special nullary terminal \(\TT{e}\) and
all the other terminals have arity \(1\);
also we call a grammar \(\GRAM\) a \emphd{word grammar} if its alphabet is a word alphabet.
For a tree \(\pi = a_1(\cdots (a_n\,\TT{e})\cdots)\) of a word grammar, we define \(\word(\pi) = a_1\cdots a_n\).
The \emphd{word language} generated by a word grammar \(\GRAM\), written \(\Wlang(\GRAM)\), is 
\(\set{\word(\pi) \mid \pi \in\Lang(\GRAM)}\).

The frontier word of a tree \(\pi\), written \(\leaves(\pi)\),
is the sequence of symbols in the leaves of \(\pi\). It is
defined inductively by:
\(\leaves(a) = a\)  when \(\TERMS(a)=0\), and \(\leaves(a\,\pi_1\,\cdots\,\pi_k) =
\leaves(\pi_1)\cdots \leaves(\pi_k)\) when \(\TERMS(a)=k>0\).  
The \emphd{frontier language} generated by \(\GRAM\), written \(\LLang(\GRAM)\),
is the set:
\(\set{\leaves(\pi) \mid S\redswith{\GRAM} \pi
}\).
A \emph{\(\br\)-alphabet} is a ranked alphabet such that it has
a special binary constant \(\br\) and a special nullary constant \(\Te\) and the other constants are nullary.
We consider \(\TT{e}\) as the empty word \(\empword\):
for a grammar with a \(\br\)-alphabet, we also define
\(
\LLangE(\GRAM) \defe
(\LLang(\GRAM) \setminus \set{\Te}) \cup \set{\empword \mid \Te \in \LLang(\GRAM)}
\).
We call a tree \(\pi\) an \emph{\(\Te\)-free \(\br\)-tree} if 
it is a tree of some \(\br\)-alphabet but does not contain \(\Te\).
\end{defn}

We note that the classes of order-0, order-1, and order-2 word languages coincide with those
of regular, context-free, and indexed languages, respectively~\cite{wand74}.

\anp
\subsection{Homeomorphic Embedding and Kruskal's Tree Theorem}
\label{sec:homeoKrus}

In our main theorem, we use the notion of homeomorphic embedding for trees.
\begin{definition}[homeomorphic embedding]
\label{def:homeoEmb}
Let \(\TERMS\) be an arbitrary ranked alphabet.
The \emph{homeomorphic embedding} order \(\he\) between \(\TERMS\)-ranked 
trees\footnote{In the usual definition, a quasi order on labels (tree constructors) is assumed.
Here we fix the quasi-order on labels to the identity relation.}
is inductively defined by the following rules:
\begin{align*}
\frac{
\pi_i \he \pi'_i \quad\text{(for all \(i\le k\))}
}{
a\,\pi_1\cdots \pi_k \he a\,\pi'_1\cdots \pi'_k
}
(k=\TERMS(a))
\qquad
\frac{
\pi \he \pi_i
}{
\pi \he a\,\pi_1\cdots \pi_k
}
(k=\TERMS(a)>0,\ i\in\set{1,\dots, k})
\end{align*}
\end{definition}
For example, \(\br\,\Ta\,\Tb \he \br\,(\br\,\Ta\,\Tc)\,\Tb\).
We extend \(\he\) to words: for \(w=a_1\cdots a_n\) and \(w'=a'_1\cdots a'_{n'}\),
we define \(w \he w'\) if \(a_1(\cdots(a_n(\Te))) \he a'_1(\cdots(a'_{n'}(\Te)))\),
where \(a_i\) and \(a'_i\) are regarded as unary constants and \(\Te\) is a nullary constant
(this order on words is nothing but the (scattered) subsequence relation).
We write \(\pi \she \pi'\) %
if \(\pi \he \pi'\) and \(\pi' \not\he \pi\).

Next we explain a basic property on \(\he\), Kruskal's tree theorem.
A \emph{quasi-order} (also called a \emph{pre-order}) is a reflexive and transitive relation.
A \emph{well quasi-order} %
on a set \(S\) is a quasi-order \(\le\) on \(S\) such that
for any infinite sequence \((s_i)_i\) of elements in \(S\) there exist \(j<k\) such that
\(s_j \le s_k\).
\begin{prop}[Kruskal's tree theorem~\cite{10.2307/1993287}]\label{prop:Kruskal}
For any (finite) ranked alphabet \(\TERMS\),
the homeomorphic embedding \(\he\) on \(\TERMS\)-ranked trees is a well quasi-order.
\end{prop}

\anp
\subsection{Conjecture and Pumping Lemma for Higher-order Grammars}

As explained in Section~\ref{sec:intro}, our pumping lemma makes use of a conjecture on
``higher-order'' Kruskal's tree theorem, which is stated below.
\begin{conj}
\label{conj:abstHighKruskal}
There exists a family \((\he_\sty)_\sty\) of relations indexed by simple types such that
\begin{itemize}
\item \(\he_{\sty}\) is a well quasi-order on the set of closed \(\stlambda\)-terms of type \(\sty\) modulo
      \(\beta\eta\)-equivalence; i.e., for an infinite sequence \(t_1,t_2,\dots\) of closed
      \(\stlambda\)-terms of type \(\sty\), there exist \(i<j\) such that \(t_i \he_\sty t_j\).
\item \(\he_{\T}\) is a conservative extension of \(\he\), i.e.,
\(t\he_{\T}t'\) if and only if \(\tree(t)\he \tree(t')\).
\item \((\he_\sty)_\sty\) is closed under applications, i.e.,
      if \(t \he_{\sty_1 \to \sty_2} t'\) and \(s \he_{\sty_1} s'\) then \(t\,s \he_{\sty_2} t'\,s'\).
\end{itemize}
\end{conj}
A candidate of \((\he_\sty)_\sty\) would be the logical relation induced from \(\he\). 
Indeed, if we choose the logical relation as \((\he_\sty)_\sty\),
the above conjecture holds up to order-2 (see Theorem~\ref{thm:order2-kruskal} in Section~\ref{sec:Kruskal}).

Actually, for our pumping lemma, 
the following, slightly weaker property called the \emph{periodicity} is sufficient.
\begin{conj}[Periodicity]
\label{conj:abstHighPeriodic}
There exists a family \((\he_\sty)_\sty\) indexed by simple types such that
\begin{itemize}
\item \(\he_{\sty}\) is a quasi-order on the set of closed \(\stlambda\)-terms of type \(\sty\) modulo \(\beta\eta\)-equivalence.
\item for any %
\(\pK t: \sty \to \sty\) and \(\pK s: \sty\), 
there exist \(i,j>0\) such that 
\[
t^i\,s \he_\sty t^{i+j}\,s \he_\sty t^{i+2j}\,s \he_\sty \cdots.
\]
\item \(\he_{\T}\) is a conservative extension of \(\he\). %
\item \((\he_\sty)_\sty\) is closed under applications.
\end{itemize}
\end{conj}
Note that Conjecture~\ref{conj:abstHighKruskal} implies
Conjecture~\ref{conj:abstHighPeriodic}, since if the former holds,
for the infinite sequence \((t^i\,s)_i\), there exist \(i<i+j\) such that
\(t^i\,s \he_\sty t^{i+j}\,s\), and then by the monotonicity of \(u \mapsto t^{j}\,u\), we have \(t^{i+kj}\,s \he_\sty
t^{i+(k+1)j}\,s\) for any \(k\geq 0\).

We can now state our pumping lemma.
\begin{theorem}[pumping lemma]
\label{thm:pumping}
Assume that Conjecture~\ref{conj:abstHighPeriodic} holds.
Then, for any order-\(n\) tree grammar \(\G\) such that \(\Lang(\G)\) is infinite,
there exist an infinite sequence of trees \(\pi_0,\pi_1,\pi_2,\ldots\in \Lang(\G)\),
and constants \(c\), \(d\) such that:
\begin{inparaenum}[(i)]
\item
\label{item:mainthmstrict}
\(\pi_0 \she \pi_1 \she \pi_2 \she \cdots \), and %
\item
\label{item:mainthmsize}
\(\size{\pi_i} \le \expn{n}{ci+d}\) for each \(i\ge 0\).
\end{inparaenum}
Furthermore, we can drop the assumption on Conjecture~\ref{conj:abstHighPeriodic}
when \(\G\) is of order up to 2.
\end{theorem}

By the correspondence between order-\(n\) tree grammars and order-\((n+1)\) grammars~\cite{DammTCS,asada_et_al:LIPIcs:2016:6246}, we also have:
\begin{corollary}[pumping lemma for word languages]
\label{thm:pumping-word}
Assume that Conjecture~\ref{conj:abstHighPeriodic} holds.
Then, for any order-\(n\) word grammar \(\G\) (where \(n\geq 1\)) such that \(\Wlang(\G)\) is infinite,
there exist an infinite sequence of words \(w_0,w_1,w_2,\ldots\in \Wlang(\G)\),
and constants \(c\), \(d\) such that:
\begin{inparaenum}[(i)]
\item
\(w_0 \she w_1 \she w_2 \she \cdots \), and %
\item
\(\size{w_i} \le \expn{n-1}{ci+d}\) for each \(i\ge 0\).
\end{inparaenum}
Furthermore, we can drop the assumption on Conjecture~\ref{conj:abstHighPeriodic}
when \(\G\) is of order up to 3.
\end{corollary}

We sketch the overall structure of the proof of Theorem~\ref{thm:pumping} below.
Let \(\GRAM\) be an order-\(n\) tree grammar.
By using the recent type system of Parys~\cite{Parys16ITRS}, if \(\Lang(\GRAM)\) is infinite,
we can construct order-\(n\) linear \(\stlambda\)-contexts \(C,D\) and an order-\(n\) \(\stlambda\)-term \(t\) such that
\(\set{\tree(C[D^i[t]])\mid i\ge 0}\) (\(\subseteq \Lang(\GRAM)\)) is infinite.
It then suffices to show that there exist constants \(p\) and \(q\) such that
\(\tree(C[D^{p}[t]]) \she \tree(C[D^{p+q}[t]]) \she \tree(C[D^{p+2q}[t]]) \she  \cdots\). The bound 
\(\tree(C[D^{p+iq}])\leq \expn{n}{c+id}\) would then follow
immediately from the standard result on an upper-bound on the size of \(\beta\)-normal forms.
Actually, assuming Conjecture~\ref{conj:abstHighPeriodic}, we can easily deduce 
\(\tree(C[D^{p}[t]]) \he \tree(C[D^{p+q}[t]]) \he \tree(C[D^{p+2q}[t]]) \he  \cdots\).
Thus, the main remaining difficulty is to show that the ``strict'' inequality
holds periodically. To this end, we prove it by induction
on the order, by making use of three ingredients: an extension of the result of
Parys' type system (again)~\cite{Parys16ITRS},
an extension 
of our previous work on a translation from word languages to tree languages~\cite{asada_et_al:LIPIcs:2016:6246},
and Conjecture~\ref{conj:abstHighPeriodic}. 
In Sections~\ref{sec:Parys} and \ref{sec:wordleafbody}, we derive corollaries from the results of
Parys' and our previous work respectively. We then provide the proof of Theorem~\ref{thm:pumping}
(except the statement ``Furthermore, ...'') in Section~\ref{sec:indOrder}. We then, in Section~\ref{sec:Kruskal},
 discharge the assumption on
Conjecture~\ref{conj:abstHighPeriodic} for order up to 2, by proving 
Conjecture~\ref{conj:abstHighKruskal}  for order up to 2.

\anp
\section{Corollaries of Parys' Results}
\label{sec:Parys}

\newcommand{\FF}{t}   %
\newcommand{\CC}{C_0} %
\newcommand{\DD}{C}   %
\newcommand{\EE}{D}   %
\newcommand{\TR}{\pi} %

Parys~\cite{Parys16ITRS} developed an intersection type system with
judgments of the form \(\TE \p s: \tau \hascost c\), where 
\(s\) is a term of a simply-typed, infinitary \(\lambda\)-calculus (that corresponds to the \(\lambda\)Y-calculus)
extended with choice, and \(c\) is a natural number.
He proved that for any order-\(n\) closed ground term \(s\),
(i) \(\emptyset \p s: \tau\hascost c\) implies that \(s\) can be reduced to a tree \(\TR\)
such that \(c\leq |\TR|\),
and (ii) if \(s\) can be reduced to a tree \(\TR\), then 
\(\emptyset \p s: \tau\hascost c\) holds for some \(c\) such that \(|\TR|\leq \expn{n}{c}\).

Let \(\GRAM\) be an order-\(n\) tree grammar and \(S\) be its start symbol.
By Parys' result,\footnote{See Section~6 of \cite{Parys16ITRS}. Parys considered
a \(\lambda\)-calculus with infinite regular terms, but the result can be easily adapted to
terms of grammars.} if \(\Lang(\GRAM)\)
is infinite, there exists a derivation for \(\emptyset\p S:\T\hascost{c_1+c_2+c_3}\)
in which \(\PTE \p A:\pty\hascost{c_1+c_2}\) is derived from \(\PTE \p A:\pty\hascost{c_1}\)
for some non-terminal \(A\).
Thus, by ``pumping'' the derivation of
\(\PTE \p A:\pty\hascost{c_1+c_2}\) from \(\PTE \p A:\pty\hascost{c_1}\),
we obtain a derivation for
\(\emptyset \p S: \T\hascost c_1+kc_2+c_3\)
for any \(k\ge 0\). From the derivation,
we obtain a \(\stlambda\)-term \(t\) and \(\stlambda\)-contexts \(\DD,\EE\) of at most order-\(n\),
such that \(\DD[\EE^k[t]]\) generates a tree \(\TR_k\) such that
\(c_1+kc_2+c_3\leq |\TR_k|\).
By further refining the argument above (see 
Appendix~\ref{sec:parys-details} 
for details), 
we can also ensure that the pair \((\DD, \EE)\) is linear.
Thus, we obtain the following lemma.
\begin{lemma}\label{lem:parys3}
Given an order-\(n\) tree grammar \(\G\) such that \(\Lang(\G)\) is infinite,
there exist order-\(n\)  linear \(\stlambda\)-contexts \(\DD,\EE\), and
an order-\(n\) \(\stlambda\)-term \(t\) 
such that
\begin{enumerate}
\item
\(\set{\tree(C[D^k[t]]) \mid k \ge 1} \subseteq \Lang(\G)\)
\item
\(\set{\tree(C[D^{\ell_k}[t]]) \mid k \ge 1}\) is infinite for any strictly increasing sequence
 \((\ell_k)_k\).
\end{enumerate}
\end{lemma}

By slightly modifying Parys' type system, we can also reason about
the length of a particular path of a tree.
Let us annotate each constructor \(a\) as \(\directedT{i}{a}\),
where  \(0\leq i\leq \TERMS(a)\). We call \(i\) a \emph{direction}.
We define \(\plen{\TR}\) by:
\[
\plen{\directedT{0}{a}\,\TR_1\,\cdots\,\TR_k}=1
\qquad 
\plen{\directedT{i}{a}\,\TR_1\,\cdots\,\TR_k}=\plen{\TR_i}+1 \quad(1\le i \le k).
\]
We define \(\rmd\) as the function that removes all the direction annotations.

\begin{lemma}\label{lem:parys5}
For any 
order-\(n\) linear \(\stlambda\)-contexts \(C,D\)
and any order-\(n\) \(\stlambda\)-term \(t\)
such that \(\set{\tree(C[D^k[t]]) \mid k \ge 1}\) is infinite,
there exist
direction-annotated order-\(n\) linear \(\stlambda\)-contexts \(G,H\),
a direction-annotated order-\(n\) \(\stlambda\)-term \(u\),
and \(p,q>0\)
such that
\begin{enumerate}
\item \(\rmd(\tree(G[H^k[u]]))=\tree(C[D^{pk+q}[t]])\)
for any \(k\geq 1\)
\item \(\set{\plen{\tree(G[H^{\ell_k}[u]])} \mid k\geq 1}\) is infinite for any strictly increasing sequence
 \((\ell_k)_k\).
\end{enumerate}
\end{lemma}

\anp
\section{Word to Frontier Transformation}
\label{sec:wordleafbody}

We have
an ``order-decreasing'' transformation~\cite{asada_et_al:LIPIcs:2016:6246}
that transforms an order-\((n+1)\) word grammar \(\G\) 
to an order-\(n\) tree grammar \(\G'\) (with a \(\br\)-alphabet) such that
\(\Wlang(\G) = \LLangE(\G')\).
We use this as a method for induction on order;
this method was originally suggested by Damm~\cite{DammTCS} for safe languages.

The transformation in the present paper has been modified from the original one in~\cite{asada_et_al:LIPIcs:2016:6246}.
On the one hand, the current transformation is a specialized version in that
we apply the transformation only to \(\stlambda\)-terms instead of terms of (non-deterministic) grammars.
On the other hand, the current transformation has been strengthened in that
the transformation preserves linearity.
Due to the preservation of linearity, a \emph{single}-hole context is
transformed to a \emph{single}-hole context, and the uniqueness of an occurrence of \(\Hole\)
will be utilized for the calculation of the size of ``pumped trees'' in Lemma~\ref{lem:orderInd}.

The definition of the current transformation is given just by
translating the transformation rules in~\cite{asada_et_al:LIPIcs:2016:6246} by
following the idea of the embedding of \(\stlambda\)-terms into grammars.
For the detailed definition, see 
Appendix~\ref{sec:wordleaf}.
By using this transformation, we have:

\begin{lemma}\label{lem:iteLinPres}
Given order-\(n\) \(\stlambda\)-contexts \(C\), \(D\), and an order-\(n\) \(\stlambda\)-term \(t\)
such that 
\begin{itemize}
\item
the constants in \(C,\,D,\,t\) are in a word alphabet,
\item
\(\set{ \tree(C[D^{\ell_i}[t]]) \mid i \ge 0 }\) is infinite for any strictly increasing sequence
 \((\ell_i)_i\), and
\item
\(C\) and \(D\) are %
 linear,
\end{itemize}
there exist order-\((n-1)\) \(\stlambda\)-contexts \(G\), \(H\), order-\((n-1)\) \(\stlambda\)-term \(u\),
and some constant numbers \(c,\,d \ge 1\)
such that 
\begin{itemize}
\item
the constants in \(G,\,H,\,u\) are in a \(\br\)-alphabet
\item
for \(i \ge 0\),
\(\tree(G[H^i[u]])\) is either an \(\Te\)-free \(\br\)-tree or \(\Te\), and
\[
\word(\tree(C[D^{ci+d}[t]]))
 = 
\begin{cases}
\empword & (\tree(G[H^i[u]]) = \Te)
\\
\leaves(\tree(G[H^i[u]])) & (\tree(G[H^i[u]]) \neq \Te)
\end{cases}
\]
\item
\(G\) and \(H\) are %
 linear.
\end{itemize}
\end{lemma}
\begin{proof}
The preservation of meaning (the second condition) 
follows as a corollary of a theorem in~\cite{asada_et_al:LIPIcs:2016:6246}.
Also, the preservation of linearity (the third condition) can be proved 
in a manner similar to the proof of the preservation of
 meaning in~\cite{asada_et_al:LIPIcs:2016:6246}, 
using a kind of subject-reduction.
See 
Appendix~\ref{sec:wordleaf} 
for the detail.
\end{proof}

\anp
\section{Proof of the Main Theorem}
\label{sec:indOrder}

We first prepare some lemmas.
\begin{lemma}\label{lem:strictlemma}
For \(\Te\)-free \(\br\)-trees \(\pi\) and \(\pi'\),
if \(\pi \she \pi'\) then \(\leaves(\pi) \she \leaves(\pi')\).
\end{lemma}
\begin{proof}
We can show that \(\pi \he \pi'\) implies \(\leaves(\pi) \he \leaves(\pi')\)
and then the statement, both by straightforward induction on the derivation of \(\pi \he \pi'\).
\end{proof}
\begin{remark}
The above lemma does not necessarily hold for an arbitrary ranked alphabet,
especially that with a unary constant; e.g., \(\Ta\,\Te \she \Ta\,(\Ta\,\Te)\) but their leaves are both \(\Te\).
Also, it does not hold if a tree contains \(\Te\) and if we regard \(\Te\) as \(\empword\) in the leaves word;
e.g., for \(\br\,\Ta\,\Tb \she \br\,(\br\,\Ta\,\Te)\,\Tb\), their leaves are \(\Ta\Tb \she \Ta\Te\Tb\),
but if we regard \(\Te\) as \(\empword\) then \(\Ta\Tb \not\she \Ta\Tb\).
\end{remark}
\begin{lemma}
\label{lem:rmdirOrdMon}
For direction-annotated trees \(\pi\) and \(\pi'\),
if \(\pi \she \pi'\) then \(\rmd(\pi) \she \rmd(\pi')\).
\end{lemma}
\begin{proof}
We can show that \(\pi \he \pi'\) implies \(\rmd(\pi) \he \rmd(\pi')\)
and then the statement, both by straightforward induction on the derivation of \(\pi \he \pi'\).
\end{proof}

Now, we prove the following lemma (Lemma~\ref{lem:orderInd}) by the induction on order.
Theorem~\ref{thm:pumping} (except the last statement) will then follow as an immediate corollary of
Lemmas~\ref{lem:parys3} and \ref{lem:orderInd}.

\begin{lemma}
\label{lem:orderInd}
Assume that the statement of Conjecture~\ref{conj:abstHighPeriodic} is true.
For any 
order-\(n\) linear \(\stlambda\)-contexts \(C,D\)
and any order-\(n\) \(\stlambda\)-term \(t\)
such that \(\set{\tree(C[D^i[t]]) \mid i \ge 1}\) is infinite,
there exist \(c\), \(d\), \(j\), \(k \ge 1\) such that
\begin{itemize}
\item
\(\tree(C[D^j[t]]) \she \tree(C[D^{j+k}[t]]) \she \tree(C[D^{j+2k}[t]]) \she \cdots\)
\item
\(\size{\tree(C[D^{j+ik}[t]])} \le \expn{n}{ci+d}
\qquad (i = 0,1,\dots)
\)
\end{itemize}
\end{lemma}
\begin{proof}
The proof proceeds by induction on \(n\).
The case \(n=0\) is clear, and we discuss the case \(n>0\) below.
By Lemma~\ref{lem:parys5}, from \(C\), \(D\), and \(t\),
we obtain 
direction-annotated order-\(n\) linear \(\stlambda\)-contexts \(G,H\),
a direction-annotated order-\(n\) \(\stlambda\)-term \(u\),
and \(j_0,\, k_0 >0\)
such that
\begin{gather}
\label{eq:relbefore}
\text{\(\rmd(\tree(G[H^i[u]]))=\tree(C[D^{j_0 + ik_0}[t]])\) for any \(i \ge 1\)}
\\
\label{eq:path-inf}
\text{\(\set{\plen{\tree(G[H^{\ell_i}[u]])} \mid i \ge 1}\) is infinite for any strictly increasing sequence
 \((\ell_i)_i\).}
\end{gather}

Next we transform \(G\), \(H\), and \(u\) by choosing a path according to directions,
i.e., we define \(\tp{G}\), \(\tp{H}\), and \(\tp{u}\) as the contexts/term obtained from
\(G\), \(H\), and \(u\) by replacing
each \(\directedT{i}{a}\) with: (i) \(\lambda x_1\dots x_\ell. a_i x_i\)
if \(i>0\) or (ii) \(\lambda x_1\dots x_\ell. \Te\) if \(i=0\),
where \(\ell=\TERMS(a)\) and \(a_i\) is a fresh
 unary constant.
For any \(i \ge 0\),
\begin{equation}
\label{eq:reldirsize}
\plen{\tree(G[H^{i}[u]])} = |\word(\tree(\tp{G}[{\tp{H}}^{i}[\tp{u}]]))|+1.
\end{equation}
We also define a function \(\cp\) on trees annotated with directions, by the following induction:
\(\cp(\directedT{i}{a}\,\pi_1\cdots \pi_\ell) = a_i\,\cp(\pi_i)\) if \(i>0\) and
\(\cp(\directedT{0}{a}\,\pi_1\cdots \pi_\ell) = \Te\).
Then for any \(i \ge 0\),
\begin{gather}
\label{eq:reldir}
\cp(\tree(G[H^{i}[u]])) = \tree(\tp{G}[{\tp{H}}^{i}[\tp{u}]]).
\end{gather}

By~\eqref{eq:path-inf} and~\eqref{eq:reldirsize},
\(\set{ \tree(\tp{G}[{\tp{H}}^{\ell_i}[\tp{u}]]) \mid i \ge 0 }\) 
is infinite for any strictly increasing sequence \((\ell_i)_i\).
Also, the transformation from \(G,\,H\) to \(\tp{G}\), \(\tp{H}\) preserves the linearity, because:
let \(N\) be the normal form of \(G[H^i[x]]\) where \(x\) is fresh,
and \(\tp{N}\) be the term obtained by applying this transformation to \(N\);
then \(\tp{G}[\tp{H}^i[x]] \reds \tp{N}\), and by the infiniteness of
\(\set{ \tree(\tp{G}[{\tp{H}}^{i}[\tp{u}]]) \mid i \ge 0 }\), \(\tp{N}\) must contain \(x\),
which implies \(\tp{N}\) is a linear normal form.

Now we decrease the order by using the transformation in Section~\ref{sec:wordleafbody}.
By Lemma~\ref{lem:iteLinPres} to \(\tp{G}\), \(\tp{H}\), and \(\tp{u}\),
there exist order-\((n-1)\) linear \(\stlambda\)-contexts \(\wl{G}\), \(\wl{H}\), an order-\((n-1)\) \(\stlambda\)-term \(\wl{u}\),
and some constant numbers \(c',\,d' \ge 1\)
such that, for any \(i \ge 0\),
\(\tree(\wl{G}[\wl{H}^i[\wl{u}]])\) is either an \(\Te\)-free \(\br\)-tree or \(\Te\), and
\begin{equation}
\label{eq:leavestree}
\word(\tree(\tp{G}[{\tp{H}}^{c'i+d'}[\tp{u}]])) 
= 
\begin{cases}
\empword & (\tree(\wl{G}[\wl{H}^i[\wl{u}]]) = \Te)
\\
\leaves(\tree(\wl{G}[\wl{H}^i[\wl{u}]])) & (\tree(\wl{G}[\wl{H}^i[\wl{u}]]) \neq \Te).
\end{cases}
\end{equation}
By~\eqref{eq:path-inf},~\eqref{eq:reldirsize}, and~\eqref{eq:leavestree},
\(\set{\tree(\wl{G}[\wl{H}^{i}[\wl{u}]]) \mid i \ge 1}\) is also infinite.

By the induction hypothesis,
there exist \(j_1\) and \(k_1\) such that
\[
\tree(\wl{G}[\wl{H}^{j_1}[\wl{u}]]) \she \tree(\wl{G}[\wl{H}^{j_1+k_1}[\wl{u}]]) \she \tree(\wl{G}[\wl{H}^{j_1+2k_1}[\wl{u}]]) \she \cdots.
\]
Hence by Lemma~\ref{lem:strictlemma}, we have
\[
\leaves(\tree(\wl{G}[\wl{H}^{j_1}[\wl{u}]])) \she 
\leaves(\tree(\wl{G}[\wl{H}^{j_1+k_1}[\wl{u}]])) \she 
\leaves(\tree(\wl{G}[\wl{H}^{j_1+2k_1}[\wl{u}]])) \she \cdots.
\]
Then by~\eqref{eq:leavestree}, we have
\begin{equation*}
\tree(\tp{G}[{\tp{H}}^{c' j_1 + d'}[\tp{u}]]) \she 
\tree(\tp{G}[{\tp{H}}^{c' (j_1+k_1) + d'}[\tp{u}]]) \she 
\tree(\tp{G}[{\tp{H}}^{c' (j_1+2k_1) + d'}[\tp{u}]]) \she \cdots.
\end{equation*}
Let \(j'_1 = c' j_1 + d'\) and \(k'_1 = c' k_1\); then
\begin{equation}
\label{eq:pathStrict}
\tree(\tp{G}[{\tp{H}}^{j'_1        }[\tp{u}]]) \she 
\tree(\tp{G}[{\tp{H}}^{j'_1 +   k'_1}[\tp{u}]]) \she 
\tree(\tp{G}[{\tp{H}}^{j'_1 + 2 k'_1}[\tp{u}]]) \she \cdots.
\end{equation}

Now, by Conjecture~\ref{conj:abstHighPeriodic},
there exist \(j_2 \ge 0\) and \(k_2 >0\) such that
\begin{equation}
\label{eq:byKruskal}
H^{j_2}[u] \he_\sty H^{j_2+k_2}[u] \he_\sty H^{j_2+2 k_2}[u] \he_\sty \cdots.
\end{equation}
Let \(j_3\) be the least \(j_3\) such that \(j_3= j'_1+i_3\,k'_1 = j_2 + m_0\) for some \(i_3\) and \(m_0\),
and \(k_3\) be the least common multiple of \(k'_1\) and \(k_2\), whence \(k_3=m_1 k'_1=m_2 k_2\) for some
 \(m_1\) and \(m_2\).
Then since the mapping \(s \mapsto \tree(G[{H}^{m_0}[s]])\) is monotonic,
from~\eqref{eq:byKruskal} we have:
\[
\tree(G[{H}^{j_3}[u]]) \he \tree(G[{H}^{j_3+k_2}[u]]) \he \tree(G[{H}^{j_3+2k_2}[u]]) \he \cdots.
\]
Since \(j_3 + i k_3 = j_3 + (i m_2) k_2\), we have
\begin{equation}
\label{eq:byKruskalAdjusted}
\tree(G[{H}^{j_3}[u]]) \he \tree(G[{H}^{j_3+k_3}[u]]) \he \tree(G[{H}^{j_3+2k_3}[u]]) \he \cdots.
\end{equation}
Also, since \(j_3 + i k_3 = j'_1 + (i_3 + i m_1)k'_1\), from~\eqref{eq:pathStrict} we have
\begin{equation}
\label{eq:strictAdjusted}
\tree(\tp{G}[\tp{H}^{j_3}[\tp{u}]]) \she \tree(\tp{G}[\tp{H}^{j_3+k_3}[\tp{u}]]) 
\she \tree(\tp{G}[\tp{H}^{j_3+2k_3}[\tp{u}]]) \she \cdots.
\end{equation}
Thus, from~\eqref{eq:reldir},~\eqref{eq:byKruskalAdjusted}, and~\eqref{eq:strictAdjusted} we obtain
\begin{equation}
\label{eq:periDir}
\tree(G[{H}^{j_3}[u]]) \she \tree(G[{H}^{j_3+k_3}[u]]) \she \tree(G[{H}^{j_3+2k_3}[u]]) \she \cdots.
\end{equation}
By applying \(\rmd\) to this sequence, and by~\eqref{eq:relbefore} and Lemma~\ref{lem:rmdirOrdMon}, we have
\begin{equation}
\label{eq:original}
\tree(C[D^{j_0 + j_3 k_0}[t]]) \she \tree(C[D^{j_0 + (j_3+k_3) k_0}[t]])  
\she \tree(C[D^{j_0 + (j_3 + 2 k_3) k_0}[t]])  \she \cdots.
\end{equation}
We define \(j = j_0 + k_0 j_3\) and \(k = k_0 k_3 \); then we obtain
\[
\tree(C[D^j[t]]) \she \tree(C[D^{j+k}[t]]) \she \tree(C[D^{j+2k}[t]]) \she \cdots.
\]

Finally, we show that
\(\size{\tree(C[D^{j+ik}[t]])} \le \expn{n}{ci+d}\) for some \(c\) and \(d\).
Since \(C\) and \(D\) are single-hole contexts, 
\(|C[D^{j+ik}[t]]|= |C|+ (j+ik)|D|+|t|\).
Let \(c= k|D|\) and \(d= |C| + j|D| + |t|\);
then \(|C[D^{j+ik}[t]]|= ci+d\).
It is well-known that, for an order-\(n\) \(\stlambda\)-term \(s\),
we have \(|\tree(s)| \le \expn{n}{|s|}\)
(see, e.g., \cite[Lemma~3]{Terui12RTA}).
Thus, we have 
\(\size{\tree(C[D^{j+ik}[t]])} \le \expn{n}{ci+d}\).
\end{proof}

The step obtaining~\eqref{eq:periDir} (the steps using Lemma~\ref{lem:strictlemma}
and obtaining~\eqref{eq:original}, resp.)
indicates why we need to require
\(\tree(C[D^{j+ik}[t]])\she \tree(C[D^{j+i'k}[t]])\) for any \(i<i'\)
rather than \(|\tree(C[D^{j+ik}[t]])| < |\tree(C[D^{j+i'k}[t]])|\)
\ (\(\tree(C[D^{j+ik}[t]]) \neq \tree(C[D^{j+i'k}[t]])\), resp.)\,
to make the induction work.

\anp
\section{Second-order Kruskal's theorem}
\label{sec:Kruskal}

In this section, we prove Conjecture~\ref{conj:abstHighKruskal} (hence also
Conjecture~\ref{conj:abstHighPeriodic}) up to order-2.
First, we extend the \contraction{} \(\tle\) on trees to a family of relations
\(\tle_{\sty}\) by using logical relation:
\begin{inparaenum}[(i)]
\item \(t_1\tle_{\T}t_2\) if \(\emptyset\pK t_1:\T\),
\(\emptyset\pK t_2:\T\), and \(\eval{t_1} \tle \eval{t_2}\).
\item \(t_1\tle_{\sty_1\to\sty_2} t_2\) if 
\(\emptyset\pK t_1:\sty_1\to\sty_2\),
\(\emptyset\pK t_2:\sty_1\to\sty_2\), and
\(t_1s_1\tle_{\sty_2}t_2s_2\) holds 
for every \(s_1,s_2\) such that \(s_1\tle_{\sty_1}s_2\).
\end{inparaenum}
We often omit the subscript \(\sty\) and just write \(\tle\) for \(\tle_\sty\).
We also write \(x_1\COL\sty_1,\ldots,x_k\COL\sty_k \models t\tle_{\sty} t'\)
if 
\([s_1/x_1,\ldots,s_k/x_k]t \tle_{\sty} [s'_1/x_1,\ldots,s'_k/x_k]t'\)
for every \(s_1,\ldots,s_k,s_1',\ldots,s_k'\) such that \(s_i\tle_{\sty_i} s_i'\).

The relation \(\tle_\sty\) is well-defined for \(\beta\eta\)-equivalence classes,
and by the abstraction lemma of logical relation, it turns out that
the relation \(\tle_{\sty}\) is a pre-order for any \(\sty\)
(see 
Appendix~\ref{sec:treeFunction} 
for these).
Note that the relation is also preserved by applications %
by the definition of the logical relation.
It remains
 to show that \(\he_\sty\) is a well quasi-order for \(\sty\) of order up to 2.

For \(\ell\)-ary terminal \(a\) and \(k \ge \ell\),
we write \(\CanoTerms{a}{k}\) for the set of terms
\[
\set{\lambda x_1.\cdots\lambda x_k.a\,x_{i_1}\,\ldots\,x_{i_\ell}\mid 
    i_1\cdots i_\ell \mbox{ is a subsequence of } 1\cdots k}.
\]
We define \(\T^0 \to \T \ \defe \ \T\) and \(\T^{n+1}\to\T \ \defe \ \T \to (\T^{n} \to \T)\).

The following lemma allows us to reduce %
\(t\,{\tle_{\sty}}\,t'\) on any order-2 type \(\sty\) to 
(finitely many instances of) that on order-0 type \(\T\).
\begin{lemma}
\label{lem:representative}
Let \(\TERMS\) be a ranked alphabet;
\(\sty\) be \((\T^{k_1}\to\T)\to \cdots \to (\T^{k_m}\to\T)\to\T\);
\(a^{j}_{i}\) be a \(j\)-ary terminal not in \(\TERMS\) for \(1 \le i \le m\) and \(0 \le j \le k_i\);
and \(t, t'\) be \(\stlambda\)-terms whose type is \(\sty\) and whose terminals are in \(\TERMS\).
Then 
\(t\tle_{\sty} t'\) if and only if
\(t\,u_1\,\ldots\,u_m \tle_{\T} t'\,u_1\,\ldots\,u_m\)
for every \(u_i\in \cup_{j\le k_i} \CanoTerms{a^j_i}{k_i}\).
\end{lemma}
\begin{proof}
The ``only if'' direction is trivial by the definition
of \(\tle_{\sty}\).
To show the opposite, assume the latter holds.
We need to show that
\(t\,s_1\,\ldots\,s_m \tle_{\T} t'\,s_1\,\ldots\,s_m\) holds for
every combination of \(s_1,\ldots,s_m\) such that \(\pK s_i:\sty_i\) for each \(i\).
Without loss of generality, 
we can assume that \(t\), \(t'\), \(s_1,\ldots,s_m\) are \(\beta\eta\) long normal forms,
and hence that
\[
\begin{aligned}
t &= \lambda f_1.\cdots\lambda f_m.t_0 &
f_1\COL\T^{k_1}\to\T,\ldots,f_m\COL\T^{k_m}\to\T &\pK t_0:\T
\\
t' &= \lambda f_1.\cdots\lambda f_m.t'_0 &
f_1\COL\T^{k_1}\to\T,\ldots,f_m\COL\T^{k_m}\to\T &\pK t'_0:\T
\\
s_i &= \lambda x_1.\ldots\lambda x_{k_i}.s_{i,0} \mspace{30mu}&
x_1\COL\T,\ldots,x_{k_i}\COL\T &\pK s_{i,0}:\T \quad\mbox{ (for each \(i\))}
\end{aligned}
\]

For each \(i \le m\), let \(\FV(s_{i,0}) = \set{x_{q(i,1)},\ldots,x_{q(i,\ell_i)}}\), and
\(u_i \in \CanoTerms{a^{\ell_i}_i}{k_i}\) be the term
\(\lambda x_1.\cdots\lambda x_{k_i}.a^{\ell_i}_i\,x_{q(i,1)}\,\cdots\,x_{q(i,\ell_i)}\).
Let \(\theta\) and \(\theta'\) be the substitutions \([u_1/f_1,\ldots,u_m/f_m]\) and 
\([s_1/f_1,\ldots,s_m/f_m]\) respectively.
It suffices to show that \(\theta t_0 \tle_{\T}\theta t_0'\) implies
\(\theta' t_0 \tle_{\T}\theta' t_0'\), which we prove
by induction on \(|t_0'|\).

By the condition \(f_1\COL\T^{k_1}\to\T,\ldots,f_m\COL\T^{k_m}\to\T\pK t_0:\T\),
\(t_0\) must be of the form \(h\,t_1\,\cdots\,t_\ell\) where \(h\) is \(f_i\)
or a terminal \(a\) in \(\TERMS\), and \(\ell\) may be \(0\).
Then we have
\[
\tree(\theta t_0) \!=\!
\begin{cases}
a\tree(\theta t_1)\cdots \tree(\theta t_\ell) &\mspace{-10mu} (h\,{=}\,a)
\\
a^{\ell_i}_i\tree(\theta t_{q(i,1)})\cdots \tree(\theta t_{q(i,\ell_i)}) &\mspace{-10mu} (h\,{=}\,f_i)
\end{cases}
\]
Similarly, \(t'_0\) must be of the form \(h'\,t'_1\,\cdots\,t'_{\ell'}\) and the corresponding
equality on \(\tree(\theta t'_0)\) holds.
By the assumption \(\theta t_0\tle_{\T}\theta t_0'\), we have 
\(\tree(\theta t_0) \he \tree(\theta t_0')\).
We perform case analysis on the rule used for deriving \(\tree(\theta t_0) \he \tree(\theta t_0')\)
(recall Definition~\ref{def:homeoEmb}).
\begin{itemize}
\item Case of the first rule:
In this case, the roots of \(\tree(\theta t_0)\) and \(\tree(\theta t'_0)\) are the same and hence
\(h=h'\) and \(\ell=\ell'\).
We further perform case analysis on \(h\).
\begin{asparaitem}
\item Case \(h=a\): For \(1 \le j \le \ell\), since \(\tree(\theta t_j) \he \tree(\theta t'_j)\),
by induction hypothesis, we have \(\theta' t_j \tle_\T \theta' t'_j\).
Hence \(\theta' t_0 \tle_\T \theta' t'_0\).
\item Case \(h=f_i\): For \(1 \le j \le \ell_i\), since \(\tree(\theta t_{q(i,j)}) \he \tree(\theta t'_{q(i,j)})\), 
by induction hypothesis, we have \(\theta' t_{q(i,j)} \tle_\T \theta' t'_{q(i,j)}\).
Hence, \([\theta' t_{q(i,j)}/x_{q(i,j)}]_{j \le \ell_i}s_{i,0} \tle_\T [\theta' t'_{q(i,j)}/x_{q(i,j)}]_{j \le \ell_i}s_{i,0}\).
By the definition of \(q(i,j)\), \(\theta' t_0 \red [\theta' t_j/x_j]_{j \le k_i}s_{i,0}
 = [\theta' t_{q(i,j)}/x_{q(i,j)}]_{j \le \ell_i}s_{i,0}\),
and similarly, \(\theta' t'_0 \red [\theta' t'_{q(i,j)}/x_{q(i,j)}]_{j \le \ell_i}s_{i,0}\);
hence we have \(\theta' t_0 \tle_\T \theta' t'_0\).
\end{asparaitem}
\item Case of the second rule:
We further perform case analysis on \(h'\).
\begin{asparaitem}
\item Case \(h'=a\): We have \(\tree(\theta t_0) \he \tree(\theta t'_p)\) for some \(1 \le p \le \ell'\).
Hence by induction hypothesis, we have \(\theta' t_0 \tle_{\T} \theta' t'_p\),
and then \(\theta' t_0 \tle_{\T} \theta' t'_0\).
\item Case \(h'=f_i\): We have \(\tree(\theta t_0) \he \tree(\theta t'_{q(i,p)})\) for some \(1 \le p \le \ell_i\).
Hence by induction hypothesis, we have \(\theta' t_0 \tle_{\T} \theta' t'_{q(i,p)}\).
Also, by the definition of \(q(i,p)\), \(x_{q(i,p)}\) occurs in \(s_{i,0}\).
Since \(s_{i,0}\) is a \(\beta\eta\) long normal form of order-0, the order-0 variable \(x_{q(i,p)}\)
occurs as a leaf of \(s_{i,0}\); 
hence \(\tree(\theta' t'_{q(i,p)}) \he [\tree(\theta' t'_{q(i,j)})/x_{q(i,j)}]_{j \le \ell_i}s_{i,0}\).
Therefore \(\theta' t_0 \tle_\T [\theta' t'_{q(i,j)}/x_{q(i,j)}]_{j \le \ell_i}s_{i,0}\).
Since \(\theta' t'_0 \red [\theta' t'_{q(i,j)}/x_{q(i,j)}]_{j \le \ell_i}s_{i,0}\), we have
\(\theta' t_0 \tle_{\T} \theta' t'_0\).
\end{asparaitem}
\end{itemize}
\end{proof}

As a corollary, we obtain a second-order version of Kruskal's tree theorem.
\begin{theorem}
\label{thm:order2-kruskal}
Let \(\TERMS\) be a ranked alphabet,
\(\sty\) be an at most order-2 type, and
\(t_0,t_1,t_2,\ldots\) be an infinite sequence of 
\(\stlambda\)-terms whose type is \(\sty\) and whose terminals are in \(\TERMS\).
Then, there exist \(i<j\) such that \(t_i \tle_{\sty} t_j\).
\end{theorem}
\begin{proof}
Since \(\sty\) is at most order-2,
it must be of the form \((\T^{k_1}\to\T)\to \cdots \to (\T^{k_m}\to\T)\to\T\).
Let \(a^{j}_{i}\) be a \(j\)-ary terminal not in \(\TERMS\) for \(1 \le i \le m\) and \(0 \le j \le k_i\);
\((\cup_{j \le k_1} \CanoTerms{a^{j}_1}{k_1}) \times \cdots \times (\cup_{j \le k_m} \CanoTerms{a^j_m}{k_m})\)
be \(\set{(u_{1,1},\ldots,u_{1,m}), \ldots, (u_{p,1},\ldots,u_{p,m})}\);
\(b\) be a \(p\)-ary terminal not in \(\TERMS\cup\set{a^j_i \mid 1 \le i \le m,\, 0 \le j \le k_i}\);
and \(s_i\) be the term \(b\,(t_i\,u_{1,1}\,\cdots\,u_{1,m})\,\cdots\, (t_i\,u_{p,1}\,\cdots\,u_{p,m})\)
for each \(i\in\set{0,1,2,\ldots}\).
Since the set of terminals in \(s_0,s_1,s_2,\dots\) is finite,
by Kruskal's tree theorem, there exist \(i,j\) such that 
\(s_i \tle_{\T} s_j\) and \(i<j\).
Since \(b\) occurs just at the root of \(s_k\) for each \(k\),
\(s_i\tle_{\T}s_j\) implies \(t_i\,u_{k,1}\,\cdots\,u_{k,m}
\tle_{\T} t_j\,u_{k,1}\,\cdots\,u_{k,m} \) for every \(k\in\set{1,\ldots,p}\).
Thus, by Lemma~\ref{lem:representative}, we have \(t_i\tle_{\sty}t_j\) as required.
\end{proof}

\section{Related Work}
\label{sec:related}

As mentioned in Section~\ref{sec:intro}, to our knowledge,
pumping lemmas for higher-order word languages have been established only up to
order-2~\cite{Hayashi73}, whereas we have proved (unconditionally) a pumping lemma
for order-2 tree languages and order-3 word languages.
Hayashi's pumping lemma for indexed languages (i.e., order-2 word languages)
is already quite complex, and it is unclear how to generalize it to arbitrary orders.
In contrast, our proof of a pumping lemma works for arbitrary orders, although 
it relies on the conjecture on higher-order Kruskal's tree theorem. 
Parys~\cite{Parys12STACS} and Kobayashi~\cite{Kobayashi13LICS} studied pumping lemmas
for collapsible pushdown automata and higher-order recursion schemes respectively. Unfortunately,
they are not applicable to word/tree \emph{languages} generated by (non-deterministic)
grammars. 

As also mentioned in Section~\ref{sec:intro}, the strictness of
hierarchy of higher-order word languages has already been shown by using a complexity 
argument~\cite{Engelfriet91,KartzowPC}.
We can use our pumping lemma (if the conjecture is discharged) to obtain a simple alternative proof of
the strictness, using the language \(\set{ a^{\expn{n}{k}} \mid k\geq 0}\) as a witness
of the separation between the classes of order-\((n+1)\) word languages and order-\(n\) word languages.
In fact, the pumping lemma would imply that there is no order-\(n\) grammar that generates
\(\set{ a^{\expn{n}{k}} \mid k\geq 0}\), whereas an order-\((n+1)\) grammar that generates the same 
language can be easily constructed.

We are not aware of studies of the higher-order version of Kruskal's tree theorem (Conjecture~\ref{conj:abstHighKruskal})
or the periodicity of tree functions expressed by the simply-typed \(\lambda\)-calculus
(Conjecture~\ref{conj:abstHighPeriodic}), which seem to be of independent
interest. Zaionc~\cite{Zaionc87,Zaionc88} characterized the class of 
(first-order) word/tree functions definable in the simply-typed \(\lambda\)-calculus. To obtain
higher-order Kruskal's tree theorem, we may need some characterization of
\emph{higher-order} definable tree functions instead.

We have heavily used the results of Parys' work~\cite{Parys16ITRS} and 
our own previous work~\cite{asada_et_al:LIPIcs:2016:6246}, which both use
intersection types for studying properties of higher-order languages.
Other uses of intersection types in studying higher-order grammars/languages are found in
\cite{Kobayashi09POPL,KO09LICS,Parys14,Kobayashi13LICS,Parys16diagonality,KMS13HOSC,KIT14FOSSACS}.

\section{Conclusion}
\label{sec:conc}

We have proved a pumping lemma for higher-order languages of arbitrary orders,
modulo the assumption that a higher-order version of Kruskal's  tree theorem holds.
We have also proved the assumption indeed holds for the second-order case, yielding
a pumping lemma for order-2 tree languages and order-3 word languages.
Proving (or disproving) the higher-order Kruskal's  tree theorem is left for future work.

\subsection*{Acknowledgments}

We would like to thank Pawel Parys for discussions on his type system,
and anonymous referees for useful comments.
This work was supported by JSPS Kakenhi %
15H05706.

\anp

\bibliographystyle{plain}

\appendix
\asd{meta variables:
\begin{itemize}
\item term: \(s,t,u,v\) (often \(s,t\) correspond to \(v,u\))
\item context: \(C,D,G,H\)
\item evaluation context: \(E\)
\item normal form: \(N\)
\item variable: \(x,y,z,f,g\)
\item tree: \(\pi\)
\item word: \(w\)
\item grammar: \(\G\)
\item rule of grammar: \(r\)
\item non-terminal: \(A,B\)
\item terminal/constants: \(a\)
\item \(F\), \(M\) are used in Parys's type system
\item derivation: \(\dt\)
\item order: \(n\)
\item natural numbers: \(m,k,\ell,i,j,\dots\) %
\end{itemize}
}

\anp
\section{More Details on Section~\ref{sec:Parys}}
\label{sec:parys-details}

We provide more details on how Parys' type system~\cite{Parys16ITRS} can be modified to obtain
Lemmas~\ref{lem:parys3} and \ref{lem:parys5}.
Some familiarity (especially, intuitions on flags and markers) with Parys' type system
is required to understand this section.
In Section~\ref{sec:parys-type},
 we first present a variant of Parys' type system 
and state key lemmas. We then prove Lemma~\ref{lem:parys3} using the lemmas.
After preparing some basic lemmas about the type system~\ref{sec:basic}
we prove the key lemmas in Sections~\ref{sec:expansion} and \ref{sec:reduction}.
In Section~\ref{sec:parys-trans2}, we modify the type system to show how to extract a triple
\((G,H,u)\) that satisfies the requirement of Lemma~\ref{lem:parys5}.
\subsection{A Variant of Parys' Type System and its Key Properties}
\label{sec:parys-type}

Below we fix a grammar \(\GRAM\).

The set of types, ranged over by \(\pty\), is defined by:
\[
\begin{array}{l}
\pty \mbox{ (types) }::= (F, M, \prty) \qquad \prty \mbox{ (raw types) }::=\T \mid \pity\to \prty\\
\pity \mbox{ (intersection types) }::= \set{\pty_1,\ldots,\pty_k}
\end{array}
\]
Here, \(F\) and \(M\) range over the finite powerset of the set of natural numbers,
and \(F\) and \(M\) must be disjoint.
Intuitively, \((F, M,\T)\) is the type of trees which contain flags of orders in \(F\)
and markers of orders in \(M\). The type \((F, M, \pity\to \prty)\) describes a function term which,
when viewed as a function, takes an argument that has all the types in \(\pity\) and returns
a value of type \(\prty\), and when viewed as a term, contains  flags of orders in \(F\)
and markers of orders in \(M\). Thus, each type expresses the ``dual'' views of a term,
both as a function and a term.
In order for \(\pity\to \prty\) to be well-formed,
for each \((F,M,\prty')\in \pity \), it must be the case that \(F,M\subseteq \set{i
\mid 0\leq i < \order(\pity\to\prty)}\).
For a set of natural numbers \(S\) and a natural number \(n\), we define 
\(S\restrict{<n} \defe \set{m \in S \mid m < n}\) and
\(S\restrict{\ge n} \defe \set{m \in S \mid m \ge n}\).
We assume some total order \(<\) on types.

A type environment is a set of bindings of the form \(x\COL \pty\), which may contain 
more than one binding on the same variable. 
We use the meta-variable \(\PTE\) for a type environment.
We require that all the marker sets occurring in \(\PTE\) are mutually disjoint, i.e.,
if \(x\COL(F,M,\prty), x'\COL(F',M',\prty')\in \PTE\), then
either \(M\cap M'=\emptyset\) or \(x\COL(F,M,\prty)=x'\COL(F',M',\prty')\).
We write \(\Markers(\PTE)\) for the set of markers occurring in \(\PTE\), i.e.,
\(\uplus \set{M \mid x\COL(F,M,\prty)\in \PTE}\).
Also, we write \(\Markers(\pity)\) for \(\Markers(x\COL \pity)\).
We write \(\PTE\subPTE \PTE'\) if \(\PTE'\subseteq \PTE\) and
\(M=\emptyset\) for each \((x\COL(F,M,\prty))\in\PTE\setminus\PTE'\).
We define \(\PTE_1 \PTEcup \PTE_2\) as \(\PTE_1\cup\PTE_2\) only if
\(M=\emptyset\) for any \((x\COL (F,M,\prty)) \in \PTE_1\cap\PTE_2\).

The operation \(\Comp_{n}(\set{(F_1,c_1),\ldots,(F_k,c_k)},M)=(F,c)\),
where \(c, c_i\) are natural numbers, 
is defined by:
 \[
 \begin{array}{l}
   f'_0 = 0 \qquad f'_{\ell} = \left\{ \begin{array}{ll} f_{\ell-1} & \mbox{if \(\ell-1\in M\)}\\
        0 & \mbox{otherwise}
    \end{array}\right. \mbox{for each \(\ell\in\set{1,\ldots,n}\)}\\
   f_\ell = f'_\ell + |\set{i \mid \ell\in F_i}| \quad \mbox{for each \(\ell\in\set{0,\ldots,n-1}\)}\\
   F = \set{\ell\in\set{0,\ldots,n-1} \mid f_\ell>0}\setminus M\qquad
   c = f'_{n} + c_1+\cdots + c_k
 \end{array}
 \]
Here, (although we use a set notation) please note that the first argument of 
\(\Comp_n\) is a \emph{multiset}. In fact, \(\Comp_1(\set{(\set{0},0),(\set{0},0)},\set{0})
= 2\), but \(\Comp_1(\set{(\set{0},0)},\set{0})=1\).
We write \(\mcup\) for the union of multisets: \((X \mcup Y)(x) \defe X(x)+Y(x)\).
Note that, by unfolding the definition we have:
\begin{align*}
F 
&= \set{\ell <n \mid
\exists j \in \cup_{i \le k}F_i . \ 
\ell = \min\set{\ell' \ge j \mid \ell' \notin M }\,}
\\
f'_n &=
\big|\set{i \mid n-1 \in F_i}\big|
+ \cdots +
\big|\set{i \mid n-j \in F_i}\big|
=
\textstyle\sum_{i \le k} \big| \set{n-1,\dots,n-j} \cap F_i  \big|
\\&
\phantom{f'_n =}\ \text{where \(j\) is such that: }0\le j\le n;\ n-1,\dots,n-j \in M;\ n-(j+1) \notin M.
\end{align*}

A type judgment (or more precisely, a type-based transformation judgment)
is of the form \(\PTE\PM t:\pty\hascost{c}\tr s\) where 
\(n>0\),
\(\PTE\) is a type environment that contains types of orders at most \(n-1\),
\(\eorder(t) \le n\), \(c\) is a natural number called a (\emph{flag}) \emph{counter}, and \(s\) is a \(\stlambda\)-term. 
The flags and markers in the judgment must be at most \(n-1\).
We present typing rules below. 
For some technical convenience, 
we made some changes to Parys' original type system~\cite{Parys16ITRS}, which are summarized below.
\begin{itemize}
\item We have added the output \(s\) of the transformation. Intuitively, it simulates the behavior of \(t\).
\item We allow markers to be placed at any node of a derivation, not just at leaves (see the rule \rname{PTr-Mark} below).
\item We added a rule (\rname{PTr-NT} below) for explicitly unfolding non-terminals. 
Here we assume that a grammar is expressed as a set of equations of the form
\(\set{A_1=t_1,\ldots,A_\ell=t_\ell}\) where \(A_1,\ldots,A_\ell\) are distinct from each other. 
A set of rewriting rules \(\set{A\,x_1\,\ldots\,x_k\to t_i\mid i\in \set{1,\ldots,p}}\)
is expressed as the equation
\(A=\lambda x_1.\cdots\lambda x_k.(t_1+\cdots + t_p)\).
In the original formulation of Parys, infinite \(\lambda\)-terms (represented as regular trees) were considered instead.
\item We consider judgments \(\PTE\PM t:\pty\hascost{c}\) only when \(n>0\).
(In the original type system of Parys~\cite{Parys16ITRS}, \(n\) may be \(0\).)
\item In rule \rname{PTr-App}, we allow only a single derivation for each argument type.
\end{itemize}

\infrule[PTr-Weak]{\PTE \PM \term:\pty\hascost{c} \tr s\andalso M=\emptyset}
   {\PTE, x\COL(F,M,\prty)\PM \term:\pty\hascost{c} \tr s}

 \infrule[PTr-Mark]{\PTE\PM \term:(F',M',\prty)\hascost{c'}\tr s\\
M\subseteq \set{j \mid \eorder(\term)\leq j<n}\andalso   \Comp_{n}(\set{(F',c')},M\uplus M')= (F,c)
}
        {\PTE\PM\term:(F, M\uplus M',\prty )\hascost{c}\tr s}

         \infrule[PTr-Var]{}{x\COL \pty\PM x\COL \pty\hascost{0}\tr x_{\pty}}
         
         \infrule[PTr-Choice]{\PTE\PM \term_i:\pty\hascost{c}\tr s_i \qquad i=1\lor i=2}
                 {\PTE\PM \term_1+\term_2\COL \pty\hascost{c}\tr s_i}

\infrule[PTr-Abs]{\PTE, x\COL \pity\PM \term:(F,M,\prty) \hascost{c}\tr s
}
  {\PTE \PM \lambda x.\term: (F, %
    M\setminus \Markers(\pity),\pity\to \prty)\hascost{c}\tr
     \lambda \seq{x}_\pity.s
      }
Here, \(\lambda \seq{x}_\pity.s\) stands for \(\lambda x_{\pty_1}.\cdots \lambda x_{\pty_k}.s\)
when \(\pity=\set{\pty_1,\ldots,\pty_k}\) with \( \pty_1< \cdots <\pty_k\).

\infrule[PTr-App]{\eorder(\term_0)=\ell\\
\PTE_0\PM \term_0:(F_0,M_0,\set{(F_1,M_1,\prty_1),\ldots,
(F_k,M_k,\prty_k)}\to \prty)\hascost{c_0}\tr s_0\\
  \PTE_{i}\PM \term_1:(F_{i}',M_{i}',\prty_i)\hascost{c_{i}}\tr s_{i}
\andalso F'_{i}\restrict{<\ell}=F_i\andalso M'_{i}\restrict{<\ell}=M_i
\mbox{ for each $i\in\set{1,\ldots,k}$}\\
  M = M_0\uplus (\biguplus_{i\in\set{1,\ldots,k}} M'_{i})\\
  \Comp_{n}(\set{(F_0,c_0)}\mcup \set{(F'_{i}\restrict{\geq \ell},c_{i})\mid i\in\set{1,\ldots,k}}, M)
   = (F, c)\\
(F_1,M_1,\prty_1)<\cdots< (F_k,M_k,\prty_k)
}
{\PTE_0\PTEcup (\textstyle\sum_{i\in\set{1,\ldots,k}} \PTE_{i}) \PM \term_0\term_1:
  (F, M, \prty)\hascost{c}\tr s_0\, s_{1}\,\cdots\,s_{k}}

\infrule[PTr-Const]{\arity(a)=k \\ %
  \PTE_i \PM \term_i:(F_i,M_i,\T)\hascost{c_i}\tr s_i\mbox{ for each $i\in\set{1,\ldots,k}$}\\
  M = M_1\uplus\cdots\uplus M_k\\
  \Comp_{n}(\set{(\set{0},0),(F_1,c_1),\ldots,(F_k,c_k)},M)=(F,c)
  }
        {\PTE_1\PTEcup\cdots \PTEcup \PTE_k \PM a\,\term_1\,\cdots\,\term_k:
          (F,M,\T)\hascost c\tr a\,s_1\,\cdots\,s_k}

\infrule[PTr-NT]{\PTE \PM t:\pty\hascost{c}\tr s\andalso A=t\in \GRAM}
  {\PTE \PM A:\pty\hascost{c}\tr s}

We write \(\dt \pdt \PTE\PM t:\pty\hascost c\tr s\) if
\(\dt\) is a derivation tree for \(\PTE\PM t:\pty\hascost c\tr s\).

Next we state two key properties of the type system.
The following lemma states that if \(\Lang(\GRAM)\) is infinite, there exists a pumpable derivation tree
in which a part of the derivation can be repeated arbitrarily many times.
\begin{lemma}[existence of a pumpable derivation]
\label{lem:pumpable-derivation}
Let \(\GRAM\) be an order-\(n\) tree grammar and \(S\) be its start symbol.
If \(\Lang(\GRAM)\) is infinite, then 
there exists a derivation 
for \(\emptyset \PM S:(\emptyset,\set{0,\ldots,n-1},\T)\hascost c_1+c_2+c_3\tr C[D[s]]\) 
with \(c_1,c_2>0\), in which for some \(A\), \(\PTE\), and \(\pty\),
\(\PTE \PM A:\pty\hascost{c_1+c_2}\tr D[s]\) is derived from \(\PTE \PM A:\pty\hascost{c_1}\tr s\).
Furthermore, the contexts \(C\) and \(D\) are linear.
\end{lemma}

The following lemma states that any simply-typed \(\lambda\)-term \(s\) 
obtained by the transformation generates a member of \(\Lang(\GRAM)\), and 
its size is bounded below by \(c\).
\begin{lemma}[soundness of transformation]
\label{lem:soundness-tr}
Let \(\GRAM\) be an order-\(n\) tree grammar and \(S\) be its start symbol.
If \(\emptyset \PM S: (\emptyset,\set{0,\ldots,n-1},\T) \hascost c\tr s\), then 
\(s\) is a \(\stlambda\)-term of order at most \(n\), and \(\tree(s)\in \Lang(\GRAM)\), with \(c \leq |\tree(s)|\). 
\end{lemma}

We will prove Lemmas~\ref{lem:pumpable-derivation} and \ref{lem:soundness-tr}
above in Sections~\ref{sec:expansion} and \ref{sec:reduction} respectively.
Using the lemmas above, we can prove Lemma~\ref{lem:parys3}.
\begin{proof}[Proof of Lemma~\ref{lem:parys3}.]
Suppose that \(\Lang(\GRAM)\) is infinite.
By Lemma~\ref{lem:pumpable-derivation}, we have a pumpable derivation for
\(\emptyset \PM S:(\emptyset,\set{0,\ldots,n-1},\T)\hascost c_1+c_2+c_3\tr C[D[s]]\) 
with \(c_1,c_2>0\), in which
\(\PTE \PM A:\pty\hascost{c_1+c_2}\tr D[s]\) is derived from \(\PTE \PM A:\pty\hascost{c_1}\tr s\),
and contexts \(C,D\) are linear. The orders of \(s\), \(C\), and \(D\) are at most \(n\);
by inserting a dummy subterm, we can assume that the orders of them are \(n\).
By repeating the subderivation from
\(\PTE \PM A:\pty\hascost{c_1}\tr s\) to \(\PTE \PM A:\pty\hascost{c_1+c_2}\tr D[s]\),
we obtain a derivation for 
\[\PTE\PM S:(\emptyset,\set{0,\ldots,n-1},\T)\hascost c_1+kc_2+c_3\tr C[D^k[s]]\]
for any \(k\geq 0\).
By Lemma~\ref{lem:soundness-tr}, we have 
\(\set{\tree(C[D^k[t]]) \mid k \ge 1} \subseteq \Lang(\G)\).
Let \((\ell_k)_k\) be a strictly increasing sequence.
Then, the set \(\set{c_1+\ell_k c_2+c_3\mid k\geq 1}\) is infinite. Thus, by the condition
\(|\tree(C[D^{\ell_k}[t]])|\geq c_1+\ell_k c_2+c_3\), 
\(\set{|\tree(C[D^{\ell_k}[t]])| \mid k \ge 1}\) must be infinite, which also implies that 
\(\set{\tree(C[D^{\ell_k}[t]]) \mid k \ge 1}\) is infinite.
\end{proof}

\anp
\subsection{Basic Definitions and Lemmas}
\label{sec:basic}
Here we prepare some definitions and lemmas that are 
 commonly used in Sections~\ref{sec:expansion} and \ref{sec:reduction}. 

We first define a refined notion of reductions. For proving 
the key lemmas (Lemmas~\ref{lem:pumpable-derivation} and~\ref{lem:soundness-tr}),
we consider a specific reduction sequence in which redexes of higher-order
are reduced first. Thus we consider a restricted version \(\ored{n}\) of 
the reduction relation, where only redexes of order-\(n\) can be reduced
(in addition to unfolding of non-terminals).

We define \emph{order-\(n\) reduction}, written by \(\ored{n}\), by the following rules.
\begin{align*}&
C[(\lambda x.t) s] \ored{n} C[[s/x]t] \quad\text{if}\quad (\eorder(\lambda x.t) =n)
\\&
C[A] \ored{n} C[t] \quad\text{if}\quad ((A=t) \in \G).
\end{align*}

Also we define a reduction on the nondeterministic choice, written \(\cred\), as the following:
\[
C[t_1 + t_2] \cred C[t_i] \quad(i=1,2).
\]

The following lemma states that any reduction sequence for generating a tree can be normalized,
so that reductions are applied in a decreasing order.
\begin{lemma}
\label{lem:order-reduction}
Let \(\GRAM\) be an order-\(n\) grammar and \(S\) be its start symbol.
If \(\pi\in \Lang(\GRAM)\), then 
\[
S \oreds{n} t_{n-1} \oreds{n-1} \cdots \oreds{1} t_0 \creds \pi.
\]
\end{lemma}
\begin{proof}
Given a reduction sequence \(S\reds \pi\), we can move any unfolding of a non-terminal
to the left, and any reduction of choice to the right, and obtain a reduction sequence
\[
S \oreds{n} t_n \reds t_0 \creds \pi
\]
in which only \(\beta\)-reductions are applied in \(t_n\reds t_0\) 
(recall that we use only the non-deterministic choice of the ground type,
by the assumption for grammars in Section~\ref{sec:parys-type}).
Note that if the largest order of redexes in \(t\) is \(k\), 
reducing the rightmost, innermost order-\(k\) redex neither introduces any new redex of order higher than \(k\),
nor copies any redex of order \(k\). Thus,
we can obtain a normalizing\footnote{Here, we consider non-terminals and the choice operator as constants.} 
sequence \(t_n \oreds{n}t_{n-1} \oreds{n-1} \cdots \oreds{1} t'_0\),
where \(t'_0\) does not contain any \(\beta\)-redex. By Church-Rosser theorem, \(t_0 = t'_0\).
Thus, we have 
\[
S \oreds{n} t_{n-1} \oreds{n-1} \cdots \oreds{1} t_0 \creds \pi
\]
as required.
\end{proof}

\begin{lemma}
\label{lem:PTrFV}
If \(x_1 \COL \prty_1 ,\dots, x_k \COL \prty_k \PM t:(F,M,\prty) \hascost{c}\tr s\),
then \(\FV(s) \subseteq \set{(x_1)_{\prty_1}, \dots, (x_k)_{\prty_k}}\).
\end{lemma}
\begin{proof}
By straightforward induction on
a derivation tree of \(x_1 \COL \prty_1 ,\dots, x_k \COL \prty_k \PM t:(F,M,\prty) \hascost{c}\tr s\)
and by case analysis on the last rule of the derivation.
\end{proof}

We say \((F,M,\pity_1 \to \dots \to \pity_k \to\T)\) is \emph{\(n\)-clear}
if \(n\notin M\cup(\textstyle\bigcup_{i \le k}\Markers(\pity_i))\).
\begin{lemma}
\label{lem:markerAux}
For \(\PTE\PMm{n} t:\pty \hascost{c} \tr s\) with \(n>0\),
if \(\pty\) is \((n-1)\)-clear, then \(c=0\).
\end{lemma}
\begin{proof}
The proof proceeds by induction on the derivation
\(\PTE\PMm{n} t:\pty \hascost{c} \tr s\)
and its case analysis.
The cases of \rname{PTr-Var}, \rname{PTr-Choice}, \rname{PTr-Abs}, and \rname{PTr-NT} are clear.
The cases of \rname{PTr-Mark} and \rname{PTr-Const} are similar to (and easier than) the case of \rname{PTr-App}.

In the case of \rname{PTr-App},
we have
\infrule{\eorder(\term_0)=\ell\\
\PTE_0\PM \term_0:(F_0,M_0,\set{(F_1,M_1,\prty_1),\ldots,
(F_k,M_k,\prty_k)}\to \prty)\hascost{c_0}\tr s_0\\
  \PTE_{i}\PM \term_1:(F_{i}',M_{i}',\prty_i)\hascost{c_{i}}\tr s_{i}
\andalso F'_{i}\restrict{<\ell}=F_i\andalso M'_{i}\restrict{<\ell}=M_i
\mbox{ for each $i\in\set{1,\ldots,k}$}\\
  M = M_0\uplus (\biguplus_{i\in\set{1,\ldots,k}} M'_{i})\\
  \Comp_{n}(\set{(F_0,c_0)}\mcup \set{(F'_{i}\restrict{\geq \ell},c_{i})\mid i\in\set{1,\ldots,k}}, M)
   = (F, c)\\
(F_1,M_1,\prty_1)<\cdots< (F_k,M_k,\prty_k)
}
{\PTE_0\PTEcup (\textstyle\sum_{i\in\set{1,\ldots,k}} \PTE_{i}) \PM \term_0\term_1:
  (F, M, \prty)\hascost{c}\tr s_0\, s_{1}\,\cdots\,s_{k}}
It is clear from the assumption that 
\((F_0,M_0,\set{(F_1,M_1,\prty_1),\ldots,(F_k,M_k,\prty_k)}\to \prty)\)
is \((n-1)\)-clear.
Also, each \((F'_{i},M'_{i},\prty_i)\) is \((n-1)\)-clear 
since \(\eorder(t_1) \le n-1\) and due to the well-formedness of types.
Hence, by induction hypothesis we have \(c_0 = 0\) and \(c_{i} = 0\) for any \(i\).
It is clear that, in general,
\(\Comp_{n}(\set{(F_1,0),\dots,(F_k,0)}, M)
   = (F, 0)\)
if \(n-1 \notin M\).
Thus \(c=0\), as required.
\end{proof}

\begin{lemma}
  \label{lem:marker}
  For \(\PTE\PMm{n} t:(F,M,\prty)\hascost c\tr s\) with \(\eorder(t) \le n-1\),
if \(c>0\), then \(n-1\in M\).
  \end{lemma}
\begin{proof}
This follows from Lemma~\ref{lem:markerAux} (and the well-formedness of types).
\end{proof}

The following lemma corresponds to~\cite[Lemma~24]{Parys16ITRS}
\begin{lemma}
\label{lem:markContextType}
If \(\PTE \PM t: (F,M,\prty) \hascost c\tr s\) then \(\Markers(\PTE) \subseteq M\).
\end{lemma}
\begin{proof}
By straightforward induction on the derivation of \(\PTE \PM t: (F,M,\prty) \hascost c\tr s\)
and case analysis of the last rule of the derivation.
\end{proof}

\begin{lemma}
[\mbox{\cite[Lemma~26]{Parys16ITRS}}]
\label{lem:empDisjComp}
For \(F_0, M' \subseteq \set{0,\dots,n-1}\),
if
\[
\Comp_{n}(\set{(F_0,c_0)}\mcup \set{(\emptyset,c_{i})\mid i\in\set{1,\ldots,k}}, M')
   = (F', c')
\]
and \(F_0 \cap M' = \eset\)
then
\[
F' = F_0
\qquad
c' = c_0 + \textstyle\sum_{i\in\set{1,\ldots,k}} c_{i}.
\]
\end{lemma}

\anp %
\subsection{Proof of Lemma~\ref{lem:pumpable-derivation}}
\label{sec:expansion}

To prove Lemma~\ref{lem:pumpable-derivation}, we first introduce a linear type system
for typing the output of the transformation, in Section~\ref{sec:linear-type}. 
The linear type system is required to
guarantee that the contexts \(C\) and \(D\) are linear.
We will then prove, in Section~\ref{sec:completeness-tr},
a certain completeness property of the transformation relation,
that if \(\pi\in \GRAM\), then there exists a derivation
\(\emptyset \PM S:(\emptyset,\set{0,\ldots,n-1},\T)\hascost c\tr s\)
such that \(\expn{n}{c}\geq |\pi|\), and \(s\) is well-typed in the linear type system.
We will then prove Lemma~\ref{lem:pumpable-derivation} in Section~\ref{sec:pumpability}.

\subsubsection{Linear Type System and Translation}
\label{sec:linear-type}

The syntax of \emph{linear/non-linear types} is given by:
\[
\begin{array}{l}
\lty ::= \lrty^m \\
\lrty ::= \T \mid \lty_1\to \lty_2\\
m ::= 1 \mid \omega
\end{array}
\]
When \(\lty=\lrty^m\), we write \(\mult(\lty)\) for \(m\) and call it the \emph{multiplicity} of \(\lty\). 
Intuitively, \(\lrty^1\) (\(\lrty^\omega\), resp.) represents the type of values that can be used
once (arbitrarily many times, resp.).
We call \(\lty\) \emph{linear} if \(\mult(\lty)=1\) and \emph{non-linear} if \(\mult(\lty)=\omega\).
We require the well-formedness condition that in every function type
 \((\lty_1\to\lty_2)^m\), if \(m\) or \(\mult(\lty_1)\) is \(1\), then so is \(\mult(\lty_2)\),
and exclude out types containing ill-formed types below.

We define \emph{partial} operations on multiplicities, types, and type environments by:
\[ 
\begin{array}{l}
\omega + \omega = \omega\\
1 \cdot m = m \qquad \omega \cdot \omega = \omega\\
\lrty^{m_1}+\lrty^{m_2} = \lrty^{m_1+m_2}\\
m' \lrty^{m} = \lrty^{m' \cdot m}\\
(\Gamma_0+\Gamma_1)(x) =
  \left\{
 \begin{array}{ll}
   \Gamma_0(x)+\Gamma_1(x) & \mbox{if \(x\in\dom(\Gamma_0)\cap\dom(\Gamma_1)\)}\\
   \Gamma_i(x) & \mbox{if \(x\in\dom(\Gamma_i)\setminus\dom(\Gamma_{1-i})\)}\\
\end{array}\right.
\\
(m\Gamma)(x) = m(\Gamma(x)) \quad \mbox{if \(x\in\dom(\Gamma)\)}
\end{array}
\]
Note that \(\Gamma_0+\Gamma_1\) is defined if and only if
\(\lrty_0=\lrty_1\) and \(m_0=m_1=\omega\) whenever \(x:\lrty_0^{m_0} \in \Gamma_0\) and \(x:\lrty_1^{m_1} \in
\Gamma_1\); in this case, \(\Gamma_0+\Gamma_1 = \Gamma_0 \cup \Gamma_1\).
Also, \(m\Gamma\) is defined if and only if either \(m=\omega\) and \(\Gamma\) does not contain a linear
type or \(m=1\); in this case, \(m\Gamma=\Gamma\).

The linear type judgment relation \(\LTE \pLin s:\lty\) is defined by the following typing rules;
note that rules \rname{LT-Const}, \rname{LT-Abs}, and \rname{LT-App} are applicable only when the above partial operations on environments occurring in the rules are defined.

\infrule[LT-Weak]{\Gamma\pLin s:\lty}
  {\Gamma, x\COL\lrty^\omega\pLin s:\lty}

\infrule[LT-Var]{}
 {x:\lty \pLin x:\lty}

\infrule[LT-Const]
 {  \arity(a)=k \\
\Gamma_i\pLin s_i:\T^{m} \quad \mbox{ for each \(i\in\set{1,\ldots,k}\)}}
 {\Gamma_1+\cdots+\Gamma_k \pLin a\,s_1\,\cdots\,s_k:\T^m}

\infrule[LT-Abs]{\Gamma,x:\lty_1\pLin s:\lty_2}
  {m\Gamma \pLin \lambda x.s:(\lty_1\to\lty_2)^m}

\infrule[LT-App]{\Gamma_0\pLin s_0:(\lty_1\to\lty)^m\andalso \Gamma_1\pLin s_1:\lty_1}
  {\Gamma_0+\Gamma_1 \pLin s_0s_1:\lty}

\infrule[LT-Dereliction]{\Gamma\pLin s:\lrty^\omega}
   {\Gamma\pLin s:\lrty^1}

\begin{lemma}
\label{lem:variable}
If \(\LTE \pLin s:\lty\) and \(x\) occurs in \(s\), then \(x\COL\lty'\in \LTE\) for some \(\lty'\).
\end{lemma}
\begin{proof}
Straightforward induction on the derivation of \(\LTE\pLin s:\lty\).
\end{proof}
\begin{lemma}
\label{lem:used-once}
If \(\LTE, x\COL\lrty^1 \pLin s:\lty\), then \(s\) contains exactly one occurrence of \(x\).
\end{lemma}
\begin{proof}
This follows by straightforward induction on the derivation of \(\LTE, x\COL\lrty^1 \pLin s:\lty\).
We discuss only the case where the last rule is \rname{LT-Abs}.
In that case, \(s=\lambda y.s'\) and \(\lty=(\lty_1\to \lty_2)^m\), where
 \(m\) must be \(1\) (because \(\omega\cdot 1\)
is undefined). Thus, we have \(\LTE, x\COL\lrty^1, y\COL\lty_1\pLin s':\lty_2\).
By the induction hypothesis, \(x\) occurs exactly once in \(s'\), hence also in \(s\).
\end{proof}
\begin{lemma}
\label{lem:env-for-omega}
If \(\LTE \pLin s:\lrty^\omega\), then \(\mult(\lty)=\omega\) for every \(x\COL\lty \in \LTE\).
\end{lemma}
\begin{proof}
This follows by induction on the derivation of \(\LTE \pLin s:\lrty^\omega\),
with case analysis on the last rule used.
Since the other cases are trivial, we discuss only the case where the last rule is \rname{LT-App}.
In that case, we have \(s=s_0s_1\) and:
\[
\begin{array}{l}
\LTE_0 \pLin s_0: (\lty_1\to\lrty^\omega)^m\\
\LTE_1 \pLin s_1: \lty_1
\end{array}
\]
By the well-formedness condition on linear/non-linear types, both \(m\) and \(\mult(\lty_1)\) must be \(\omega\).
Thus, the results follows immediately from the induction hypothesis.
\end{proof}
\begin{lemma}[substitution]
\label{lem:linear-substitution}
Suppose \(\LTE_0, x\COL\lty' \pLin s_0:\lty\) and \(\LTE_1 \pLin s_1:\lty'\).
If \(\LTE_0+\LTE_1\) is well-defined, then \(\LTE_0+\LTE_1\pLin [s_1/x]s_0:\lty\).
\end{lemma}
\begin{proof}
This follows by induction on the derivation of \(\LTE_0, x\COL\lty' \pLin s_0:\lty\),
with case analysis on the last rule used. We discuss only the main cases; the other cases are trivial
or similar.
\begin{itemize}
\item Case \rname{LT-Var}: In this case, \(s_0=x\) and \(\LTE_0=\emptyset\). Thus, the result follows
immeidately.
\item Case \rname{LT-App}: In this case, we have \(s_0=s_{0,0}s_{0,1}\) and:
\[
\begin{array}{l}
\LTE_{0,0} \pLin s_{0,0}:(\lty_1\to\lty)^m\\
\LTE_{0,1}\pLin s_{0,1}:\lty_1\\
\LTE_{0,0}+\LTE_{0,1} = \LTE_0, x\COL\lty'
\end{array}
\]
Let \(\LTE_{0,i}'\) be the environment obtained by removing \(x\COL\lty'\) from \(\LTE_{0,i}\).
We perform case analysis on \(\mult(\lty')\).
\begin{itemize}
\item If \(\mult(\lty')=\omega\), then 
we have \(\LTE'_{0,0}+\LTE_1\pLin [s_1/x]s_{0,0}:(\lty_1\to\lty)^m\),
because: 
\begin{enumerate} 
\item If \(\LTE_{0,0}=\LTE'_{0,0}, x\COL\lty'\), then the result follows from the induction hypothesis.
\item If \(\LTE_{0,0}=\LTE'_{0,0}\), then by Lemma~\ref{lem:variable}, \(x\) does not occur in \(s_{0,0}\).
Thus, we have \([s_1/x]s_{0,0}=s_{0,0}\), and hence \(\LTE'_{0,0}\pLin [s_1/x]s_{0,0} : (\lty_1\to\lty)^m\).
By Lemma~\ref{lem:env-for-omega}, \(\LTE_1\) contains only non-linear types. Therefore we obtaned
the required result by using \rname{LT-Weak}.
\end{enumerate}
Similarly, we also have \(\LTE'_{0,1}+\LTE_1\pLin [s_1/x]s_{0,1}:\lty_1\).
Since \(\LTE_1\) contains only non-linear types, we have:
\[(\LTE'_{0,0}+\LTE_1)+(\LTE'_{0,1}+\LTE_1) = (\LTE'_{0,0}+\LTE'_{0,1})+\LTE_1 = \LTE_0+\LTE_1.\]
Thus, by using \rname{LT-App}, we have the required result.
\item If \(\mult(\lty')=1\), then by Lemma~\ref{lem:used-once}, \(x\) occurs exactly once
in either \(s_{0,0}\) or \(s_{0,1}\). Since the other case is similar, let us consider only the case
where \(x\) occurs in \(s_{0,0}\). Then by Lemma~\ref{lem:variable}, we have
\(\LTE'_{0,0},x\COL\lty'\pLin s_{0,0}:(\lty_1\to\lty)^m\) and \(\LTE'_{0,1}\pLin s_{0,1}:\lty_1\).
By the induction hypothesis, we have \(\LTE'_{0,0}+\LTE_1 \pLin [s_1/x]s_{0,0}:(\lty_1\to\lty)^m\).
By applying \rname{LT-App}, we obtain 
\((\LTE'_{0,0}+\LTE_1)+\LTE'_{0,1} \pLin ([s_1/x]s_{0,0})s_{0,1}:\lty\).
The result follows, since \((\LTE'_{0,0}+\LTE_1)+\LTE'_{0,1} = (\LTE'_{0,0}+\LTE'_{0,1})+\LTE_1 = \LTE_0+\LTE_1\),
and \(([s_1/x]s_{0,0})s_{0,1} = [s_1/x](s_{0,0}s_{0,1})\).
\end{itemize}
\end{itemize}
\end{proof}
\begin{lemma}[subject reduction]
\label{lem:linear-subject-reduction}
If \(\LTE\pLin s:\lty\) and \(s\red s'\), then \(\LTE\pLin s':\lty\).
\end{lemma}
\begin{proof}
This follows by induction on the derivation of \(\LTE\pLin s:\lty\), with case analysis on the last rule used.
Since the other cases are trivial, we discuss only the case where the last rule is \rname{LT-App},
in which case \(s=s_0s_1\). If the reduction \(s\red s'\) comes from that of \(s_0\) or \(s_1\),
the result follows immediately. Thus, we can focus on the case where
\(s_0 = \lambda x.s_0'\) and \(s'=[s_1/x]s_0'\).
By \rname{LT-App}, we have:
\[
\LTE_0 \pLin \lambda x.s_0':(\lty'\to\lty)^m
\qquad \LTE_1\pLin s_1:\lty'\qquad \LTE = \LTE_0+\LTE_1.
\]
By the first condition, we also have \(\LTE_0, x\COL\lty' \pLin s_0:\lty\).
By Lemma~\ref{lem:linear-substitution}, we have \(\LTE_0+\LTE_1 \pLin [s_1/x]s_0 : \lty\)
as required.
\end{proof}
\begin{lemma}
\label{lem:linearity}
If \(\lty\) is a ground type, \(s\) is closed, and there exists a derivation tree for \(\emptyset \pLin C[s]:\lty\) in which
a linear type is assigned to every subterm of \(C[s]\) containing the occurrence \(s\), then \(C\) is a linear context.
\end{lemma}
\begin{proof}
By the assumption that \(s\) is closed and
a linear type is assigned to every subterm containing \(s\), we have
\(x\COL \lrty^1 \pLin C[x]:\lty\),
where \(x\) is a fresh variable, and
\(\lrty^1\) is the type assigned to \(s\) in the derivation of 
\(\emptyset \pLin C[s]:\lty\). (Note that 
\(m\) in \rname{LT-Abs} must be \(1\) whenever \rname{LT-Abs} is applied to a term containing \(s\).)
Let \(t\) be the call-by-name normal form of \(C[x]\). By Lemma~\ref{lem:linear-subject-reduction},
we have \(x\COL\lrty^1 \pLin t:\lty\). By Lemma~\ref{lem:used-once}, \(t\) contains exactly
one occurrence of \(x\).
\end{proof}
We now give a translation \(\toLty{(\cdot)}\) from intersection types to linear/non-linear types. 
Our intention is that if \(\PTE\PM t:\pty\hascost c\tr s\), then \(\toLty{\pty}\) represents the type
of \(s\) (although it does not always hold, actually).
We translate a type with a non-empty marker set as a linear type. In the case of a function type
\((F,M,\pity\to \pty)\), markers in the argument type \(\pity\) are passed to a return value type;
thus we take them into account to determine the linearity of \(\pty\).
\[
\begin{array}{l}
\toLty{(F,M,\T)} = 
 \left\{
  \begin{array}{ll}
     \T^1 & \mbox{if \(M\neq \emptyset\)}\\
    \T^\omega & \mbox{otherwise}
  \end{array}\right.\\
\toLty{(F,M,\pity\to \prty)} =\\\quad
 \left\{
  \begin{array}{ll}
   \toLty{(F,M,\prty)} & \mbox{if $\pity=\emptyset$}\\
  \big(\toLty{\pty}\to 
    \toLty{(F, M\uplus \markers(\pty), \pity'\to \prty)}\big)^1 &
          \hfill \mbox{if \(M\neq \emptyset\), \(\pity=\set{\pty}\cup\pity'\), and \(\pty<\pity'\)}\\
  \big(\toLty{\pty}\to 
    \toLty{(F, M\uplus \markers(\pty), \pity'\to \prty)}\big)^\omega &
          \hfill \mbox{if \(M= \emptyset\), \(\pity=\set{\pty}\cup\pity'\), and \(\pty<\pity'\)}\\
  \end{array}
 \right.
\end{array}
\]
Here, \(\pty<\pity\) means \(\pty<\pty'\) holds for every \(\pty'\in\pity\).
Note that for every \(\pty\), \(\toLty{\pty}\) is a well-formed linear/non-linear type.
Also note that \(\toLty{(F, M, \prty)}\) is linear if and only if \(M \neq \eset\).

The translation is extended to type envrionments by:
\[
\toLty{\PTE} = \set{x_{\pty}\COL \toLty{\pty} \mid x\COL\pty\in \PTE}.
\]

The translation judgment \(\PTE\PM t:\pty\hascost c\tr s\) is transformed to
a linear type judgment by:
\[
\toLty{(\PTE\PM t:\pty\hascost c\tr s)} =
\left\{
\begin{array}{ll}
  \toLty{\PTE}\pLin s: \lrty^1 & \mbox{if \(c>0\) and \(\toLty{\pty}=\lrty^m\)}\\
  \toLty{\PTE}\pLin s: \toLty{\pty} & \mbox{otherwise}
\end{array}
\right.
\]

Given a derivation tree \(\dt\) for \(\PTE\PM t:\pty\hascost c\tr s\),
we write \(\toLty{\dt}\) for the derivation tree obtained by replacing each judgment
\(\PTE'\PM t':\pty'\hascost c'\tr s'\) with
\(\toLty{(\PTE'\PM t':\pty'\hascost c'\tr s')}\).
Note that \(\toLty{\dt}\) may not be a valid derivation tree in the linear type system.
We say that \(\toLty{\dt}\) is \emph{admissible} if,
for each derivation step 
\(\raisebox{-1ex}{\infers{J}{J_1 & \cdots & J_k}}\) in \(\toLty{\dt}\),
\(J\) can be obtained from \(J_1,\ldots,J_k\) in the linear type system.

\begin{example}
Let \(\dt_0\) be:
\[
\infers{\emptyset\PMm{1} \lambda x.\Ta\,x: (\emptyset, \emptyset, (\emptyset, \set{0},\T)\to\T)\hascost 1
  \tr \lambda x_{(\emptyset, \set{0},\T)}.\Ta\,x_{(\emptyset, \set{0},\T)}}
{\infers{x\COL (\emptyset, \set{0},\T)\PMm{1} \Ta\,x: (\emptyset,\set{0},\T)\hascost 1\tr \Ta\,x_{(\emptyset, \set{0},\T)}}
{{x\COL (\emptyset, \set{0},\T)\PMm{1} x\COL (\emptyset, \set{0},\T)\hascost{0}\tr x_{(\emptyset, \set{0},\T)}}}}
\]
\(\toLty{\dt_0}\) is:
\[
\infers{\emptyset \pLin \lambda x_{(\emptyset, \set{0},\T)}.\Ta\,x_{(\emptyset, \set{0},\T)}: (\T^1\to\T^1)^1}
 {\infers{x_{(\emptyset, \set{0},\T)}\COL\T^1\pLin \Ta\,x_{(\emptyset, \set{0},\T)}:\T^1}
  {x_{(\emptyset, \set{0},\T)}\COL\T^1\pLin x_{(\emptyset, \set{0},\T)}\COL\T^1}},
\]
which is admissible.
\end{example}

\begin{example}
It is not true that
every valid derivation in the intersection type system is mapped to an admissible derivation
in the linear type system. Let \(\dt_1\) be:
\[
\small
\infers{\emptyset \PMm{2}(\lambda x.\Ta\,x\,x)\Tc: (\set{0},\set{1},\T)\hascost 0\tr (\lambda x_{(\set{0},\emptyset,\T)}.
\Ta\, x_{(\set{0},\emptyset,\T)}\, x_{(\set{0},\emptyset,\T)})\Tc}
{\infers{\emptyset \PMm{2}\lambda x.\Ta\,x\,x:(\set{0},\emptyset,(\set{0},\emptyset,\T)\to\T)\hascost 0\tr 
\lambda x_{(\set{0},\emptyset,\T)}.\Ta\, x_{(\set{0},\emptyset,\T)}\, x_{(\set{0},\emptyset,\T)}}
 {\infers{x\COL(\set{0},\emptyset,\T)\PMm{2}\Ta\,x\,x:(\set{0},\emptyset,\T)\hascost 0\tr 
\Ta\, x_{(\set{0},\emptyset,\T)}\, x_{(\set{0},\emptyset,\T)}}{\rule{0ex}{2.2ex}\smash{\vdots}}}
& \infers[PTr-Mark]{\emptyset \PMm{2}\Tc: (\set{0},\set{1},\T)\hascost 0\tr \Tc}
 {\emptyset \PMm{2}\Tc: (\set{0},\emptyset,\T)\hascost 0\tr \Tc}}
\]
However, \(\toLty{\dt_1}\) is:
\[
\small
\infers{\emptyset \pLin(\lambda x_{(\set{0},\emptyset,\T)}.
\Ta\, x_{(\set{0},\emptyset,\T)}\, x_{(\set{0},\emptyset,\T)})\Tc:\T^1}
{\infers{\emptyset \pLin
\lambda x_{(\set{0},\emptyset,\T)}.\Ta\, x_{(\set{0},\emptyset,\T)}\, x_{(\set{0},\emptyset,\T)}: (\T^\omega\to\T^\omega)^\omega}
 {\infers{x\COL \T^\omega\pLin
\Ta\, x_{(\set{0},\emptyset,\T)}\, x_{(\set{0},\emptyset,\T)}:\T^\omega}{\rule{0ex}{2.2ex}\smash{\vdots}}}
& \infers{\emptyset \pLin\Tc: \T^1}
 {\emptyset \pLin\Tc: \T^\omega}}
\]
which is not admissible. Note that the argument type does not match in the last inference step.
In the next subsection (in Theorem~\ref{th:completeness-tr}), we show that
if \(\pi\in\Lang(\GRAM)\), we can construct a derivation \(\dt\) such that
\(\toLty{\dt}\) is an admissible derivation.
\end{example}

\subsubsection{Completeness}
\label{sec:completeness-tr}

Here we show the following theorem:
\begin{theorem}[completeness]
\label{th:completeness-tr}
Let \(\GRAM\) be an order-\(n\) grammar \(\GRAM\), and
let \(S\) be its start symbol.
If \(\pi\in \Lang(\GRAM)\), then 
there exist \(\dt\), \(c\), and \(s\) such that
\(\dt\pdt \emptyset \PM S:(\emptyset,\set{0,\ldots,n-1},\T)\hascost c\tr s\)
and \(\expn{n}{c}\geq |\pi|\).
Furthermore, \(\toLty{\dt}\) is an admissible derivation 
in the linear type system.
\end{theorem}
The theorem follows from the following three lemmas.
\begin{lemma}[base case]
\label{lem:subj-expansion-base}
If \(t\creds \pi\), then there exist \(\dt\) and \(c>0\) such that
\(\dt\pdt \emptyset\PMm{1} t:(\emptyset,\set{0},\T)\hascost c\tr \pi\)
and \(2^c\geq |\pi|\). Furthermore, \(\toLty{\dt}\) is an admissible derivation.
\end{lemma}

\begin{lemma}[subject expansion for closed, ground-type terms]
\label{lem:subj-expansion}
If \(t\ored{n} t'\) and \(\dt'\pdt \emptyset \PM t':(\emptyset,\set{0,\ldots,n-1},\T)\hascost c\tr s'\),
then \(\dt \pdt \emptyset \PM t:(\emptyset,\set{0,\ldots,n-1},\T)\hascost c\tr s\) for some 
\(\dt\) and \(s\).
Furthermore, if \(\toLty{\dt'}\) is an admissible derivation, so is
\(\toLty{\dt}\).
\end{lemma}

\begin{lemma}[increase of order]
\label{lem:increase-of-order}
If \(\dt'\pdt \emptyset \PMm{n} t:(\emptyset,\set{0,\ldots,n-1},\T)\hascost c'\tr s\) and \(c'>0\),
then \(\dt \pdt \emptyset \PMm{n+1} t:(\emptyset,\set{0,\ldots,n},\T)\hascost c\tr s\) for some 
\(\dt\) and \(c\) such that \(2^c \geq c'\) and \(c>0\).
Furthermore, if \(\toLty{\dt'}\) is an admissible derivation, so is
\(\toLty{\dt}\).
\end{lemma}

\begin{proof}[Proof of Theorem~\ref{th:completeness-tr}.]
Suppose \(S \redswith{\GRAM} \pi\), then we have
\[S = t_n \oreds{n} t_{n-1} \oreds{n-1} \cdots \oreds{1} t_0 \creds \pi.\]
By Lemma~\ref{lem:subj-expansion-base}, we have \[\dt_0\pdt \emptyset \PMm{1} t_0:
(\emptyset,\set{0},\T)\hascost c_0 \tr \pi\]
for some \(\dt_0\) and \(c_0 > 0\) such that \(2^{c_0}\geq |\pi|\) and
\(\toLty{\dt_0}\) is admissible.
By repeated applications of Lemmas~\ref{lem:subj-expansion}
and~\ref{lem:increase-of-order}, we have \(\dt \PMm{n} S:(\emptyset,\set{0,\ldots,n-1},\T)\hascost c\tr s\)
for some \(\dt\) and \(c\) such that \(\expn{n}{c}\geq |\pi|\) and
 \(\toLty{\dt}\) is admissible.
\end{proof}

We prove the three lemmas above in the rest of this subsection.
\subsubsection*{Proof of Lemma~\ref{lem:subj-expansion-base}}

Lemma~\ref{lem:subj-expansion-base} is a trivial corollary of the following two lemmas.

\begin{lemma}
\label{lem:typing-of-trees}
Let \(\pi\) be a tree. %
Then 
\begin{enumerate}
\item \(\dt\pdt \emptyset \PMm{1} \pi:(\emptyset,\set{0},\T)\hascost c\tr \pi\)
for some \(\dt\) and \(c>0\) such that \(2^c\geq |\pi|\) and \(\toLty{\dt}\) is an admissible derivation.
\item \(\dt\pdt \emptyset \PMm{1} \pi:(\set{0}, \emptyset,\T)\hascost 0\tr \pi\)
for some \(\dt\) such that \(\toLty{\dt}\) is an admissible derivation.
\end{enumerate}
\end{lemma}
\begin{proof}
The proof proceeds by induction on the structure of \(\pi\).
Since the second property is trivial, we discuss only the first property.
Suppose \(\pi = a\,\pi_1\,\cdots\,\pi_k\) where \(k=\arity(a)\) may be \(0\).
By the induction hypothesis,
\[
\begin{array}{l}
\dt_i\pdt \emptyset \PMm{1} \pi_i:(\emptyset,\set{0},\T)\hascost c_i\tr \pi_i
\qquad 2^{c_i} \geq |\pi_i|
\qquad c_i > 0\\
\dt'_i\pdt \emptyset \PMm{1} \pi_i:(\set{0}, \emptyset,\T)\hascost 0\tr \pi_i.
\end{array}
\]
for each \(i\in\set{1,\ldots,k}\).
If \(k>0\), then
pick \(j\in\set{1,\ldots,k}\) such that \(c_j = \max(c_1,\ldots,c_k)\).
Then, we have the following derivation \(\dt\):
\[
\infers[PTr-Const]{\emptyset
  \PMm{1} \pi :(\emptyset,\set{0},\T)\hascost c \tr \pi}
  {\dt''_1 & \cdots & \dt''_k}
\]
where \(c=c_j+k\), \(\dt''_j=\dt_j\) and \(\dt''_i = \dt'_i\) for  \(i\in\set{1,\ldots,k}\setminus \set{j}\).
We have 
\[
\begin{array}{l}
2^c -|\pi| =2^{c_j+k} - (1+|\pi_1|+\cdots + |\pi_k|)\\
 \geq 2^{c_j+k} - (1+2^{c_1}+\cdots + 2^{c_k})
 \geq 2^{c_j} \cdot 2^{k} - (1+k 2^{c_j})\\
 \geq 2^{c_j} \cdot (k+1) - (1+k 2^{c_j})\qquad \mbox{ (by \(2^k\geq k+1\))}\\
  = 2^{c_j}-1\geq 0.
\end{array}
\]
If \(k=0\), then we have
the following derivation \(\dt\):
\[
\infers[PTr-Mark]{\emptyset
  \PMm{1} a :(\emptyset,\set{0},\T)\hascost 1 \tr \pi}
{\infers[PTr-Const]{\emptyset
  \PMm{1} a :(\set{0},\emptyset,\T)\hascost 0 \tr \pi}
  {}}
\]
as required.
\end{proof}
\begin{lemma}
If \(\dt'\PM t':\pty\hascost c\tr s\) and \(t\cred t'\), then
\(\dt\PM t:\pty\hascost c\tr s\) for some \(\dt\). Furthermore 
if \(\toLty{\dt'}\) is admissible, so is \(\toLty{\dt}\).
\end{lemma}
\begin{proof}
This follows by straightforward induction on the structure of the context
used for deriving \(t\cred t'\).
\end{proof}
\subsubsection*{Proof of Lemma~\ref{lem:subj-expansion}}

\begin{lemma}[de-substitution]
\label{lem:desubstitution}
If \(\dt \pdt \PTE \PM [t_1/x]t_0:\pty\hascost c\tr s\),
then 
\[
\begin{array}{l}
\dt_0\pdt \PTE_0, x\COL \set{\pty_1,\ldots, \pty_k} \PM \term_0:\pty \hascost c_0\tr s_0\\
\dt_i\pdt\PTE_i \PM \term_1: \pty_i\hascost c_i\tr s_i \mbox{ for each \(i\in\set{1,\ldots,k}\)}\\
\PTE = \PTE_0\PTEcup \PTE_1\PTEcup\cdots \PTEcup \PTE_k\\
c = c_0+c_1+\cdots + c_k.
\end{array}
\]
Furthermore, if \(\toLty{\dt}\) is admissible, so are \(\toLty{\dt_0},\toLty{\dt_1},\ldots,\toLty{\dt_k}\).
\end{lemma}
\begin{proof}
We first discuss the case where \(t_0\) is a variable.
If \(t_0=x\), the required result holds for \(k=1\), \(\dt_i=\dt\) and
\[\dt_0 = \raisebox{-1ex}{\infers[PTr-Var]{x\COL \pty \PM x\COL\pty\hascost{0}\tr x_\pty}{}}.\]
Note that \(\toLty{\dt_0} = 
\rb{-.6ex}{\infers{x_\pty\COL \toLty{\pty} \pLin x_\pty:\toLty{\pty}}{}}\), which is admissible.
If \(t_0=y\neq x\), then the required result holds for
\(k=0\) and \(\dt_0=\dt\).

We show the other cases by
induction on the derivation tree \(\dt\), with case analysis on the last rule used.
\begin{itemize}
\item Case \rname{PTr-Weak}:
In this case, we have 
\(\dt' \pdt \PTE' \PM [t_1/x]t_0:\pty\hascost c\tr s\) and
\(\PTE= (\PTE',x\COL (F,\emptyset,\prty))\),
with \(\dt = 
\rb{-1.6ex}{\infers{\PTE',y\COL(F,M,\prty) \PM [t_1/x]t_0:\pty\hascost c\tr s}{\dt'}}\).
By the induction hypothesis, we have:
\[
\begin{array}{l}
\dt'_0\pdt \PTE'_0, x\COL \set{\pty_1,\ldots, \pty_k} \PM \term_0:\pty \hascost c_0\tr s_0\\
\dt_i\pdt\PTE_i \PM \term_1: \pty_i\hascost c_i\tr s_i \mbox{ for each \(i\in\set{1,\ldots,k}\)}\\
\PTE' = \PTE'_0\PTEcup \PTE_1\PTEcup\cdots \PTEcup \PTE_k\\
c = c_0+c_1+\cdots + c_k\\
\end{array}
\]
Let \(\PTE_0 = \PTE'_0,y\COL(F,\emptyset,\prty)\) and \(\dt_0 = \rb{-1.6ex}{\infers[PTr-Weak]{
\PTE_0, x\COL \set{\pty_1,\ldots, \pty_k} \PM \term_0:\pty \hascost c_0\tr s_0}
{\dt'_0}}\).
Then we have the required result.
Note that \(\toLty{\dt_0} = \rb{-1.6ex}{\infers{\toLty{\PTE'_0}, y_{(F,\emptyset,\prty)}\COL \toLty{(F,\emptyset,\prty)} \pLin s_0:\toLty{(\pty,c_0)}}{\toLty{\dt'_0}}}\) is admissible if \(\toLty{\dt'_0}\) is so,
because \(\toLty{(F,\emptyset,\prty)}\) is non-linear.
\item Case \rname{PTr-Mark}:
In this case,  we have:
\[
\begin{array}{l}
\dt' \pdt \PTE \PM [t_1/x]t_0: (F',M',\prty)\hascost c'\tr s\\
\Comp_n(\set{(F',c')},M \uplus M') = (F,c)\\
M \subseteq \set{j\in F \mid \eorder(t_0)\leq j<n}\\
\pty = (F,M\uplus M', \prty)
\end{array}
\]
By the induction hypothesis, we have:
\[
\begin{array}{l}
\dt'_0\pdt \PTE_0, x\COL \set{\pty_1,\ldots, \pty_k} \PM \term_0:(F',M',\prty) \hascost c'_0\tr s_0\\
\dt_i\pdt\PTE_i \PM \term_1: \pty_i\hascost c_i\tr s_i \mbox{ for each \(i\in\set{1,\ldots,k}\)}\\
\PTE = \PTE_0\PTEcup \PTE_1\PTEcup\cdots \PTEcup \PTE_k\\
c' = c'_0+c_1+\cdots + c_k\\
\end{array}
\]
Let \(\dt_0 = \rb{-1.6ex}{\infers[PTr-Mark]{
\PTE_0, x\COL \set{\pty_1,\ldots, \pty_k} \PM \term_0:(F,M,\prty) \hascost c_0\tr s_0}{\dt'_0}}\),
where
\((F,c_0) = \Comp_n(\set{(F',c_0')},M \uplus M')\).
It remains to check that \(c=c_0+c_1+\cdots + c_k\).
By the definition of \(\Comp_n\), we have \(c_0-c_0' = c-c'\). Thus, we have
\(c = c_0+c_1+\cdots + c_k\) as required.
\item Case \rname{PTr-Var}: 
In this case \(t_0\) must be a variable, which has been discussed already.
\item Case \rname{PTr-Choice}: Trivial by the induction hypothesis.
\item Case \rname{PTr-Abs}: In this case, we have:
\[
\begin{array}{l}
t_0 = \lambda y.t_0'\\
\pty = (F,M\setminus \Markers(\pity)),
    \pity\to \prty)\\
\dt' \pdt \PTE,y\COL\pity
   \PM [t_1/x]t'_0: (F,M,\prty)\hascost c\tr s'\\
s = \lambda \tilde{x}_{\pity}.s'
\end{array}
\]
By the induction hypothesis, we have:
\[
\begin{array}{l}
\dt'_0 \pdt \PTE_0, y\COL\pity_0, x\COL \set{\pty_1,\ldots,\pty_k} 
   \PM t'_0: (F,M,\prty)\hascost{c_0}\tr s_0\\
\dt'_i \pdt \PTE_i, y\COL\pity_i \PM t_1: \pty_i\hascost{c_i}\tr s_i \mbox{ for $i\in\set{1,\ldots,k}$}\\
\PTE = \PTE_0\PTEcup\cdots \PTEcup \PTE_k\\
\pity=\pity_0 \PTEcup\cdots\PTEcup\pity_k\\
c = c_0+\cdots + c_k
\end{array}
\]
By the convention on bound variables, we can assume that \(y\) does not occur in \(t_1\).
We have thus \(\pity_i\subptyk{}\emptyset \) for each \(i\in\set{1,\ldots,k}\).
Thus by repeated applications of weakening and strengthening, we have:
\[
\begin{array}{l}
\dt''_0 \pdt \PTE_0, y\COL\pity, x\COL \set{\pty_1,\ldots,\pty_k} 
   \PM t'_0: (F,M,\prty)\hascost{c_0}\tr s_0\\
\dt_i \pdt \PTE_i \PM t_1: \pty_i\hascost{c_i}\tr s_i \mbox{ for $i\in\set{1,\ldots,k}$}
\end{array}
\]
The required result holds for:
\[\dt_0 = \rb{-1.6ex}{\infers{\PTE_0, x\COL \set{\pty_1,\ldots,\pty_k} 
\PM \lambda y.t'_0: (F,M\setminus \Markers(\pity),\pity\to\prty)\hascost{c_0}\tr 
\lambda \seq{y}_\pity.s_0}{\dt''_0}}.\]
To check that \(\toLty{\dt_0}\) is admissible if so is \(\toLty{\dt}\), it suffices to 
check that the last consecutive applications of abstractions to obtain the typing of
\(\lambda \seq{y}_\pity.s_0\) is valid. Suppose that \(\pity=\set{\pty'_1,\ldots,\pty'_\ell}\)
with \(\pty'_1<\cdots <\pty'_\ell\), and that \(M\setminus (\Markers(\pty'_{j+1})\uplus\cdots\uplus
\Markers(\pty'_k)) = \emptyset\). We need to show
\(\toLty{(\PTE_0, y\COL\set{\pty'_{1},\ldots,\pty'_j}, x\COL \set{\pty_1,\ldots,\pty_k})}\) is non-linear. 
By \(\dt''_0\) and Lemma~\ref{lem:markContextType},
we have:
\[\Markers(\PTE_0,x\COL \set{\pty_1,\ldots,\pty_k})\uplus \Markers(\pity)\subseteq  M.\] 
Thus, we have 
\[\Markers(\PTE_0, y\COL\set{\pty'_{1},\ldots,\pty'_j},x\COL \set{\pty_1,\ldots,\pty_k})\subseteq 
M\setminus 
(\Markers(\pty'_{j+1})\uplus\cdots\uplus
\Markers(\pty'_k))=\emptyset\] as required.
\item Case \rname{PTr-App}:
In this case, we have:
\[
\begin{array}{l}
\eorder(t_{0,0})=\ell\qquad
t_0 = t_{0,0}t_{0,1}\\
\dt'_0\pdt \PTE'_0 \PM [t_1/x]t_{0,0}: (F_0',M_0',\set{(F'_1,M_1',\prty_1),\ldots,(F'_p,M_p',\prty_p)}\to\prty)\hascost c_0'\tr s'_0\\
\dt'_i\pdt\PTE'_i \PM [t_1/x]t_{0,1}: (F_i'',M_i'',\prty_i)\hascost c_i'\tr s'_i\quad
F''_i\restrict{<\ell}=F'_i
\quad M''_i\restrict{<\ell}=M'_i
 \mbox{ for each $i\in\set{1,\ldots,p}$}\\
M = M_0'\uplus M_1''\uplus \cdots \uplus M''_p\\
\Comp_n(\set{(F_0',c'_0)}\mcup \set{(F''_i\restrict{\geq \ell},c_i')\mid i\in\set{1,\ldots,p}},M)=(F,c)\\
\pty = (F, M, \prty)\\
\end{array}
\]
By the induction hypothesis, we have:
\[
\begin{array}{l}
\dt'_{0,0}\pdt \PTE'_{0,0},x\COL \set{\pty_{0,1},\ldots,\pty_{0,k_0}} 
\PM t_{0,0}: (F_0',M_0',\set{(F'_1,M_1',\prty_1),\ldots,(F'_p,M_p',\prty_p)}\to\prty)\hascost c_{0,0}\tr s''_0
 \\
\dt'_{i,0}\pdt \PTE'_{i,0},x\COL \set{\pty_{i,1},\ldots,\pty_{i,k_i}}  \PM t_{0,1}:(F_i'',M_i'',\prty_i)\hascost c_{i,0}\tr s''_i
 \mbox{ for each $i\in\set{1,\ldots,p}$}\\
\dt'_{i,j}\pdt \PTE'_{i,j} \PM t_1: \pty_{i,j}\hascost{c_{i,j}}\tr s_{i,j}
 \mbox{ for each $i\in\set{0,\ldots,p}, j\in \set{1,\ldots,k_i}$}\\
c_i' = \Sigma_{j\in \set{0,\ldots,k_i}} c_{i,j} \mbox{ for each $i\in\set{0,\ldots,p}$}\\
\PTE_i' = \PTEsigma_{j\in \set{0,\ldots,k_i}}\PTE'_{i,j} \mbox{ for each $i\in\set{0,\ldots,p}$}
\end{array}
\]
Let \(\set{\pty_{i,j}\mid i\in \set{0,\ldots,p},j\in\set{1,\ldots,k_i}}=
\set{\pty_1,\ldots,\pty_k}\) with \(\pty_1<\cdots <\pty_k\).
For each \(q\in\set{1,\ldots,k}\), we pick a pair
\((i_q,j_q)\) such that \(\pty_{q} = \pty_{i_q,j_q}\), and write
\(I\) for \(\set{(i,j)\mid i\in\set{0,\ldots,p},j\in\set{1,\ldots,k_i}}
\setminus\set{(i_q,j_q) \mid q\in\set{1,\ldots,k}}\).
Let \(\PTE_0 = (\PTEsigma_{i\in\set{0,\ldots,p}} \PTE'_{i,0}) \PTEcup (\PTEsigma_{(i,j)\in I} \PTE'_{i,j})\).
Note that \(\PTE_0\) is well defined, because \(\PTE'_{i,j}\) contains no markers for each
\((i,j)\in I\).
Let \(\PTE_{q}\) and \(c_q\) be \(\PTE_{i_q,j_q}\) and \(c_{i_q,j_q}\) respectively
for each \(q\in\set{1,\ldots,k}\).
Let \(\dt_0\) be:
\[
\infers[PTr-Weak]{\PTE_0, x\COL\set{\pty_1,\ldots,\pty_k}\PM t_0:(F,M,\prty)\hascost 
  c_0\tr s_0''s''_1\cdots s''_p}
 { \infers[PTr-Weak]{\rule{0ex}{2.2ex}\smash{\vdots}}
   {\infers[PTr-App]{(\PTEsigma_{i\in\set{0,\ldots,p}} \PTE'_{i,0}),
  x\COL\set{\pty_1,\ldots,\pty_k}\PM t_0:(F,M,\prty)\hascost 
  c_0\tr s_0''s''_1\cdots s''_p}
  {\dt'_{0,0} & \dt'_{1,0} & \cdots & \dt'_{p,0}}}}
\]
where \(c_0 = c - \Sigma_{i\in \set{0,\ldots,p}}(c_i'-c_{i,0})\).
Let  \(\dt_q\) be \(\dt'_{i_q,j_q}\). Then we have the required result.
The condition \(c=c_0+c_1+\cdots + c_k\) is verified by:
\[
\begin{array}{l}
c-c_0 = 
\Sigma_{i\in \set{0,\ldots,p}}(c_i'-c_{i,0})
 = \Sigma_{i\in \set{0,\ldots,p},j\in\set{1,\ldots,k_i}} c_{i,j}\\
 = (\Sigma_{q\in\set{1,\ldots,k}}c_{i_q,j_q}) + (\Sigma_{(i,j)\in I} c_{i,j})
 = (\Sigma_{q\in\set{1,\ldots,k}}c_{i_q,j_q}) 
 = c_1+\cdots + c_k
\end{array}
\]
Note that \(c_{i,j}=0\) for \((i,j)\in I\).
\item Case \rname{PTr-Const}:
Similar to the case for \rname{PTr-App} above.
\item Case \rname{PTr-NT}: In this case, \(t_0\) does not contain a variable (recall that now \(t_0\) is not a
      variable).
Thus, the result follows immeidately for \(k=0\) and \(\dt_0=\dt\).
\end{itemize}
\end{proof}
Lemma~\ref{lem:subj-expansion} is a special case of the following lemma
\begin{lemma}[subject expansion]
\label{lem:subj-expansion-for-induction}
If \(t\ored{n} t'\) and \(\dt'\pdt \PTE \PM t':\pty\hascost c\tr s'\),
then \(\dt \pdt \PTE \PM t:\pty\hascost c\tr s\) for some 
 \(\dt\) and \(s\).
Furthermore, if \(\toLty{\dt'}\) is an admissible derivation, so is
\(\toLty{\dt}\).
\end{lemma}
\begin{proof}
This follows by induction on the context used for deriving \(t\ored{n} t'\).
Since the induction steps are trivial, we discuss only
the base case, where the context is \(\Hole\). In this case, either \(t=A\) with \(A=t'\in \GRAM\),
or \(t=(\lambda x.t_0)t_1\) with \(t' = [t_1/x]t_0\) and \(\eorder(\lambda x.\term_0)=n\).
In the former case, the result holds for
\(\dt = \dfrac{\dt'}{\PTE\PM A:\pty\hascost c\tr s'}\) and \(s=s'\).

In the latter case, by Lemma~\ref{lem:desubstitution}, we have:
\[
\begin{array}{l}
\dt_0\pdt\PTE_0, x\COL \set{\pty_1,\ldots, \pty_k} \PM \term_0:(F,M,\prty) \hascost c_0\tr s_0\\
\dt_i\pdt \PTE_i \PM \term_1: \pty_i\hascost c_i\tr s_i\\
c = c_0+c_1+\cdots + c_k\\
\pty_1<\cdots < \pty_k\\
\PTE = \PTE_0 \PTEcup (\textstyle\sum_{i\in\set{1,\ldots,k}} \PTE_{i})
\end{array}
\]
where \(\pty = (F,M,\prty)\). Furthermore, \(\toLty{\dt_0},\ldots,\toLty{\dt_k}\) are 
admissible if \(\toLty{\dt'}\) is.
By \rname{PTr-Abs}, we have:
\[
\PTE_0 \PM \lambda x.\term_0:(F,M\setminus(\biguplus_i \Markers(\pty_i)),\set{\pty_1,\ldots,\pty_k}\to \prty) \hascost c_0\tr 
  \lambda x_{\pty_1}.\cdots \lambda x_{\pty_k}.s_0.
\]
Since \(\eorder(\lambda x.\term_0)=n\), by applying the following special case
of \rname{PTr-App}:
\infrule[PTr-App-Order-n]{\eorder(t'_0)=n\\
\PTE_0\PM t'_0:(F_0,M_0,\set{\pty_1,\ldots,\pty_k}\to \prty)\hascost{c_0}\tr s'_0\\
  \PTE_{i}\PM t_1:\pty_i\hascost{c_{i}}\tr s_{i}
\mbox{ for each $i\in\set{1,\ldots,k}$}\\
  M' = M_0\uplus (\biguplus_{i\in\set{1,\ldots,k}} \Markers(\pty_i))\\
  \Comp_{n}(\set{(F_0,c_0)}\mcup \set{(\emptyset,c_{i})\mid i\in\set{1,\ldots,k}}, M')
   = (F', c')\\
\pty_1<\cdots< \pty_k
}
{\PTE_0\PTEcup (\textstyle\sum_{i\in\set{1,\ldots,k}} \PTE_{i}) \PM t'_0t_1:
  (F', M', \prty)\hascost{c'}\tr s'_0\, s_{1}\,\cdots\,s_{k}}
we obtain 
\[ \PTE \PM (\lambda x.t_0)t_1:(F',M',\prty)\hascost c'\tr (\lambda x_{\pty_1}.\cdots\lambda x_{\pty_k}.s_0) s_1\cdots s_k\]
for \(M',F',c'\) such that:
\[
\begin{array}{l}
  M' = (M\setminus(\biguplus_i \Markers(\pty_i)))\uplus (\biguplus_{i\in\set{1,\ldots,k}} \Markers(\pty_i))\\
  \Comp_{n}(\set{(F,c_0)}\mcup \set{(\emptyset,c_{i})\mid i\in\set{1,\ldots,k}}, M')
   = (F', c')
\end{array}
\]
By Lemma~\ref{lem:markContextType}, \(\biguplus_i \Markers(\pty_i)\subseteq  M\),
and therefore we have \(M=M'\).
By the condition
\(\PTE_0, x\COL \set{\pty_1,\ldots, \pty_k} \PM \term_0:(F,M,\prty) \hascost c_0\tr s_0\),
we have \(F\cap M=\emptyset\). Thus, by Lemma~\ref{lem:empDisjComp},
we have also \(F=F'\) and \(c=c'\).
Therefore, 
\[\dt \pdt \PTE \PM \term: \pty\hascost{c} \tr s\]
where \(s=(\lambda x_{\pty_1}.\cdots\lambda x_{\pty_k}.s_0) s_1\cdots s_k\) and \(\dt\) is:
\[ \infers{\PTE \PM \term: \pty\hascost{c} \tr s}
  {
\infers{\PTE_0 \PM \lambda x.\term_0:(F,M\setminus(\biguplus_i \Markers(\pty_i)),\set{\pty_1,\ldots,\pty_k}\to \prty) \hascost c_0\tr \lambda x_{\pty_1}.\cdots \lambda x_{\pty_k}.s_0}
  {\dt_0}
 & \dt_1 & \cdots & \dt_k}
\]
If \(\toLty{\pty_i} = \lrty_i^{m_i}\), then 
the type \(\lty_i\) of \(s_i\) in the conclusion of \(\dt_i\) is \(\lrty_i^{1}\) if \(c_i>0\)
and \(\lrty_i^{m_i}\) otherwise. In the former case, \(n-1\in \Markers(\pty_i)\)
by Lemma~\ref{lem:marker}. So, \(\lty_i=\toLty{\pty_i}\) in any case.
Thus, if \(\toLty{\dt'}\) is admissible (which implies \(\toLty{\dt_0},\ldots,\toLty{\dt_k}\) are also admissible),
then \(\toLty{\dt}\) is also admissible.
\end{proof}

\subsubsection*{Proof of Lemma~\ref{lem:increase-of-order}}

\begin{lemma}
\label{lem:increase-of-order-without-mark}
If \(\dt'\pdt \PTE \PMm{n} t:(F,M,\prty)\hascost c'\tr s\),
then \(\dt \pdt \PTE \PMm{n+1} t:(F',M,\prty)\hascost 0\tr s\) for some 
\(\dt\), where \(F'=F\cup\set{n}\) if \(c'>0\) and \(F'=F\) otherwise.
Furthermore, if \(\toLty{\dt'}\) is an admissible derivation, so is
\(\toLty{\dt}\).
\end{lemma}
\begin{proof}
This follows by induction on the derivation tree \(\dt'\).
We discuss below only the case where the last rule is \rname{PTr-Abs}, since
the other cases follow immediately from the induction hypothesis.
In this case, we have \(t = \lambda x.t_0\) and:
\[
\begin{array}{l}
\dt_0' \pdt \PTE, x\COL\pity \PM t_0:(F,M_0,\prty')\hascost c'\tr s_0\\
M = M_0\setminus \Markers(\pity)\qquad
\prty = \pity\to \prty'\qquad s = \lambda \seq{x}_{\pity}.s_0
\end{array}
\]
By the induction hypothesis, we have:
\[
\dt_0\pdt \PTE, x\COL\pity \PMm{n+1} t_0:(F',M_0,\prty')\hascost 0\tr s_0.
\]
By this, we have \(M_0 = M \uplus \Markers(\pity)\) by Lemma~\ref{lem:markContextType},
and also we have 
\[
\dt = \rb{-1.6ex}{\infers[PTr-Abs]{ \PTE \PMm{n+1} t: (F',M,\prty)\hascost 0\tr  s}{\dt_0}}.
\]
It remains to check that the last step of \(\toLty{\dt}\) is admissible.
Suppose \(\pity = \set{\pty_1,\ldots,\pty_k}\) with \(\pty_1<\cdots < \pty_k\).
We show that
\[
\infers{\toLty{\PTE} \pLin \lambda x_{\pty_1}\dots\lambda x_{\pty_k}.s_0: \toLty{(F',M , \set{\pty_1,\ldots,\pty_k} \to \prty')}}
 { \infers{\rule{0ex}{2.2ex}\smash{\vdots}}
{
\toLty{(\PTE,x\COL\pty_1,\ldots,x\COL\pty_{k})} \pLin s_0: 
\toLty{(F',M \uplus \Markers(\pty_{1})\uplus \cdots \uplus \Markers(\pty_{k}), \prty')}
}}
\]
is derivable by \rname{LT-Abs} (and possibly by \rname{LT-Dereliction} when we apply \rname{LT-Abs} with \(m=\omega\)).
For this, it suffices to check that, for each \(i \in \set{1,\dots,k}\), if \(\toLty{(\PTE,x\COL\pty_1,\ldots,x\COL\pty_{i-1})}\) contains
a linear type (i.e., \(\Markers(\PTE,x\COL\pty_1,\ldots,x\COL\pty_{i-1})\neq \emptyset\)),
then %
\begin{align*}
&\toLty{(F',M\uplus \Markers(\pty_{1})\uplus \cdots \uplus
 \Markers(\pty_{i-1}),\set{\pty_i,\dots,\pty_k}\to\prty')}
\\ =\ &
(\toLty{\pty_i} \to \toLty{(F',M\uplus \Markers(\pty_{1})\uplus \cdots \uplus
 \Markers(\pty_{i}),\set{\pty_{i+1},\dots,\pty_k}\to\prty')})^{1}
\end{align*}
(i.e., \(M\uplus \Markers(\pty_{1})\uplus \cdots \uplus \Markers(\pty_{i-1})\neq \emptyset\)).
This follows from \(\Markers(\PTE)\subseteq M\), which follows from Lemma~\ref{lem:markContextType}.
\end{proof}

Lemma~\ref{lem:increase-of-order} is a special case of the following lemma.
\begin{lemma}
\label{lem:increase-of-order-for-induction}
If \(\dt'\pdt \PTE \PMm{n} t:(F,M,\prty)\hascost c'\tr s\) and \(c'>0\),
then there exist \(\dt\) and \(c>0\) such that 
\(\dt \pdt \PTE \PMm{n+1} t:(F,M\uplus\set{n},\prty)\hascost c\tr s\)
 and \(2^c \geq c'\).
Furthermore, if \(\toLty{\dt'}\) is an admissible derivation, so is
\(\toLty{\dt}\).
\end{lemma}

\begin{proof}
Note that, by the well-formedness condition, the orders of types in \(\PTE\) are at most 
\(n-1\), and \(\eorder(t)\leq n\).
The proof proceeds by induction on the derivation tree \(\dt'\), with the case analysis on 
the last rule used. The cases for \rname{PTr-Weak}, \rname{PTr-Choice}, and \rname{PTr-NT}
follow immediately from the induction hypothesis. We discuss the other cases below.
\begin{itemize}
\item Case \rname{PTr-Mark}:
In this case, we have:
\[
\begin{array}{l}
\dt'_0\pdt \PTE \PMm{n} t:(F_0,M_0,\prty)\hascost c_0'\tr s\\
M_1 \subseteq \set{j\in F_0\mid \eorder(t)\leq j<n}\\
M = M_0\uplus M_1\\
\Comp_n(\set{(F_0,c_0')},M) = (F,c')
\end{array}
\]
We consider only the case where \(M_1\neq \emptyset\), and hence
\(\eorder(t)\leq n-1\). 

If \(c_0'=0\), by the assumption \(c'>0\), we have \(c'=1\) and \(n-1\in M_1\cap F_0\). 
By Lemma~\ref{lem:increase-of-order-without-mark},
we have 
\[\dt_0\pdt \PTE \PMm{n+1} t:(F_0,M_0,\prty)\hascost 0\tr s
\]
By applying \rname{PTr-Mark} (to add markers \(M_1\uplus\set{n}\)),
we have the following derivation \(\dt\):
\[{\infers[PTr-Mark]{
 \PTE \PMm{n+1} t:(F,M\uplus\set{n},\prty)\hascost 1\tr s}{\dt_0}}\]
as required.

If \(c_0'>0\), then by the induction hypothesis, we have:
\[ \dt_0 \pdt \PTE \PMm{n+1}t:(F_0, M_0\uplus\set{n},\prty)\hascost c_0\tr s\]
for some \(\dt_0\) and \(c_0\) such that \(2^{c_0}\geq c_0'\).
Let \(\dt\) be:
\[ \infers[PTr-Mark]{\PTE\PMm{n+1} t:(F', M\uplus\set{n}, \prty)
\hascost c\tr s}{\dt_0},\]
where \((F',c) = \Comp_{n+1}(\set{(F_0,c_0)},M\uplus\set{n})\).
By direct calculation, we can check \(F'=F\).
We check the condition \(2^c\geq c'\). 
Since \(c'>0\), by Lemma~\ref{lem:marker}, we have \(n-1\in M\).
Let \(f_n, f'_n, f'_{n+1}\) be those that occur in the calculation of 
\(\Comp_{n+1}(\set{(F_0,c_0)},M\uplus\set{n})\). Note that \(f'_n\) is the same as the one
that occurs in the calculation of 
\(\Comp_n(\set{(F_0,c_0')},M) = (F,c')\).
Thus, we have \(c = f'_{n+1} + c_0 = f_n + c_0 = f_n' +c_0\) and \(c'=f_n'+c_0'\).
Therefore, we have:
\[ 2^c -c' = 2^{f_n'+c_0} -(f_n'+c_0')= 2^{f_n'} 2^c_0 - (f_n'+c_0')
 = (2^{f_n'}-1) 2^{c_0} - f_n' + (2^{c_0}-c_0')
 \geq 2^{f_n'}-1 - f_n' \geq 0.\]
It remains to check that the admissibility of \(\toLty{\dt'}\) implies that of 
\(\toLty{\dt}\). If \(\toLty{\dt'}\) is admissible, so is \(\dt'_0\). By the induction 
hypothesis, \(\toLty{\dt_0}\) is also admissible. Thus, \(\toLty{\dt}\) is also admissible.
Note that the last step in \(\toLty{\dt}\) does not change the judgment.
\item Case \rname{PTr-Var}: 
This case does not occur, since the flag counter is non-zero.
\item Case \rname{PTr-Abs}: In this case, we have \(t = \lambda x.t_0\) and:
\[
\begin{array}{l}
\dt_0' \pdt \PTE, x\COL \pity \PM t_0: (F, M_0, \prty_0)\hascost c'\tr s_0'\\
M = M_0\setminus\Markers(\pity)\qquad \prty = \pity\to\prty_0\\
\dt' = \rb{-1.6ex}{\infers[PTr-Abs]{\PTE \PM t: (F,M,\prty)\hascost c'\tr \lambda \seq{x}_\pity.s_0'}{\dt_0'}}
\end{array}
\]
By the induction hypothesis, we have
\[ \dt_0 \pdt \PTE, x\COL\pity \PMm{n+1} t_0:(F, M_0\uplus\set{n}, \prty_0)\hascost c\tr s_0\]
for some \(\dt_0, c, s_0\) such that \(2^c \geq c'\).
By using \rname{PTr-Abs}, we have the required derivation:
\[ \dt = \rb{-1.6ex}{\infers[PTr-Abs]{\PTE \PMm{n+1} t:(F, M\uplus\set{n}, \prty)\hascost c\tr \lambda \seq{x}_\pity.s_0}{\dt'}}\]
To check the admissibility condition, suppose \(\toLty{\dt'}\) is admissible. Then so is
\(\toLty{\dt_0'}\); hence \(\toLty{\dt_0}\) is also admissible by the induction hypothesis. 
It remains to check that the last step of \(\toLty{\dt}\) is admissible, which is indeed the case,
since a linear type is assigned to \(\lambda \seq{x}_\pity.s_0\) (recall \(c'>0\)).
\item Case \rname{PTr-App}: In this case, we have:
\[
\begin{array}{l}
t=t_0t_1 \qquad \eorder(t_0)=\ell\\
\dt_0' \pdt \PTE_0 \PM t_0: (F_0,M_0, \set{(F_1,M_1,\prty_1),\ldots,(F_k,M_k,\prty_k)}\to \prty)\hascost c'_0\tr s_0\\
\dt_i' \pdt \PTE_i \PM t_1: (F_i',M_i', \prty_i)\hascost c'_i\tr s_i\qquad
F_i = F_i'\restrict{<\ell}\qquad M_i = M_i'\restrict{<\ell}\\
\hfill \mbox{(for each $i\in\set{1,\ldots,k}$)}\\
M = M_0\uplus \cdots \uplus M_k\\
\Comp_n(\set{(F_0,c'_0)}\mcup \set{(F'_i\restrict{\geq \ell},c_i')\mid i\in\set{1,\ldots,k}}, M) = 
(F,c')\\
s = s_0 s_1\cdots s_k\\
\PTE = \PTE_0\PTEcup \cdots \PTEcup \PTE_k
\end{array}
\]
We consider two cases:
\begin{itemize}
\item Case where \(c_i'=0\) for every \(i\in\set{0,\ldots,k}\):
By Lemma~\ref{lem:increase-of-order-without-mark}, 
we have:
\[
\dt''_0 \pdt \PTE_0 \PMm{n+1} t_0: 
(F_0,M_0, \set{(F_1,M_1,\prty_1),\ldots,(F_k,M_k,\prty_k)}\to \prty)\hascost 0\tr s_0
\]
and 
\[
\dt''_i \pdt \PTE_i \PMm{n+1} t_1: (F_i',M_i', \prty_i)\hascost 0\tr s_i
\]
for each \(i\in\set{1,\ldots,k}\).
By the conditions
\(\Comp_n(\set{(F_0,c'_0)}\mcup \set{(F'_i\restrict{\geq \ell},c_i')\mid i\in\set{1,\ldots,k}}, M) = 
(F,c')\) and \(c'>0\), it must be the case that \(n-1\in M\), hence \(n-1\in M_j\) for some 
\(j\in\set{0,\ldots,k}\).
By applying \rname{PTr-Mark} to \(\dt''_j\), we obtain a derivation \(\dt_j\), whose conclusion is
the same as that of \(\dt''_j\) except that marker \(n\) has been added. Note that the last step
of \(\toLty{\dt_j}\) is admissible (and in fact, does not change the judgment) since 
\(n-1\in M_j\neq \emptyset\). Let \(\dt_i\) be \(\dt''_i\) for \(i\in\set{0,\ldots,k}\setminus\set{j}\).
Let \(\dt\) be:
\[ \infers[PTr-App]{\PTE\PMm{n+1} t_0t_1: (F',M\uplus\set{n},\prty)\hascost c\tr s_0s_1\cdots s_k}
{\dt_0 & \dt_1 & \cdots & \dt_k}
\]
where
\[\Comp_{n+1}(\set{(F_0,0)}\mcup\set{(F_1'\restrict{\geq \ell},0),\ldots,(F_k'\restrict{\geq \ell},0)}, 
  M\uplus\set{n}) = (F', c)\]
Let \(f_i, f_i'\) be those occurring in the calculation of \(\Comp_{n+1}(\cdots)\) above.
Note that \(f_i\ (i\leq n-1)\) and \(f_i'\ (i\leq n)\) are equivalent to those to occurring 
in the calculation of \(\Comp_n(\cdots)=(F,c')\). Thus, we have:
\[F'
= \set{\ell\in\set{0,\ldots,n} \mid f_\ell>0}\setminus (M\uplus\set{n})
 = \set{\ell\in\set{0,\ldots,n-1} \mid f_\ell>0}\setminus M
 = F.\]
Furthermore, \(c = f'_{n+1} = f_n = f'_n = c'\). 
It remains to check that the admissibility of \(\toLty{\dt'}\) implies that of \(\toLty{\dt}\),
which is trivial from the construction of \(\dt\) above (recall that
the last step of \(\toLty{\dt_j}\) is admissible).
\item Case where \(c_i'>0\) for some \(i\in\set{0,\ldots,k}\). 
Pick \(j\) such that \(c'_j = \max(c_0',\ldots,c_k')\).
Let \((F'_0, M'_0, \pty_0)\) be 
\((F_0,M_0, \set{(F_1,M_1,\prty_1),\ldots,(F_k,M_k,\prty_k)}\to \prty)\).
For \(i\in\set{0,\ldots,k}\), let \(t_i'=t_0\) if \(i=0\), and \(t_i' = t_1\) otherwise.
By applying the induction hypothesis to \(\dt_j'\),
we have:
\[ \dt_j \pdt \PTE_j \PMm{n+1} t_j': (F'_j, M'_j\uplus \set{n}, \prty_i)\hascost c_j\tr s_j.\]
For \(i\in\set{0,\ldots,k}\setminus\set{j}\), by Lemma~\ref{lem:increase-of-order-without-mark}, we have
\[ \dt_i \pdt \PTE_i \PMm{n+1} t_i': (F''_i, M'_i, \prty_i)\hascost 0\tr s_i\]
where \(F''_i = F'_i\uplus \set{n}\) if \(c'_i>0\) and \(F''_i=F'_i\) otherwise.
Let \(\dt\) be:
\[
\infers[PTr-App]{\PTE \PMm{n+1} t_0t_1: (F', M\uplus \set{n}, \prty)\hascost c\tr s_0s_1\cdots s_k}
{\dt_0 & \dt_1 & \cdots & \dt_k}
\]
where 
\[
\Comp_{n+1}(\set{(F'_j,c_j)}\mcup \set{(F''_i, 0)\mid i\in \set{0,\ldots,k}\setminus\set{j}},
  M\uplus \set{n}) = (F',c).
\]
If \(\toLty{\dt'}\) is admissible, then the last step of \(\toLty{\dt}\) is also admissible:
the change from \(\toLty{\dt'_i}\) to \(\toLty{\dt_i}\) may change the type of \(s_i\) in
the conclusion to a non-linear type, but then we can use \rname{LT-Dereliction} to adjust the linearity.
Thus, it remains to check \(F'=F\) and \(2^c \geq c'\).
Let \(f_i, f_i'\) be those occurring in the calculation of \(\Comp_{n+1}(\cdots)\) above.
Note that \(f_i\ (i\leq n-1)\) and \(f_i'\ (i\leq n)\) are equivalent to those occurring 
in the calculation of \(\Comp_n(\cdots)=(F,c')\). We have:
\[ F' = \set{\ell \in \set{0,\ldots,n}\mid f_\ell>0}\setminus (M\uplus\set{n})
  = \set{\ell \in \set{0,\ldots,n-1}\mid f_\ell>0}\setminus M
  = F.\]
We also have:
\[ 
\begin{array}{l}
c = f'_{n+1}+c_j
= f_n + c_j
= f'_n + \set{i \mid n\in F''_i} + c_j
= f'_n + (\set{i \mid c'_i>0}-1) + c_j.
\end{array}
\]
Let \(p = \set{i \mid c'_i>0}-1\). 
Then, we have:
\[
\begin{array}{l}
2^c - c' \geq 2^c - (f'_n+(p+1)c_j') \hfill \mbox{ (by \(c'_j=\max(c'_0,\ldots,c'_k)\))}\\
= 2^{f'_n+p+c_j} -(f'_n+(p+1)c'_j)\\
= 2^{f'_n+p}\cdot 2^{c_j} - (f'_n+(p+1)c'_j)\\
\geq 2^{f'_n+p}\cdot c'_j - (f'_n+(p+1)c'_j)
\hfill \mbox{ (by \(2^{c_j}\geq c'_j\))}
\\
\geq (f'_n+p+1)\cdot c'_j - (f'_n+(p+1)c'_j)
\hfill \mbox{ (by \(2^x \geq x+1\))}
\\
= f'_n(c'_j-1)
\geq 0.
\hfill \mbox{ (by \(c'_j\geq 1\))}
\end{array}
\]
\end{itemize}
\item Case \rname{PTr-Const}: Similar to the case for \rname{PTr-App} above.
\end{itemize}
\end{proof}
 
\subsubsection{Existence of Pumpable Derivation}
\label{sec:pumpability}
We are now ready to prove Lemma~\ref{lem:pumpable-derivation}.
\begin{proof}[Proof of Lemma~\ref{lem:pumpable-derivation}]
Let \(\GRAM\) be an order-\(n\) tree grammar and \(S\) be its start symbol.
Suppose that \(\Lang(\GRAM)\) is infinite.
By Theorem~\ref{th:completeness-tr}, for any \(c'\), 
there exist \(\dt\), \(c \ge c'\), and \(s\) such that
\(\dt \pdt \emptyset \PM S:(\emptyset,\set{0,\ldots,n-1},\T)\hascost c\tr s\),
and \(\toLty{\dt}\) is an admissible derivation.
We can assume that,
in each path of the derivation tree \(\dt\),
there do not exist two judgments of the form \(\PTE\PM t:\pty\hascost c''\tr s'\) 
and \(\PTE\PM t:\pty\hascost c''\tr s''\) in different positions,
since otherwise we can ``shrink'' that part without changing the conclusion.
For every \(c'\), \(\dt\) can contain only a subterm of the right-hand-side of a rule in \(\GRAM\)
and the number of such subterms are bounded above by a constant determined by \(\GRAM\).
Also the numbers of possible types and type environments in \(\dt\) are bounded above.
Further note that, flag counters \(c''\) occurring in \(\dt\) are not essential information;
they can be recovered from types in \(\dt\) and the tree structure of \(\dt\).
Therefore, since the number of premises occurring at each node of \(\dt\) is bounded above,
the height of the derivation trees \(\dt\) must be unbounded.

Then, for sufficiently large \(c'\),
the derivation tree \(\dt\) must contain a path in which 
three type judgments
\(\PTE \PM A:\pty\hascost{c'_i}\tr s_i\) \((i=0,1,2)\) occur
where \(c'_2>c'_1>c'_0\geq 0\).
This is shown as follows:
(i) Let \(a= a_{\mathrm{term}} \times (a_{\mathrm{env}} + n) \)
where 
\(a_{\mathrm{term}}\) is the maximum size of terms occurring in all \(\dt\)
and
\(a_{\mathrm{env}}\) is the maximum length of type environments in all \(\dt\).
Then in every path in \(\dt\), we encounter the rule \rname{PTr-NT}
in every \(a\)-steps,
since for a rule other than \rname{PTr-Weak}, \rname{PTr-Mark}, and \rname{PTr-NT},
the size of the term decreases,
and for rule \rname{PTr-Weak} or \rname{PTr-Mark},
the size of the term is kept and the sum of
the length of the type environment
and
the number of the markers (bounded by \(n\))
decreases (note that \(M\) in \rname{PTr-Mark} must be non-empty by the above assumption on \(\dt\)).
(ii) Let \(b_{te}\) be the number of type environments in all \(\dt\),
\(b_{\mathrm{nt}}\) be the number of non-terminals of \(\GRAM\),
and 
\(b_{t}\) be the number of types in all \(\dt\).
Then if the height of \(\dt\) is larger than \(a (2 b_{\mathrm{te}} b_{\mathrm{nt}} b_{\mathrm{t}} + 1)\),
there exist \(\PTE\), \(A\), and \(\pty\) such that
in a longest path of \(\dt\) there exist three different occurrences of judgments of the form \(\PTE \PM A:\pty\hascost{c'_i}\tr s_i\)
(\(i=0,1,2\)) with \(c'_2 \ge c'_1 \ge c'_0\geq 0\).
(iii) By the above assumption on \(\dt\), \(c'_0\), \(c'_1\), and \(c'_2\) must differ from each other.

By the transformation rules, \(s\) and \(s_2\) must be of the forms \(C[s_2]\) and \(D[s_1]\), respectively.
Furthermore, since \(c \ (\ge c'_2) > 0\) and \(\toLty{\dt}\) is an admissible derivation, we have \(\emptyset \pLin C[D[s_1]]:\T^1\).
Moreover, since \(c'_1>0\),
the flag counters of all
judgments in the path between \(\PTE \PM A:\pty\hascost{c'_1}\tr s_1\) and \(\emptyset \PM S:(\emptyset,\set{0,\ldots,n-1},\T)\hascost c\tr s\)
are non-zero. Thus, by the admissibility of
\(\toLty{\dt}\), 
linear types are assinged to the corresponding terms containing \(s_1\) (hence also those containing \(D[s_1]\))
in the derivation \(\toLty{\dt}\) of \(\emptyset\pLin C[D[s_1]]:\T^1\).
Thus, by Lemma~\ref{lem:linearity}, 
each of the contexts \(C\) and \(C[D]\) is linear.
In \(\dt\), the part corresponding to \(D\) is pumpable, so we can show in the same way that \(C[D^n]\) is
 lienar for any \(n\).
Let \(c_1 = c'_1\), \(c_2 = c'_2-c'_1\), and \(c_3 = c-c'_2\); then we have obtained the required result.
\end{proof}
\anp
\subsection{Proof of Lemma~\ref{lem:soundness-tr}}
\label{sec:reduction}
Similarly to~\cite{Parys16ITRS}, the soundness (Lemma~\ref{lem:soundness-tr}) is proved
in three steps (Lemmas~\ref{lem:subRedParys},~\ref{lem:subRed2Parys}, and~\ref{lem:soundTreeParys} below).
For \(n\) and a derivation tree \(\dt\), 
we write \(\ac{n}{\dt}\) for the number of order-\(n\) \rname{PTr-App} used in \(\dt\),
and write \(\ntc{\dt}\) for the number of \rname{PTr-NT} used in \(\dt\).

\begin{lemma}[substitution lemma]
\label{lem:substParys}
Given
\begin{align*}&
\dt_0 \pdt \PTE_0,\, x\COL \set{(F_1,M_1,\prty_1),\ldots,(F_k,M_k,\prty_k)} \PM t:(F_0,M_0,\prty) \hascost{c_0}\tr s
\\& 
\dt_i \pdt \PTE_{i}\PM \term_1:(F_{i},M_{i},\prty_i)\hascost{c_{i}}\tr s_{i}
\qquad\mbox{for each $i\in\set{1,\ldots,k}$}
\\&
\PTE_{0}+(\textstyle\sum_{i\in\set{1,\ldots,k}}\PTE_{i})
\text{ is well-defined}
\\&
\ac{n}{\dt_i} = 0 
\qquad\text{for each $i\in\set{0,1,\ldots,k}$}
\\&
\eorder(t_1) < n
\end{align*}
there exists \(\dt\) such that \(\ac{n}{\dt} = 0\) and
\begin{align*}&
\dt \pdt \PTE_0\PTEcup (\textstyle\sum_{i\in\set{1,\ldots,k}} \PTE_{i}) \PM [\term_1 / x] t:
  (F_0, M_0, \prty)
\\&\qquad
\hascost{(c_0 + \textstyle\sum_{i\in\set{1,\ldots,k}}c_{i})}\tr
[s_{i} / x_{(F_i,M_i,\prty_i)}]_{i \in \set{1,\dots, k}}s.
\end{align*}
\end{lemma}
\begin{proof}
The proof proceeds by induction on \(\dt_0\)
with case analysis on the rule used last for \(\dt_0\).

\begin{itemize}
\item Case of \rname{PTr-Weak}: Clear.

\item Case of \rname{PTr-Mark}:
Let the last rule of \(\dt_0\) be:\\[5pt]
\infer
{
\PTE_0,\, x\COL \set{(F_1,M_1,\prty_1),\ldots,(F_k,M_k,\prty_k)} \PM t:(F_0, M_0,\prty )\hascost{c_0}\tr s
}{
\begin{gathered}
\dt'_0 \pdt \PTE_0,\, x\COL \set{(F_1,M_1,\prty_1),\ldots,(F_k,M_k,\prty_k)} \PM t:(F',M',\prty)\hascost{c'}\tr s\\
M' \subseteq M_0
\andalso
M_0 \setminus M' \subseteq \set{j \mid \eorder(t)\leq j<n}
\andalso   
\Comp_{n}(\set{(F',c')},M_0)= (F_0,c_0)
\end{gathered}
}\\[5pt]
Then by induction hypothesis for \(\dt'_0\), we have
\begin{align*}
\dt' \pdt \ &\PTE_0\PTEcup (\textstyle\sum_{i\in\set{1,\ldots,k}} \PTE_{i})
\PM [\term_1 / x] t: (F', M', \prty)
\\
&\hascost{(c' + \textstyle\sum_{i\in\set{1,\ldots,k}}c_{i})}\tr
[s_{i} / x_{(F_i,M_i,\prty_i)}]_{i \in \set{1,\dots, k}}s
\end{align*}
with \(\ac{n}{\dt'}=0\).

Since \(\eorder([t_1/x]t) = \eorder(t)\),
\[
M_0 \setminus M' \subseteq \set{j \mid \eorder([t_1/x]t)\leq j<n},
\]
and we can check that \(\Comp_{n}(\set{(F',c')},M_0)= (F_0,c_0)\) implies
\[
\Comp_{n}(\set{(F',c' + \textstyle\sum_{i\in\set{1,\ldots,k}}c_{i})},M_0)
= (F_0,c_0 + \textstyle\sum_{i\in\set{1,\ldots,k}}c_{i})
\]
by calculation of \(\Comp_{n}\).
Hence by \rname{PTr-Mark}, we have
\[
\dt \defe \frac{\dt'}{
\begin{aligned}&
\PTE_0\PTEcup (\textstyle\sum_{i\in\set{1,\ldots,k}} \PTE_{i}) \PM [\term_1 / x] t:
  (F_0, M_0, \prty)
\\&\qquad
\hascost{(c_0 + \textstyle\sum_{i\in\set{1,\ldots,k}}c_{i})}\tr
[s_{i} / x_{(F_i,M_i,\prty_i)}]_{i \in \set{1,\dots, k}}s
\end{aligned}
}
\]
\(\ac{n}{\dt}=\ac{n}{\dt'}=0\), as required.

\item Case of \rname{PTr-Var}:
In this case, we further perform case analysis on whether the variable \(t\) is \(x\) or not.
In the case \(t=x\),
let the last rule of \(\dt_0\) be:
\infrule{}{x\COL \pty\PM x\COL \pty\hascost{0}\tr x_{\pty}}
and then
\begin{align*}
\PTE_0 &= \eset
\\
k &= 1
\\
(F_{1},M_{1},\prty_{1}) &= \pty
\\
(F_0,M_0,\prty) &= \pty
\\
c_0 &= 0
\\
s &= x_{\pty}.
\end{align*}
Now the goal is
\[
\PTE_{1} \PM \term_1 : (F_{1},M_{1},\prty_1)\hascost{c_{1}}\tr s_{1}\,,
\]
which is just \(\dt_1\).

Next, in the case \(t=y \neq x\),
let the last rule of \(\dt_0\) be:
\infrule{}{y\COL \pty\PM y\COL \pty\hascost{0}\tr y_{\pty}}
and then
\begin{align*}
\PTE_0 &= y\COL \pty
\\
k &= 0
\\
(F_0,M_0,\prty) &= \pty
\\
c_0 &= 0
\\
s &= y_{\pty}.
\end{align*}
Now the goal is
\[
y\COL \pty
\PM y :
\pty \hascost{0}\tr y_{\pty}\,,
\]
which is just \(\dt_0\) as above.

\item Case of \rname{PTr-Choice}: Clear.

\item Case of \rname{PTr-Abs}:
Let the last rule of \(\dt_0\) be:
\infrule{
\dt'_0 \pdt \PTE_0,\, x\COL \set{(F_1,M_1,\prty_1),\ldots,(F_k,M_k,\prty_k)}, x'\COL \pity\PM t':(F_0,M,\prty') \hascost{c_0}\tr s'
}{
\PTE_0,\, x\COL \set{(F_1,M_1,\prty_1),\ldots,(F_k,M_k,\prty_k)}
\PM \lambda x'.t': (F_0, M\setminus \Markers(\pity),\pity\to \prty')\hascost{c_0}\tr
\lambda \seq{x'}_\pity.s' }
where \(x'\) is fresh, and then we have
\begin{align*}
t &= \lambda x'.t'
\\
M_0 &= M\setminus \Markers(\pity)
\\
\prty &= \pity\to \prty'
\\
s &= \lambda \seq{x'}_\pity.s'.
\end{align*}

Since \(x'\) was chosen as a fresh variable,
\((\PTE_{0},\, x'\COL \pity) + (\textstyle\sum_{i\in\set{1,\ldots,k}}\PTE_{i})\)
is well-defined.
Hence, by induction hypothesis for \(\dt'_0\), we have
\begin{align*}
\dt' \pdt \ &(\PTE_0,\, x'\COL \pity) \PTEcup (\textstyle\sum_{i\in\set{1,\ldots,k}} \PTE_{i}) \PM [\term_1 / x] t':
  (F_0, M, \prty')
\\&
\hascost{(c_0 + \textstyle\sum_{i\in\set{1,\ldots,k}}c_{i})}\tr
[s_{i} / x_{(F_i,M_i,\prty_i)}]_{i \in \set{1,\dots, k}}s'.
\end{align*}
with \(\ac{n}{\dt'} = 0\).
Let \(\dt\) be \(\dt'\) plus \rname{PTr-Abs}, so that we have \(\ac{n}{\dt}=\ac{n}{\dt'}=0\) and
\begin{align*}
\dt \pdt \ &\PTE_0 \PTEcup (\textstyle\sum_{i\in\set{1,\ldots,k}} \PTE_{i})
  \PM \lambda x'.[t_1/x]t':
 (F_0, M\setminus \Markers(\pity),\pity\to \prty')
\\&
\hascost{(c_0 + \textstyle\sum_{i\in\set{1,\ldots,k}}c_{i})}\tr
\lambda \seq{x'}_\pity.[s_{i} / x_{(F_i,M_i,\prty_i)}]_{i \in \set{1,\dots, k}}s'
\end{align*}
i.e.,
\begin{align*}
\dt \pdt \ &\PTE_0 \PTEcup (\textstyle\sum_{i\in\set{1,\ldots,k}} \PTE_{i})
 \PM [t_1/x](\lambda x'.t'):
 (F_0, M_0, \prty)
\\&
\hascost{(c_0 + \textstyle\sum_{i\in\set{1,\ldots,k}}c_{i})}\tr
[s_{i} / x_{(F_i,M_i,\prty_i)}]_{i \in \set{1,\dots, k}}(\lambda x'_{\pty_1}.\cdots\lambda x'_{\pty_k}.s'),
\end{align*}
as required.

\item Case of \rname{PTr-App}:
Let the last rule of \(\dt_0\) be:
\infrule{
\eorder(t'_0)=\ell'\\
\dt'_0 \pdt \PTE'_0\PM t'_0:(F'_0,M'_0,\set{(F'_1,M'_1,\prty'_1),\ldots,
(F'_{k'},M'_{k'},\prty'_{k'})}\to \prty)\hascost{c'_0}\tr s'_0
\\
\mspace{-8mu}\left.
\begin{aligned}&
\dt'_{i'} \pdt \PTE'_{i'}\PM t'_1:(F''_{i'},M''_{i'},\prty'_{i'})\hascost{c'_{i'}}\tr s'_{i'}
\\&
F''_{i'}\restrict{<\ell'}=F'_{i'}
\qquad
M''_{i'}\restrict{<\ell'}=M'_{i'}
\end{aligned}
\right\}
\mbox{for each $i'\in\set{1,\ldots,{k'}}$}
\\
  M_0 = M'_0\uplus (\biguplus_{i'\in\set{1,\ldots,{k'}}} M''_{i'})\\
  \Comp_{n}(\set{(F'_0,c'_0)}\mcup \set{(F''_{i'}\restrict{\geq \ell'},c'_{i'})\mid i'\in\set{1,\ldots,{k'}}}, M_0)
   = (F_0, c_0)\\
(F'_1,M'_1,\prty'_1)<\cdots< (F'_{k'},M'_{k'},\prty'_{k'})
}
{\PTE'_0\PTEcup (\sum_{i'\in\set{1,\ldots,{k'}}} \PTE'_{i'}) \PM t'_0t'_1:
  (F_0, M_0, \prty)\hascost{c_0}\tr s'_0\, s'_{1}\,\cdots\,s'_{{k'}}}
and then we have
\begin{align*}&
\PTE_0,\, x\COL \set{(F_1,M_1,\prty_1),\ldots,(F_k,M_k,\prty_k)} = 
\PTE'_0\PTEcup (\textstyle\sum_{i'\in\set{1,\ldots,{k'}}} \PTE'_{i'})
\\&
t = t'_0t'_1
\\&
s = s'_0\, s'_{1}\,\cdots\,s'_{{k'}}.
\end{align*}
Note that \(\ell' < n\) since \(\ac{n}{\dt_0}=0\).

Let:
\begin{align*}
\PTE'_0 &=
\PTE''_0,\, x\COL \pity_{0}
\\
\PTE'_{i'} &=
\PTE''_{i'},\, x\COL \pity_{i'}
\quad\text{for each \(i'\in\set{1,\ldots,{k'}}\)}
\end{align*}
where \(x \notin \dom(\PTE''_0)\) and \(x \notin \dom(\PTE''_{i'})\); then
\begin{align}&
\PTE_0 = \PTE''_0 \PTEcup (\textstyle\sum_{i'\in\set{1,\ldots,{k'}}} \PTE''_{i'})
\label{eq:subLemAppEnvDiv}
\\&
\set{(F_1,M_1,\prty_1),\ldots,(F_k,M_k,\prty_k)} = 
\pity_0 \cup (\textstyle\bigcup_{i'\in\set{1,\ldots,{k'}}} \pity_{i'}).
\notag
\end{align}
Also, we define
\begin{align*}
I_0 &\defe \set{ i \in \set{1,\dots,k} \mid (F_{i},M_{i},\prty_{i}) \in \pity_0 }
\\
I_{i'} &\defe \set{ i \in \set{1,\dots,k} \mid (F_{i},M_{i},\prty_{i}) \in \pity_{i'} }
\quad\text{for each \(i'\in\set{1,\ldots,{k'}}\)}.
\end{align*}
Then \(\set{1,\dots,k} = I_0 \cup (\textstyle\bigcup_{i'\in\set{1,\ldots,{k'}}} I_{i'})\) and
\[
\textstyle\sum_{i \in I} \PTE_{i}
=
(\textstyle\sum_{i \in I_0} \PTE_{i}) 
+ \textstyle\sum_{i'\in\set{1,\ldots,{k'}}} 
(\textstyle\sum_{i\in I_{i'}} \PTE_{i})
\]
where note that the right hand side is well-defined.

Now we have
\begin{align*}&
\dt'_0 \pdt \PTE''_0,\, x\COL \pity_0 \PM t'_0 : (F'_0,M'_0, \set{(F'_1,M'_1,\prty'_1),\ldots,(F'_{k'},M'_{k'},\prty'_{k'})} \to \prty)
\hascost{c'_0}\tr s'_0
\\& 
\dt_i \pdt \PTE_{i}\PM \term_1:(F_{i},M_{i},\prty_i)\hascost{c_{i}}\tr s_{i}
\quad\mbox{for each $i\in I_0$}
\\&
\PTE''_0 + (\textstyle\sum_{i\in I_0} \PTE_{i})
\text{ is well-defined (by~\eqref{eq:subLemAppEnvDiv} and the assumption)}
\\&
\ac{n}{\dt_i} = 0 \quad\mbox{for each $i\in I_0$}
\qquad\quad
\eorder(t_1) < n.
\end{align*}
Hence by induction hypothesis we have
\begin{align*}&
\dt''_0 \pdt \PTE''_0 \PTEcup (\textstyle\sum_{i\in I_0} \PTE_{i}) \PM [\term_1 / x] t'_0:
(F'_0, M'_0, \set{(F'_1,M'_1,\prty'_1),\ldots,(F'_{k'},M'_{k'},\prty'_{k'})} \to \prty)
\\&\qquad
\hascost{c'_0 + (\textstyle\sum_{i \in I_0}c_{i})}  \tr
[s_{i} / x_{(F_i,M_i,\prty_i)}]_{i \in I_0}s'_0.
\end{align*}
with \(\ac{n}{\dt''_0} = 0\).
Also, for each $i'\in\set{1,\ldots,{k'}}$, we have
\begin{align*}&
\dt'_{i'} \pdt \PTE''_{i'},\, x\COL \pity_{i'} \PM t'_1:(F''_{i'},M''_{i'},\prty'_{i'})\hascost{c'_{i'}}\tr s'_{i'}
\\& 
\dt_i \pdt \PTE_{i}\PM \term_1:(F_{i},M_{i},\prty_i)\hascost{c_{i}}\tr s_{i}
\quad\mbox{for each $i\in I_{i'}$}
\\&
\PTE''_{i'} + (\textstyle\sum_{i\in I_{i'}} \PTE_{i})
\text{ is well-defined (by~\eqref{eq:subLemAppEnvDiv} and the assumption)}
\\&
\ac{n}{\dt_i} = 0 \quad\mbox{for each $i\in I_0$}
\qquad\quad
\eorder(t_1) < n.
\end{align*}
Hence by induction hypothesis we have
\begin{align*}&
\dt''_{i'} \pdt \PTE''_{i'} \PTEcup (\textstyle\sum_{i\in I_{i'}} \PTE_{i}) 
\PM [\term_1 / x] t'_1: (F''_{i'}, M''_{i'}, \prty'_{i'})
\\&\qquad
\hascost{c'_{i'} + (\textstyle\sum_{i\in I_{i'}} c_{i})}\tr
[s_{i} / x_{(F_i,M_i,\prty_i)}]_{i \in I_{i'}}s'_{i'}.
\end{align*}
with \(\ac{n}{\dt''_{i'}} = 0\).

Now we have:
\begin{align*}&
\eorder([t_1/x]t'_0)\ (= \eorder(t'_0)) = \ell'
\\&
\dt''_0 \pdt \PTE''_0 \PTEcup (\textstyle\sum_{i\in I_0} \PTE_{i})
\PM [t_1/x]t'_0:(F'_0,M'_0,\set{(F'_1,M'_1,\prty'_1),\ldots,
(F'_{k'},M'_{k'},\prty'_{k'})}\to \prty)
\\&
\qquad\ \ \hascost{c'_0 + (\textstyle\sum_{i \in I_0}c_{i})}
\tr [s_{i} / x_{(F_i,M_i,\prty_i)}]_{i \in I_0}s'_0
\\&
\mspace{-8mu}\left.
\begin{aligned}&
\dt''_{i'} \pdt \PTE''_{i'} \PTEcup (\textstyle\sum_{i\in I_{i'}} \PTE_{i}) 
\PM [\term_1 / x]t'_1:(F''_{i'},M''_{i'},\prty'_{i'})
\\&\qquad\ \
\hascost{(c'_{i'} + (\textstyle\sum_{i\in
 I_{i'}} c_{i}))}\tr 
[s_{i} / x_{(F_i,M_i,\prty_i)}]_{i \in I_{i'}}s'_{i'}
\\&
F''_{i'}\restrict{<\ell'}=F'_{i'}
\andalso
M''_{i'}\restrict{<\ell'}=M'_{i'}
\end{aligned}
\right\}
\mbox{for each } i'\in\set{1,\ldots,{k'}}
\\&
  M_0 = M'_0\uplus (\textstyle\biguplus_{i'\in\set{1,\ldots,{k'}}} M''_{i'})
\\&
\begin{aligned}
  \Comp_{n}\big(&\set{(F'_0, c'_0 + (\textstyle\sum_{i \in I_0}c_{i}))}
\mcup
\\
& \set{(F''_{i'}\restrict{\geq \ell'},
c'_{i'} + (\textstyle\sum_{i\in I_{i'}} c_{i}) )
\mid i'\in\set{1,\ldots,{k'}}},
\\
& M_0\big)   = (F_0, 
c_0  + (\textstyle\sum_{i \in I_0}c_{i})
 + \textstyle\sum_{i'\in\set{1,\ldots,{k'}}} (\textstyle\sum_{i\in I_{i'}} c_{i})
)
\end{aligned}
\\&
(F'_1,M'_1,\prty'_1)<\cdots< (F'_{k'},M'_{k'},\prty'_{k'})
\end{align*}
where the equation on \(\Comp_{n}\) can be easily checked by a direct calculation.
Hence we can use \rname{PTr-App}, by which we obtain \(\dt\) with
\(\ac{n}{\dt} = 0\) since \(\ell' < n\);
and we have
\begin{align*}
\dt \pdt \ &\PTE''_0 + (\textstyle\sum_{i \in I_0} \PTE_{i}) 
+ (\textstyle\sum_{i'\in\set{1,\ldots,{k'}}} 
(\PTE''_{i'} \PTEcup (\textstyle\sum_{i\in I_{i'}} \PTE_{i})) )
\\&
\PM
 ([t_1/x]t'_0)([t_1/x]t'_1):
  (F_0, M_0, \prty)
\\&
\hascost{
\big(c_0  + (\textstyle\sum_{i \in I_0}c_{i})
 + \textstyle\sum_{i'\in\set{1,\ldots,{k'}}} (\textstyle\sum_{i\in I_{i'}} c_{i})\big)
}
\\& \tr
([s_{i} / x_{(F_i,M_i,\prty_i)}]_{i \in I_0}s'_0)\, 
([s_{i} / x_{(F_i,M_i,\prty_i)}]_{i \in I_{1}}s'_{1})\,\cdots\,
([s_{i} / x_{(F_i,M_i,\prty_i)}]_{i \in I_{k'}}s'_{{k'}})
\end{align*}
i.e., we have
\begin{align*}
\dt \pdt \ &\PTE_0\PTEcup (\textstyle\sum_{i\in\set{1,\ldots,k}} \PTE_{i}) \PM [\term_1 / x] t:
(F_0, M_0, \prty)
\\&
  \hascost{(c_0 + \textstyle\sum_{i\in\set{1,\ldots,k}}c_{i})}\tr
[s_{i} / x_{(F_i,M_i,\prty_i)}]_{i \in \set{1,\dots, k}}s
\end{align*}
as required,
where note that
\[
(\textstyle\sum_{i \in I_0}c_{i})
 + \textstyle\sum_{i'\in\set{1,\ldots,{k'}}} (\textstyle\sum_{i\in I_{i'}} c_{i})
=
\textstyle\sum_{i\in\set{1,\ldots,k}}c_{i}
\]
follows from Lemma~\ref{lem:marker} applied to \(\PTE_{i}\PM \term_1:(F_{i},M_{i},\prty_i)\hascost{c_{i}}\tr s_{i}\), and
\begin{align*}&
([s_{i} / x_{(F_i,M_i,\prty_i)}]_{i \in I_0}s'_0)
=
([s_{i} / x_{(F_i,M_i,\prty_i)}]_{i \in \set{1,\dots, k}}s'_0)
\\&
([s_{i} / x_{(F_i,M_i,\prty_i)}]_{i \in I_{i'}}s'_{i'})
=
([s_{i} / x_{(F_i,M_i,\prty_i)}]_{i \in \set{1,\dots, k}}s'_{i'})
\qquad (i'\in\set{1,\ldots,{k'}})
\end{align*}
follow from Lemma~\ref{lem:PTrFV}.

\item Case of \rname{PTr-Const}:
This case is analogous to the case of \rname{PTr-App}.

\item Case of \rname{PTr-NT}:
This case is clear (note that \(t'\) is a closed term for \(A = t' \in \GRAM\)).
\end{itemize}
\end{proof}

\begin{lemma}
\label{lem:SRandProg}
Suppose 
\begin{align*}&
n>0
\\&
\dt \pdt \PTE \PM (\lambda x . t')\,t_1: (F,M,\prty) \hascost c\tr s
\\&
\text{the last rule of \(\dt\) is {\rm \rname{PTr-App}}}
\\&
\eorder(\lambda x . t') = n
\\&
\ac{n}{\dt} = 1
\\&
\order(\pty') < n \text{ for any } (x' \COL \pty') \in\PTE.
\end{align*}
Then there exist \(\PTE'\), \(v\), and \(\dt'\) %
such that
\begin{align*}&
\PTE \subPTE \PTE'
\\&
s \reds v
\\&
\dt' \pdt \PTE' \PM [t_1/x]t' : (F,M,\prty) \hascost c\tr v
\\&
\order(\pty') < n \text{ for any } (x' \COL \pty') \in\PTE'
\\&
\ac{n}{\dt'}=0
\end{align*}
\end{lemma}
\begin{proof}
Since the last rule of \(\dt\) is \rname{PTr-App}, we have:
\begin{align}&
\dt_0 \pdt \PTE_0\PM \lambda x . t' :(F_0,M_0,\set{\pty_1,\ldots,\pty_k} \to \prty)\hascost{c_0}\tr s_0
\label{eq:SRPRGsubRedApp2}
\\&
\pty_i = (F_i,M_i,\prty_i)  \text{ for each \(i \in \set{1,\dots,k}\)}
\label{eq:SRPRGsubRedApp3}
\\&
\dt_i \pdt \PTE_{i}\PM \term_1:(F_{i},M_{i},\prty_i)\hascost{c_{i}}\tr s_{i}
\mbox{ for each $i\in\set{1,\ldots,k}$}
\label{eq:SRPRGsubRedApp4}
\\&
  M = M_0\uplus (\textstyle\biguplus_{i\in\set{1,\ldots,k}} M_{i})
\label{eq:SRPRGsubRedApp12}
\\&
  \Comp_{n}(\set{(F_0,c_0)}\mcup \set{(\emptyset,c_{i})\mid i\in\set{1,\ldots,k}}, M)
   = (F, c) 
\label{eq:SRPRGsubRedApp7}
\\&
\pty_1 < \cdots < \pty_k
\label{eq:SRPRGsubRedApp8}
\\&
\PTE = \PTE_0\PTEcup (\textstyle\sum_{i\in\set{1,\ldots,k}} \PTE_{i})
\label{eq:SRPRGsubRedApp9}
\\&
s = s_0\, s_{1}\,\cdots\,s_{k}
\label{eq:SRPRGsubRedApp10}
\end{align}
We can assume that~\eqref{eq:SRPRGsubRedApp2} is not derived by \rname{PTr-Mark}, since \(\eorder(\lambda x . t')=n\).
Hence~\eqref{eq:SRPRGsubRedApp2} is derived by \rname{PTr-Abs} possibly with ~\rname{PTr-Weak}; thus we have
\begin{align}&
\overline{\dt}_0 \pdt
\PTE'_0,\, x\COL \set{\pty_1,\ldots,\pty_k}\PM t':(F_0,M'_0,\prty) \hascost{c_0}\tr s'
\label{eq:SRPRGsubRedAbs1}
\\&
\PTE_0 \subPTE \PTE'_0
\label{eq:SRPRGsubRedAbs2}
\\&
M_0 = M'_0\setminus (\textstyle\biguplus_{i\in\set{1,\ldots,k}} M_i)
\label{eq:SRPRGsubRedAbs3}
\\&
s_0 = \lambda x_{\pty_1}.\cdots\lambda x_{\pty_k}.s'.
\label{eq:SRPRGsubRedAbs4}
\end{align}

By applying Lemma~\ref{lem:markContextType} to~\eqref{eq:SRPRGsubRedAbs1}, 
we have
\[
\textstyle\biguplus_{i\in\set{1,\ldots,k}} M_i = \Markers(\set{\pty_1,\ldots,\pty_k}) \subseteq \Markers(\PTE'_0,\, x\COL \set{\pty_1,\ldots,\pty_k}) \subseteq M'_0.
\]
By this and~\eqref{eq:SRPRGsubRedAbs3},
we have \(M'_0 = M_0 \uplus (\textstyle\biguplus_{i\in\set{1,\ldots,k}} M_i)\).
Then, by~\eqref{eq:SRPRGsubRedApp12} we have \(M=M'_0\). 
By this and~\eqref{eq:SRPRGsubRedAbs1}, we have \(F_0 \cap M = F_0 \cap M'_0 = \eset\).
Hence, by applying Lemma~\ref{lem:empDisjComp} to~\eqref{eq:SRPRGsubRedApp7}, we have
\begin{align}&
F = F_0
\qquad
c = c_0 + \textstyle\sum_{i\in\set{1,\ldots,k}} c_{i}.
\label{eq:SRPRGsubRed13}
\end{align}
Thus,~\eqref{eq:SRPRGsubRedAbs1} is equal to:
\begin{equation}
\label{eq:SRPRGsubRedAbs1a}
\overline{\dt}_0 \pdt
\PTE'_0,\, x\COL \set{\pty_1,\ldots,\pty_k}\PM t':(F,M,\prty) \hascost{c_0}\tr s'.
\end{equation}

By~\eqref{eq:SRPRGsubRedApp9} and~\eqref{eq:SRPRGsubRedAbs2}, 
\(\PTE'_0 \PTEcup (\textstyle\sum_{i\in\set{1,\ldots,k}} \PTE_{i})\) is well-defined.
Hence, by applying Lemma~\ref{lem:substParys} to~\eqref{eq:SRPRGsubRedApp4} and~\eqref{eq:SRPRGsubRedAbs1a},
we have
\begin{align*}&
\dt' \pdt
\PTE'_0\PTEcup (\textstyle\sum_{i\in \set{1,\ldots,k}} \PTE_{i}) \PM [\term_1 / x] t':
\\&\qquad
(F, M, \prty)\hascost{c_0+\textstyle\sum_{i\in \set{1,\ldots,k}}c_{i}}\tr [s_{i} / x_{\pty_i}]_{i \in \set{1,\ldots,k}}s'
\end{align*}
and \(\ac{n}{\dt'} = 0\).
Now we define \(\PTE' \defe \PTE'_0\PTEcup (\textstyle\sum_{i\in \set{1,\ldots,k}} \PTE_{i})\)
and \(v \defe [s_{i} / x_{\pty_i}]_{i \in \set{1,\ldots,k}}s'\);
then \(\dt' \pdt \PTE' \PM u: (F,M,\prty) \hascost c\tr v\)
by~\eqref{eq:SRPRGsubRed13}.
By~\eqref{eq:SRPRGsubRedApp9} and~\eqref{eq:SRPRGsubRedAbs2}, we have \(\PTE \subPTE \PTE'\),
which implies \(\order(\pty') < n \text{ for any } (x' \COL \pty') \in\PTE'\).
Also we have
\(
s = (\lambda x_{\pty_1}.\cdots\lambda x_{\pty_k}.s')\, s_{1}\,\cdots\,s_{k}
\reds [s_{i} / x_{\pty_i}]_{i \in \set{1,\ldots,k}}s'
\)
by~\eqref{eq:SRPRGsubRedApp10} and~\eqref{eq:SRPRGsubRedAbs4}.
\end{proof}

For \(\dt \pdt \PTE \PM t: \pty \hascost c\tr s\) such that 
\(\ac{n}{\dt} > 0\) and \(\order(\pty') < n\) for any \(x \COL \pty' \in\PTE\),
we define \emph{\(\dt\)-evaluation context} \(\E{\dt}\) 
and \emph{\(\dt\)-redex} \(r^{\dt}\) such that 
(i) \(t = \E{\dt}[r^{\dt}]\) %
and 
(ii) either \(r^{\dt}\) is of the form \((\lambda x.t_0) t_1\) where \(\eorder(\lambda x.t_0) =n\),
or \(r^{\dt}\) is a non-terminal \(A\).
These are defined by induction on \(\dt\) and by case analysis on the rule used last:
\begin{itemize}
\item Cases of \rname{PTr-Weak} and \rname{PTr-Mark}:
Let \(\dt_0\) be the subderivation of \(\dt\).
Then \(\E{\dt} \defe \E{\dt_0}\) and \(r^{\dt} \defe r^{\dt_0}\)

\item Case of \rname{PTr-Var}: This case does not happen by the assumption.

\item Case of \rname{PTr-Choice}:
Let \(t=t_1+t_2\),\, \(\dt_0\) be the subderivation of \(\dt\),
and the root term of \(\dt_0\) be \(t_i\).
Then \(\E{\dt} \defe \E{\dt_0} + t_2 \ (\text{if } i=1)\), \(t_1 + \E{\dt_0} \ (\text{if } i=2)\), and \(r^{\dt} \defe r^{\dt_0}\).

\item Case of \rname{PTr-Abs}:
Let \(t= \lambda x . t'\) and \(\dt_0\) be the subderivation of \(\dt\).
Then \(\E{\dt} \defe \lambda x . \E{\dt_0}\) and \(r^{\dt} \defe r^{\dt_0}\).

\item Case of \rname{PTr-App}:
Suppose that \(t=t_0\,t_1\) and we have the following as the premises and a side condition of the last rule:
\begin{align*}&
\dt_0 \pdt \PTE_0\PM \term_0:(F_0,M_0,\set{(F_1,M_1,\prty_1),\ldots,
(F_k,M_k,\prty_k)}\to \prty)\hascost{c_0}\tr s_0
\\&
\dt_i \pdt \PTE_{i}\PM \term_1:(F_{i}',M_{i}',\prty_i)\hascost{c_{i}}\tr s_{i}
\quad\mbox{for each $i\in\set{1,\ldots,k}$}
\\&
(F_1,M_1,\prty_1)<\cdots< (F_k,M_k,\prty_k)
\end{align*}

When \(\ac{n}{\dt_i} > 0\) for some \(i=1,\dots,k\), 
let \(i_0\) be the largest such \(i\).
We define \(\E{\dt} \defe t_0\,\E{\dt_{i_0}}\) and \(r^{\dt} \defe r^{\dt_{i_0}}\).

Otherwise, if \(\ac{n}{\dt_0} > 0\),
then we define \(\E{\dt} \defe \E{\dt_{0}}\,t_1\) and \(r^{\dt} \defe r^{\dt_{0}}\).

Otherwise, we have \(\eorder(t_0)=n\). 
Since \(\ac{n}{\dt_0} = 0\), \(t_0\) is not an application term.
Also, \(t_0\) is not a variable by the assumption.
If \(t_0\) is a non-terminal, we define \(\E{\dt} \defe \Hole\,t_1\) and \(r^{\dt} \defe t_0\).
If \(t_0\) is a \(\lambda\)-abstraction,
then we define \(\E{\dt} \defe \Hole\) and \(r^{\dt} \defe t\).

\item Case of \rname{PTr-Const}:
Suppose that \(t=a\,t_0 \cdots t_k\) and 
let \(\dt_i\) be the subderivation of \(\dt\) whose root term is \(t_i\).
Let \(i_0\) be the largest \(i\) such that \(\ac{n}{\dt_i}>0\).
Then we define \(\E{\dt} \defe a\,t_0 \cdots t_{i_0 - 1} \E{\dt_{i_0}} t_{i_0 + 1} \cdots t_k\) 
and \(r^{\dt} \defe r^{\dt_{i_0}}\).
\item Case of \rname{PTr-NT}:
We define \(\E{\dt} \defe \Hole\) and \(r^{\dt} \defe t\).
\end{itemize}
Then, we define \emph{\(\dt\)-reduction}, written by \(\dred{\dt}\), as follows:
\begin{align*}
\E{\dt}[(\lambda x .t_0)\,t_1] &\dred{\dt} [t_0/x]t_1 \mspace{-500mu} &&(\text{if } r^{\dt} = (\lambda x .t_0)\,t_1)
\\
\E{\dt}[A] &\dred{\dt} t' \mspace{-100mu} &&(\text{if } r^{\dt}=A,\, (A=t') \in \GRAM).
\end{align*}

We write \(\lo\) for the lexicographic order on pairs of natural numbers: \((n,n') \lo (m,m')\) iff \(n < m\) or
\(n=m\) and \(n' \le m'\), and write \(\slo\) for its strict order.

\begin{lemma}[subject reduction]
\label{lem:subject-reduction-for-induction}
Suppose 
\begin{align*}&
n>0
\\&
t \dred{\dt''} u \ \text{ for some }\dt''
\\&
\dt \pdt \PTE \PM t: (F,M,\prty) \hascost c\tr s
\\&
\order(\pty) < n \text{ for any } (x \COL \pty) \in\PTE.
\end{align*}
Then there exist \(\PTE'\), \(v\), and \(\dt'\) %
such that
\begin{align*}&
\PTE \subPTE \PTE'
\\&
s \reds v
\\&
\dt' \pdt \PTE' \PM u: (F,M,\prty) \hascost c\tr v
\\&
\order(\pty') < n \text{ for any } (x' \COL \pty') \in\PTE'
\\&
(\ac{n}{\dt'},\ntc{\dt'}) \lo (\ac{n}{\dt},\ntc{\dt}).
\end{align*}
\end{lemma}
\begin{proof}
The proof proceeds by induction on \(\dt\) with case analysis on the last rule used. 
Since the other cases are straightforward, we discuss only
the cases for \rname{PTr-Choice} and \rname{PTr-App}.
\begin{itemize}
\item Case of \rname{PTr-Choice}:
We have:
\begin{align*}&
t=t_1+t_2
\qquad
s=s_i
\qquad
\text{\(i=1\) or \(2\).}
\\&
\dt_i \pdt \PTE\PM \term_i:(F,M,\prty)\hascost{c}\tr s_i .
\end{align*}
By symmetry, we assume that \(t \dred{\dt''} u\) reduces \(t_1\) side
and let \(\dt''_1\) be the subderivation of \(\dt''\).
Thus we have \(u_1\) such that \(t = t_1+t_2 \dred{\dt''} u_1+t_2 = u\).

When \(i = 1\), by induction hypothesis for \(t_1 \dred{\dt''_1} u_1\),
there exist \(\PTE'\), \(v\), and \(\dt'_1\) %
such that
\begin{align*}&
\PTE \subPTE \PTE'
\\&
(s=)\ s_1 \reds v
\\&
\dt'_1 \pdt \PTE' \PM u_1 : (F,M,\prty) \hascost c\tr v
\\&
\order(\pty') < n \text{ for any } (x' \COL \pty') \in\PTE'
\\&
(\ac{n}{\dt'_1},\ntc{\dt'_1}) \lo (\ac{n}{\dt_1},\ntc{\dt_1}).
\end{align*}
By \rname{PTr-Choice} we have
\[
\dt' \defe \frac{\dt'_1}{
\PTE' \PM u_1 + t_2 : (F,M,\prty) \hascost c\tr v
}
\]
and \((\ac{n}{\dt'},\ntc{\dt'}) \lo (\ac{n}{\dt},\ntc{\dt})\).

When \(i=2\),
By \rname{PTr-Choice} we have
\[
\dt' \defe \frac{\dt_2}{
\PTE \PM u_1 + t_2 : (F,M,\prty) \hascost c\tr s_2
}
\]
and \((\ac{n}{\dt'},\ntc{\dt'}) = (\ac{n}{\dt},\ntc{\dt})\).
Also we have \(s = s_2 \red^0 s_2\). The conditions for \(\PTE' \defe \PTE\) are trivial.

\item Case of \rname{PTr-App}:
Let \(t=t_0t_1\), and
\(\dt''_i\) (\(i=0,1,\dots,k''\)) be the subderivations
of \(\dt''\) determined by all the premises of the last \rname{PTr-App} in \(\dt''\), where
the root of \(\dt''_0\) is \(t_0\).

If \(t_i\) (\(i=0\) or \(1\)) is reduced in the reduction \(t = t_0\,t_1 \dred{\dt''} u\)
(i.e., if \(\ac{n}{\dt''_i}>0\) for some \(i=0,\dots,k''\)
or \(t_0\) is a non-terminal),
then the result follows immediately from the induction hypothesis
for the subderivations of \(\dt\) whose root is \(t_i\).

Otherwise, we have \(t_0=\lambda x.t'\) and \(u = [t_1/x]t'\) with \(\eorder(t_0) = n\).
Then the result follows from Lemma~\ref{lem:SRandProg}.
\end{itemize}
\end{proof}

\begin{lemma}[progress]
\label{lem:progress}
Suppose 
\begin{align*}&
n>0
\\&
\dt \pdt \PTE \PM t: (F,M,\prty) \hascost c\tr s
\\&
\order(\pty) < n \text{ for any } (x \COL \pty) \in\PTE
\\&
(\ac{n}{\dt},\ntc{\dt}) \loo (1,0).
\end{align*}
Then there exist \(\PTE'\), \(u\), \(v\), and \(\dt'\) %
such that
\begin{align*}&
\PTE \subPTE \PTE'
\\&
t \dred{\dt} u
\qquad
s \reds v
\\&
\dt' \pdt \PTE' \PM u: (F,M,\prty) \hascost c\tr v
\\&
\order(\pty') < n \text{ for any } (x' \COL \pty') \in\PTE'
\\&
(\ac{n}{\dt'},\ntc{\dt'}) \slo (\ac{n}{\dt},\ntc{\dt}).
\end{align*}
\end{lemma}
\begin{proof}
The proof proceeds by induction on \(\dt\)
with case analysis on the last rule used for \(\dt\). 
Since the other cases are straightforward, we discuss only
the case for \rname{PTr-App}.

Let \(t=t_0t_1\), and now we have:
\begin{align}&
\dt_0 \pdt \PTE_0\PM t_0 :(F_0,M_0,\set{\pty_1,\ldots,\pty_k} \to \prty)\hascost{c_0}\tr s_0
\label{eq:PRGsubRedApp2}
\\&
\pty_i = (F_i,M_i,\prty_i)  \text{ for each \(i \in \set{1,\dots,k}\)}
\label{eq:PRGsubRedApp3}
\\&
\dt_i \pdt \PTE_{i}\PM \term_1:(F_{i},M_{i},\prty_i)\hascost{c_{i}}\tr s_{i}
\mbox{ for each $i\in\set{1,\ldots,k}$}
\label{eq:PRGsubRedApp4}
\\&
  M = M_0\uplus (\textstyle\biguplus_{i\in\set{1,\ldots,k}} M_{i})
\label{eq:PRGsubRedApp12}
\\&
  \Comp_{n}(\set{(F_0,c_0)}\mcup \set{(\emptyset,c_{i})\mid i\in\set{1,\ldots,k}}, M)
   = (F, c) 
\label{eq:PRGsubRedApp7}
\\&
\pty_1 < \cdots < \pty_k
\label{eq:PRGsubRedApp8}
\\&
\PTE = \PTE_0\PTEcup (\textstyle\sum_{i\in\set{1,\ldots,k}} \PTE_{i})
\label{eq:PRGsubRedApp9}
\\&
s = s_0\, s_{1}\,\cdots\,s_{k}
\label{eq:PRGsubRedApp10}
\end{align}
We further perform case analysis on the subderivations \(\dt_i\):

\begin{itemize}
\item
Case where \(\ac{n}{\dt_i} > 0\) for some \(i=1,\dots,k\):
Let \(i_0\) be the largest such \(i\).
By induction hypothesis for~\eqref{eq:PRGsubRedApp4} where \(i=i_0\), there exist 
\(\PTE'_{i_0}\), \(u_1\), \(v_{i_0}\), and \(\dt'_{i_0}\) such that
\begin{align*}&
\PTE_{i_0} \subPTE \PTE'_{i_0}
\\&
t_1 \dred{\dt_{i_0}} u_1
\qquad
s_{i_0} \reds v_{i_0}
\\&
\dt'_{i_0} \pdt \PTE'_{i_0} \PM u_1: (F_{i_0},M_{i_0},\prty_{i_0}) \hascost c_{i_0} \tr v_{i_0}
\\&
\order(\pty') < n \text{ for any } (x' \COL \pty') \in\PTE'_{i_0}
\\&
(\ac{n}{\dt'_{i_0}},\ntc{\dt'_{i_0}}) \slo (\ac{n}{\dt_{i_0}},\ntc{\dt_{i_0}}).
\end{align*}
For any \(i \neq i_0\), 
by Lemma~\ref{lem:subject-reduction-for-induction} applied to~\eqref{eq:PRGsubRedApp4},
we have
\(\PTE'_{i}\), \(v_{i}\), and \(\dt'_{i}\) such that
\begin{align*}&
\PTE_{i} \subPTE \PTE'_{i}
\\&
s_{i} \reds v_{i}
\\&
\dt'_{i} \pdt \PTE'_{i} \PM u_1: (F_{i},M_{i},\prty_{i}) \hascost c_{i} \tr v_{i}
\\&
\order(\pty') < n \text{ for any } (x' \COL \pty') \in\PTE'_{i}
\\&
(\ac{n}{\dt'_{i}},\ntc{\dt'_{i}}) \lo (\ac{n}{\dt_{i}},\ntc{\dt_{i}}).
\end{align*}
Then we have a derivation:
\[
\dt' \pdt
\PTE_0\PTEcup (\textstyle\sum_{i\in\set{1,\ldots,k}} \PTE'_{i})
\PM t_0\,u_1: (F,M,\prty) \hascost c\tr s_0\, v_{1}\,\cdots\,v_{k}.
\]
For \(\PTE' \defe \PTE_0\PTEcup (\textstyle\sum_{i\in\set{1,\ldots,k}} \PTE'_{i})\),
it is clear that
\(\PTE \subPTE \PTE'\)
and
\(\order(\pty') < n \text{ for any } (x' \COL \pty') \in\PTE'\).
Also, \(\dt'\) satisfies the required condition:
if \(\ac{n}{\dt'_{i}} < \ac{n}{\dt_{i}}\) for some \(i=1,\dots,k\),
then \(\ac{n}{\dt'} < \ac{n}{\dt}\),
and if \(\ac{n}{\dt'_{i}} = \ac{n}{\dt_{i}}\) for all \(i=1,\dots,k\),
then \(\ac{n}{\dt'} = \ac{n}{\dt}\) and \(\ntc{\dt'} < \ntc{\dt}\).
Finally, we have
\(t=t_0\,t_1 \dred{\dt} t_0\,u_1\) (by the definition of \(\E{\dt}\))
and \(s = s_0\, s_{1}\,\cdots\,s_{k} \reds s_0\, v_{1}\,\cdots\,v_{k}\).

\item Case where \(\ac{n}{\dt_i} = 0\) for any \(i=1,\dots,k\) and \(\ac{n}{\dt_0} > 0\):
Similar to (and easier than) the previous case.

\item Case where \(\ac{n}{\dt_i}=0\) for any \(i = 0,\dots,k\):
Now \(t_0\) is a non-terminal or \(\lambda\)-abstraction.
The case of non-terminal is straightforward; we consider the case of \(\lambda\)-abstraction.
Let \(t_0 = \lambda x . t'\).
Since \(\ac{n}{\dt} \ge 1\) and \(\ac{n}{\dt_i}=0\) for any \(i = 0,\dots,k\), 
we have \(\eorder(\lambda x.t')=n\).
Then the result follows from Lemma~\ref{lem:SRandProg}, with \(u \defe [t_1/x]t'\).

\end{itemize}
\end{proof}

\begin{lemma}
\label{lem:comp}
If \(\Comp_n(\set{(F_1,c_1),\ldots,(F_k,c_k)},M)=(F,c)\)
and \\\(\Comp_{n-1}(\set{(F_1\restrict{<n-1},c'_1),\ldots,(F_k\restrict{<n-1},c'_k)},M\restrict{<n-1})=(F',c')\),
then \(F\restrict{<n-1}=F'\).
\end{lemma}
\begin{proof}
Trivial by the definition of \(\Comp_n\).
\end{proof}

\begin{lemma}[decrease of the order]
\label{lem:subRed2Parys2}
Suppose 
\begin{inparaenum}[(i)]
\item \(n>1\)
\item \(\dt \pdt \PTE \PM t: (F,M,\prty)\hascost c\tr s\)
with \(\ac{n}{\dt} = 0\),
\item \(\eorder(t)<n\)
\item \(\PTE\) contains neither flag \(n-1\) nor marker \(n-1\).
\end{inparaenum}
Then %
there exist \(\dt'\) and \(c'\) such that
\(\dt' \pdt \PTE \PMm{n-1} t: (F\restrict{<n-1},M\restrict{<n-1},\prty) \hascost c'\tr s\) 
with \(c' \ge c + |F \cap \set{n-1}|\). 
Furthermore, the two derivation trees have the same structure (except the labels),
and for each node of \(\dt\) labeled by
\(\PTE_1 \PMm{n} t_1: (F_1,M_1,\prty_1) \hascost c_1\tr s_1\) with \(c_1>0\),
the corresponding node of \(\dt'\) is labeled by
\(\PTE_1 \PMm{n} t_1: (F_1\restrict{<n-1},M_1\restrict{<n-1},\prty_1) 
\hascost c'_1\tr s_1\) with \(c'_1>0\).
\end{lemma}

\begin{proof}
The proof proceeds by induction on \(\dt\), with case analysis on the last rule used.
\begin{itemize}
\item Case \rname{PTr-Abs}:
In this case, we have:
\[
\begin{array}{l}
\PTE, x\COL\pity \PMm{n} t_0:(F,M_0,\prty_0)\hascost{c}\tr s_0\\
t = \lambda x.t_0\qquad 
M = M_0\setminus \Markers(\pity) \qquad
\prty = \pity\to\prty_0 \qquad
s = \lambda x_{\prty_1}.\cdots\lambda x_{\prty_k}.s
\end{array}
\]
Since \(\eorder(t_0)<n\), \(\pity\) does not contain \(n-1\) as a flag or a marker.
Therefore, by the induction hypothesis, we have a derivation for
\[ \PTE, x\COL\pity \PMm{n-1} t_0:(F\restrict{<n-1},M_0\restrict{<n-1},\prty_0)\hascost{c'}\tr s_0\]
for some \(c'\) such that \(c'\geq c+|F\cap\set{n-1}|\).
By using \rname{PTr-Abs}, we obtain a derivation for 
\[ \PTE \PMm{n-1} t_0:(F\restrict{<n-1},M\restrict{<n-1},\prty)\hascost{c'}\tr s\]
as required.
\item Case \rname{PTr-App}:
In this case, we have:
\[
\begin{array}{l}
t = t_0t_1\qquad \eorder(t_0)=\ell<n\\
\PTE_0\PM \term_0:(F_0,M_0,\set{(F_1,M_1,\prty_1),\ldots,
(F_k,M_k,\prty_k)}\to \prty)\hascost{c_0}\tr s_0\\
  \PTE_{i}\PM \term_1:(F'_{i},M'_{i},\prty_i)\hascost{c_{i}}\tr s_{i}
\andalso F'_{i}\restrict{<\ell}=F_i\andalso M'_{i}\restrict{<\ell}=M_i
\mbox{ for each $i\in\set{1,\ldots,k}$}\\
  M = M_0\uplus (\biguplus_{i\in\set{1,\ldots,k}} M'_{i})\\
  \Comp_{n}(\set{(F_0,c_0)}\mcup \set{(F'_{i}\restrict{\geq \ell},c_{i})\mid i\in\set{1,\ldots,k}}, M)
   = (F, c)\\
(F_1,M_1,\prty_1)<\cdots< (F_k,M_k,\prty_k)\\
\PTE = \PTE_0\PTEcup (\textstyle\sum_{i\in\set{1,\ldots,k}} \PTE_{i})\\
s = s_0\, s_{1}\,\cdots\,s_{k}
\end{array}
\]
By the induction hypothesis, we have
\[
\begin{array}{l}
\PTE_0\PMm{n-1} \term_0:(F_0\restrict{<n-1},M_0\restrict{<n-1},\set{(F_1,M_1,\prty_1),\ldots,
(F_k,M_k,\prty_k)}\to \prty)\hascost{c'_0}\tr s_0\\
  \PTE_{i}\PMm{n-1} \term_1:(F'_{i}\restrict{<n-1},M'_{i}\restrict{<n-1},\prty_i)\hascost{c'_{i}}\tr s_{i}\\
c'_0 \geq c_0 + |F_0\cap \set{n-1}|\qquad
c'_{i} \geq c_{i}+|F'_{i}\cap \set{n-1}|
\end{array}
\]
Let \((F',c')\) be:
\[ \Comp_{n-1}(
\set{(F_0\restrict{<n-1},c'_0)}\mcup \set{((F'_{i}\restrict{<n-1})\restrict{\geq \ell},c'_{i})\mid i\in\set{1,\ldots,k}}, M\restrict{n-1}).\]
By using \rname{PTr-App}, we obtain a derivation tree for
\[\PTE \PMm{n-1} t:(F',M\restrict{<n-1},\prty)\hascost c'\tr s.\]
By Lemma~\ref{lem:comp}, we have \(F' = F\restrict{<n-1}\). Thus, it remains to show \(c'\geq c+|F\cap\set{n-1}|\), 
as the derivation satisfies the other conditions.
We consider \(f_\ell\) and \(f'_\ell\)
in the computation of 
\(\Comp_{n}(\set{(F_0,c_0)}\mcup \set{(F'_{i}\restrict{\geq \ell},c_{i})\mid i\in\set{1,\ldots,k}}, M)\).
If \(n-1\in M\), then \(n-1\not\in F\).
We have \(f_n' = f_{n-1} = f'_{n-1} + |F_0\cap \set{n-1}| + \textstyle\sum_{i} |F'_{i}\cap \set{n-1}|\), and hence:
\[\begin{array}{l} c = f_n' + c_0+ \textstyle\sum_{i}c_{i}
 = f'_{n-1} + (c_0+|F_0\cap \set{n-1}|) + \textstyle\sum_{i}(c_{i}+|F'_{i}\cap\set{n-1}|)\\
\leq f'_{n-1} + c'_0 + \textstyle\sum_{i} c'_{i} = c'
\end{array}
\]
as required.
If \(n-1\not\in M\), then we have \(c=0\) by Lemma~\ref{lem:marker}.
If \(n-1\not\in F\), we obtain
\( c + |F\cap\set{n-1}| = 0 \leq c'\) immediately.
Otherwise, \(f_{n-1} = f'_{n-1}+
|F_0\cap \set{n-1}| + \textstyle\sum_{i}|F'_{i}\cap\set{n-1}|>0\).
Thus, we have
\[
\begin{array}{l}
 c + |F\cap\set{n-1}| = 1
 \leq f'_{n-1}+|F_0\cap \set{n-1}| + \textstyle\sum_{i}|F'_{i}\cap\set{n-1}|\\
\leq f'_{n-1}+c'_0 + \textstyle\sum_{i} c'_{i} \leq c'
\end{array}
\]
as required.
\item Case \rname{PTr-Const}:
In this case, we have:
\[
\begin{array}{l}
t = a\,t_1\,\cdots\,t_k\qquad s = a\,s_1\,\cdots\,s_k\\
  \PTE_i \PM \term_i:(F_i,M_i,\T)\hascost{c_i}\tr s_i\mbox{ for each $i\in\set{1,\ldots,k}$}\\
  M = M_1\uplus\cdots\uplus M_k\\
  \Comp_{n}(\set{(\set{0},0),(F_1,c_1),\ldots,(F_k,c_k)},M)=(F,c)\\
\PTE = \PTE_1\PTEcup\cdots \PTEcup \PTE_k
\end{array}
\]
By the induction hypothesis, we have:
\[
\begin{array}{l}
  \PTE_i \PMm{n-1} \term_i:(F_i\restrict{<n-1},M_i\restrict{<n-1},\T)\hascost{c'_i}\tr s_i\mbox{ for each $i\in\set{1,\ldots,k}$}\\
c'_i \geq c_i+|F_i\cap\set{n-1}|
\end{array}
\]
Let \(\Comp_{n-1}(\set{(\set{0},0),(F_1\restrict{<n-1},c'_1),\ldots,(F_k\restrict{<n-1},c'_k)},M\restrict{<n-1})=(F',c')\).
By Lemma~\ref{lem:comp}, \(F'=F\restrict{<n-1}\).
Thus, by using \rname{PTr-Const}, we obtain a derivation for
\[ \PTE\PMm{n-1}\term: (F\restrict{<n-1},M\restrict{<n-1},\T)\hascost{c'}\tr s.\]
It remains to show \(c'\geq c+|F\cap\set{n-1}|\), as the derivation satisfies the other conditions.
We consider \(f_\ell\) and \(f'_\ell\)
in the computation of 
 \( \Comp_{n}(\set{(\set{0},0),(F_1,c_1),\ldots,(F_k,c_k)},M)=(F,c)\).
If \(n-1\in M\), then \(n-1\not\in F\).
Thus, we have:
\[
\begin{array}{l}
c+|F\cap\set{n-1}| = c = f'_n + c_1+\cdots + c_k =
f_{n-1} + c_1+\cdots + c_k \\
=
(f'_{n-1} + |F_1\cap \set{n-1}| + \cdots +|F_k\cap\set{n-1}|)+c_1+\cdots + c_k \\
= 
f'_{n-1} + (c_1+|F_1\cap \set{n-1}|) + \cdots+(c_k+ |F_k\cap\set{n-1}|)
\leq 
f'_{n-1} + c_1'+\cdots+c'_k = c'
\end{array}
\]
If \(n-1\not\in M\), then by Lemma~\ref{lem:marker},
it must be the case that \(c=0\).
If \(n-1\not\in F\), then we obtain
\(c+|F\cap\set{n-1}|=0 \leq c'\) immediately.
Otherwise, \(f_{n-1} = f'_{n-1}+|F_1\cap\set{n-1}|+\cdots+|F_k\cap\set{n-1}|>0\).
Thus, we have:
\[
\begin{array}{l}
  c+|F\cap\set{n-1}| =1 \leq
  f'_{n-1}+|F_1\cap\set{n-1}|+\cdots+|F_k\cap\set{n-1}|\\
\leq f'_{n-1}+c_1'+\cdots + c_k' \leq c'
\end{array}
\]
as required.
\item Case \rname{PTr-NT}:
In this case, the result follows immediately from the induction hypothesis.
\end{itemize}

\end{proof}

The following two lemmas are corollaries of Lemmas~\ref{lem:progress} and \ref{lem:subRed2Parys2}
respectively.
\begin{lemma}
\label{lem:subRedParys}
If \(n>0\),\, \(\dt \pdt \emptyset \PM t: (\emptyset,\set{0,\ldots,n-1},\T) \hascost c\tr s\), and
\(\ac{n}{\dt}>0\),
then there exist \(u\), \(v\), and \(\dt'\) %
such that
\(t \oreds{n} u\),\,
\(s \reds v\),\,
\(\dt' \pdt \emptyset \PM u: (\emptyset,\set{0,\ldots,n-1},\T) \hascost c\tr v\),\, and
\(\ac{n}{\dt'} < \ac{n}{\dt}\).
\end{lemma}

\begin{lemma}
\label{lem:subRed2Parys}
If \(n>1\),\,
\(\dt \pdt \emptyset \PM t: (\emptyset,\set{0,\ldots,n-1},\T) \hascost c\tr s\),\,
and \(\ac{n}{\dt} = 0\),
then %
\(\emptyset \p_{n-1} t: (\emptyset,\set{0,\ldots,n-2},\T) \hascost c'\tr s\) is derived for some \(c' \ge c\).
\end{lemma}

\begin{lemma}
\label{lem:soundTreeParys}
If \(\dt \pdt \emptyset \p_1 t: (\emptyset,\set{0},\T) \hascost c\tr s\)
and \(\ac{1}{\dt} = 0\),
then \(\tree(s)\in \Lang(t)\),
with \(c \leq |\tree(s)|\).
\end{lemma}
\begin{proof}
The proof proceeds by induction on \(\dt\) and
case analysis on the rule used last for \(\dt\).
Below we use the following fact:
\begin{equation}
\label{eq:soundTreeParysaux}
\begin{aligned}
&\text{
if \(\emptyset\p_1 \term:(F,\eset,\T)\hascost{c}\tr s\) is derived without \rname{PTr-App}
}\\
&\text{
then \(c =0\) and \(\tree(s) \in \Lang(t)\),
}
\end{aligned}
\end{equation}
where \(c=0\) follows from Lemma~\ref{lem:marker} and
\(\tree(s) \in \Lang(t)\) can be shown by straightforward induction on the derivation.

\begin{itemize}
\item Case of \rname{PTr-Weak}: Trivial.

\item Case of \rname{PTr-Mark}: We have
\infrule{
\emptyset\p_1 \term:(F',\set{0}\setminus M,\T)\hascost{c'}\tr s
\\
M\subseteq \set{0}
\andalso
\Comp_{1}(\set{(F',c')},\set{0})= (\emptyset,c)
}{
\emptyset\p_1\term:(\emptyset, \set{0},\T )\hascost{c}\tr s
}

If \(M=\eset\), then \(F'=\eset\) and \(c'=c\). Hence the required properties follow from the induction hypothesis.

If \(M\neq \eset\), then we have 
\(\emptyset\p_1 \term:(F',\eset,\T)\hascost{c'}\tr s\).
By~\eqref{eq:soundTreeParysaux}, we have \(\tree(s) \in \Lang(t)\) and \(c' = 0\).
Since \(c \le 1\) by \(\Comp_{1}(\set{(F',c')},\set{0})= (\emptyset,c)\),
we have \(c \le |\tree(s)|\).

\item Case of \rname{PTr-Var}: This case does not happen since the environment is empty.

\item Case of \rname{PTr-Choice}: Clear by induction hypothesis.

\item Case of \rname{PTr-Abs}: This case does not happen since the simple type of \(t\) is \(\T\).

\item Case of \rname{PTr-App}: This case does not happen by the assumption. 

\item Case of \rname{PTr-Const}: We have:
\infrule{\arity(a)=k \\ %
  \eset \p_1 \term_i:(F_i,M_i,\T)\hascost{c_i}\tr s_i\mbox{ for each $i\in\set{1,\ldots,k}$}\\
  \set{0} = M_1\uplus\cdots\uplus M_k\\
  \Comp_{n}(\set{(\set{0},0),(F_1,c_1),\ldots,(F_k,c_k)},\set{0})=(\eset,c)
  }
{\eset \p_1 a\,\term_1\,\cdots\,\term_k:
          (\eset,\set{0},\T)\hascost c\tr a\,s_1\,\cdots\,s_k}
Now there exists \(i_0\) such that
\(M_{i_0} = \set{0}\) and \(M_i = \eset\) for \(i\neq i_0\).
By~\eqref{eq:soundTreeParysaux}, for \(i \neq i_0\), \(c_i = 0\) and \(\tree(s_i) \in \Lang(t_i)\).
Since \(F_{i_0} = \eset\), by induction hypothesis, \(c_{i_0} \le |\tree(s_{i_0})|\) and \(\tree(s_{i_0}) \in \Lang(t_{i_0})\).
By \(\Comp_{n}(\set{(\set{0},0),(F_1,c_1),\ldots,(F_k,c_k)},\set{0})=(\eset,c)\), we have
\[
c = 1 + \big|\set{i\le k| 0 \in F_i}\big| + c_1 +\cdots+ c_k \le 
1 + |\tree(s_1)| + \cdots + |\tree(s_k)| 
= |\tree(a\,s_1\cdots s_k)|
\]
and also we have \(\tree(a\,s_1\,\cdots\,s_k) \in \Lang(a\,\term_1\,\cdots\,\term_k)\), as required.

\item Case of \rname{PTr-NT}: Clear by induction hypothesis.
\end{itemize}
\end{proof}

\begin{proof}
[Proof of Lemma~\ref{lem:soundness-tr}]
If \(n>0\),
since \(\emptyset \PM S: (\emptyset,\set{0,\ldots,n-1},\T) \hascost c\tr s\),
by using Lemma~\ref{lem:subRedParys} repeatedly, we have
\begin{align*}&
S \oreds{n} t_{n-1}
\qquad
s \reds s_{n-1}
\\&
\dt_n \pdt \emptyset \PM t_{n-1}: (\emptyset,\set{0,\ldots,n-1},\T) \hascost c\tr s_{n-1}
\\&
\ac{n}{\dt_n} = 0.
\end{align*}
If \(n>1\),
by Lemma~\ref{lem:subRed2Parys} we have \(c_{n-1} \ge c\) and
\[
\emptyset \PMm{n-1} t_{n-1}: (\emptyset,\set{0,\ldots,n-2},\T) \hascost c_{n-1} \tr s_{n-1}.
\]
By repeating this procedure,
we obtain
\begin{align*}&
S \oreds{n} t_{n-1} \oreds{n-1} \cdots \oreds{1} t_0
\\&
s \reds s_{n-1} \reds \cdots \reds s_0
\\&
c_1 \ge \dots \ge c_{n-1} \ge c
\\&
\dt_1 \pdt \emptyset \PMm{1} t_{0}: (\emptyset,\set{0},\T) \hascost c_1 \tr s_{0}
\\&
\ac{1}{\dt_1} = 0.
\end{align*}
By Lemma~\ref{lem:soundTreeParys}, we have
\[
\tree(s_0) \in \Lang(t_0)
\qquad
c_1 \le |\tree(s_0)|
\]
and hence
\[
\tree(s) = \tree(s_0) \in \Lang(t_0) \subseteq \Lang(S) = \Lang(\GRAM)
\qquad
c \le c_1 \le |\tree(s_0)| = |\tree(s)|
\]
as required.
\end{proof} %
\subsection{Proof of Lemma~\ref{lem:parys5}}
\label{sec:parys-trans2}

We now discuss how to modify the argument in the previous subsections 
to obtain a triple \((G,H,u)\) that satisfies the requirement of Lemma~\ref{lem:parys5}.
Let \(\GRAM\) be an order-\(n\), direction-annotated (i.e., each occurrence of a terminal symbol is
annotated with a direction) tree grammar.

We first modify typing (or, type-based transformation) rules.
In the type system of Section~\ref{sec:parys-type} (and the original type system of 
Parys~\cite{Parys16ITRS}),
the size of a tree (= the number of order-0 flags) is estimated via
the number of order-1 flags placed on the path from the root of the derivation to the
(unique occurrence of) order-0 marker. Thus, for a direction-annotated grammar,
by ensuring that the order-0 marker can be placed only at the node \(\directedT{0}{a}\,T_1\,\cdots\,T_k\)
reached following the directions, we can estimate
 the length of the path following the directions, instead of the size of the whole tree.

Based on the intuition above, we modify the rules \rname{PTr-Mark} and \rname{PTr-Const}
as follows.

 \infrule[PDT-Mark]{\PTE\PMD \term:(F',M',\prty)\hascost{c'}\tr s\\
M\subseteq \set{j \mid \eorder(\term)\leq j<n}\andalso   \Comp_n(\set{(F',c')},M\uplus M')= (F,c)\\
\mbox{if \(0\in M\), then \(t\) is of the form \(\directedT{0}{a}\,\term_1\,\cdots\,\term_k\)}
}
        {\PTE\PMD\term:(F, M\uplus M',\prty )\hascost{c}\tr s}

\infrule[PDT-Const]{\arity(a)=k \\ %
  \PTE_i \PMD \term_i:(F_i,M_i,\T)\hascost{c_i}\tr s_i\mbox{ for each $i\in\set{1,\ldots,k}$}\\
  M = M_1\uplus\cdots\uplus M_k\andalso
  0\not\in \bigcup_{j\in \set{1,\ldots,k}\setminus\set{ i}}M_j\\
  \Comp_n(\set{(\set{0},0),(F_1,c_1),\ldots,(F_k,c_k)},M)=(F,c)
  }
        {\PTE_1\PTEcup\cdots \PTEcup \PTE_k \PMD \directedT{i}{a}\,\term_1\,\cdots\,\term_k:
          (F,M,\T)\hascost c\tr a\,s_1\,\cdots\,s_k}
The other rules are unchanged (except \(\PM\) is replaced by \(\PMD\)); we write \rname{PDT-X} for
the rule obtained from \rname{PTr-X}, for each \textsc{X}.

The following are variations of Theorem~\ref{th:completeness-tr}
and Lemma~\ref{lem:soundness-tr}.
\begin{lemma}
\label{th:completeness-tr2}
Let \(t\) %
be an order-\(n\) direction-annotated term. 
If \(t\reds \pi\), then
\(\dt\pdt \emptyset \PMD t:(\emptyset,\set{0,\ldots,n-1},\T)\hascost c\tr s\)
for some \(\dt\) and \(c\) such that \(\expn{n}{c}\geq \plen{\pi}\).
Furthermore, \(\toLty{\dt}\) is an admissible derivation 
in the linear type system.
\end{lemma}

\begin{lemma}
\label{lem:soundness-tr2}
Let \(t\) %
be an order-\(n\) direction-annotated term. 
If \(\emptyset \PM t:(\emptyset,\set{0,\ldots,n-1},\T)\hascost c\tr s\), then \(s\) is a \(\stlambda\)-term
of order at most \(n\), and \(t\reds \tree(s)\), %
with \(c \leq \plen{\tree(s)}\).
\end{lemma}

To prove Lemma~\ref{lem:parys5}, we will use Lemma~\ref{th:completeness-tr2} to obtain
a pumpable derivation for (the direction-annotated version of) \(C[D^i[t]]\) where \(C,D\) are linear contexts.
In the construction of the pumpable derivation, we need to require that a sub-derivation for a term of the form
\(D^j[t]\) occurs only once. To show that property, we extend the syntax of terms with labels as follows.
\[
t ::= x \mid a\,t_1\,\cdots\,t_k \mid t_0t_1 \mid \lambda x.t \mid t_0+t_1 \mid t^\LAB.
\]
Here, \(\LAB\) is a \emph{multiset} of labels of the form \(\set{\psi_1,\ldots,\psi_k}\).
We often just write \(t^\Lab\) for \(t^{\set{\Lab}}\).
We have omitted non-terminals, since they will not be used below.
We identify \((t^{\LAB_1})^{\LAB_2}\) with \(t^{\LAB_1\mcup \LAB_2}\).
We extend the reduction relation \(\red\) to that for labeled terms by:
\[ C[(\lambda x.t_0)^\LAB t_1] \red C[([t_1/x]t_0)^\LAB]
\qquad C[(t_0+t_1)^\LAB] \red C[t_i^\LAB] \mbox{ ($i=0\lor i=1$)},\]
where \(C\) is an arbitrary context.
Note that the (multi)set of labels for a function is transferred to the residual.
For example, we have:
\[ (\lambda x.x^{\Lab_1})^{\Lab_2} \Tc^{\Lab_3}
  \red ([\Tc^{\Lab_3}/x]x^{\Lab_1})^{\Lab_2}
   = ((\Tc^{\Lab_3})^{\Lab_1})^{\Lab_2} = \Tc^{\set{\Lab_1,\Lab_2,\Lab_3}}.\]

The following is a labeled version of Lemma~\ref{lem:order-reduction}.
\begin{lemma}
\label{lem:order-reduction-labeled}
Let \(t\) be a labeled, order-\(n\) term.
If \(t \reds \pi\), then 
\[
t \oreds{n} t_{n-1} \oreds{n-1} \cdots \oreds{1} t_0 \creds \pi.
\]
\end{lemma}
\begin{proof}
The proof is essentially the same as that of Lemma~\ref{lem:order-reduction}.
Note that the \(\beta\)-reduction relation satisfies the Church-Rosser property 
also for labeled terms (which can be proved by using the standard argument using parallel reduction).
\end{proof}
The transformation rules are extended for labeled terms accordingly.
We write \(\Labs(t)\) for the multiset of all the labels occurring in \(t\) (the multiplicities are counted as many as they occur).
For a transformation derivation \(\dt\), we also write \(\Labs(\dt)\) for 
the multiset of labels attached to the term in the conclusion of a judgment occurring in the derivation,
except the conclusion of \rname{PDT-Mark} and \rname{PDT-Weak}, i.e.,
\[\Labs\left(\rb{-1ex}{\infers{\PTE\PMD t^\LAB:\pty\hascost c\tr s}{\dt_1 \andalso\cdots\andalso \dt_k}}\right)
 = \LAB\mcup \Labs(\dt_1)\mcup \cdots \mcup \Labs(\dt_k)\]
if the last rule is neither \rname{PDT-Mark} nor \rname{PDT-Weak} and \(t\) is not of the form \(t'^{\LAB'}\)
(so that \(\LAB\) is the multiset of all the outermost labels), and
\[\Labs\left(\rb{-1ex}{\infers{\PTE\PMD t^\LAB:\pty\hascost c\tr s}{\dt}}\right)
 = \Labs(\dt)\]
if the last rule is \rname{PDT-Mark} or \rname{PDT-Weak}.

We can strengthen Lemma~\ref{th:completeness-tr2} as follows.
\begin{lemma}
\label{th:completeness-tr2-linear}
For a reduction sequence
\[ t \oreds{n} t_{n-1} \oreds{n-1} \cdots \oreds{1} t_0 \creds \pi,\]
there exists a derivation \(\dt\) for
\(\emptyset \PMD t:(\emptyset,\set{0,\ldots,n-1},\T)\hascost c\tr s\)
such that \(\expn{n}{c}\geq \plen{\pi}\).
Moreover, \(\toLty{\dt}\) is an admissible derivation in the linear type system.
Furthermore, if 
\(t\) is of the form \(C[t_0^\Lab]\) where \(\Lab\) does not occur in \(C\) nor \(t_0\)
and if \(\Labs(\pi)(\Lab)\leq 1\), 
then \(\Labs(\dt)(\Lab) \le 1\).
\end{lemma}
\begin{proof}
The proof of the properties except the part ``Furthermore, ...'' is almost the same as the proof of Theorem~\ref{th:completeness-tr}. To confirm the part ``Furthermore, ...'', it suffices to observe the following facts:
(i) for the derivation \(\dt_0\) for \(\pi\) constructed in the proof of (a direction-annoated version of)
Lemma~\ref{lem:subj-expansion-base}, \(\Labs(\dt_0)=\Labs(\pi)\), and
(ii) during the construction of derivations \(\dt'\) (in a backward manner with respect to the reduction sequence)
in the proofs of Lemmas~\ref{lem:subj-expansion} and \ref{lem:increase-of-order}, \(\Labs(\dt')\) decreases
 monotonically.
\end{proof}

The following lemma ensures that, if \(C\) is a linear context,
then there is a reduction sequence that satisfies the assumption of
the lemma above.
\begin{lemma}
\label{lem:linear-context-labels}
Suppose (i) a pair of contexts 
\(\Hole : \sty \pK C: \T\) and \(\Hole : \sty \pK D: \sty\) is linear, 
(ii) \(\pK t:\sty\) is a \(\stlambda\)-term,
and (iii) \(C,D,t\) do not contain labels.
Define \(t^{(k)}\) by: \(t^{(0)} = t^{\Lab_0}\) and \(t^{(i)} = D[t^{(i-1)}]^{\Lab_i}\).
Then there exists a reduction sequence
\[ C[t^{(k)}] \oreds{n} t_{n-1} \oreds{n-1} \cdots \oreds{1} t_0 \creds \pi\]
such that \(\Labs(\pi)(\Lab_i)=1\).
\end{lemma}
\begin{proof}
Let \(C[t^{(k)}] \reds \pi\) be the call-by-name reduction sequence (note that, by the typing assumptions,
the reduction sequence yields a tree). By the linearity condition,
the reduction sequence must be of the form:
\[ C[t^{(k)}] \reds E_k[t^{(k)} \seq{u}_k] \reds E_{k-1}[(t^{(k-1)}\seq{u}_{k-1}] \reds 
\cdots E_0[t_{(0)} \seq{u}_0] \reds \pi\]
where \(E_i\) is a call-by-name evaluation context,
\(t^{(i)}\seq{u}_i\) is a ground term, and
\(E_i[t^{(i)}\seq{u}_i]\) contains exactly one occurrence of \(\Lab_j\) for each \(j\in\set{0,\ldots,k}\)
(since \(\Lab_i\) occurs only in the head position in \(t^{(i)}\seq{u}_i\) and 
can never be duplicated afterwords).
Thus, \(\Labs(\pi)(\Lab_i)=1\) for each \(i\in\set{0,\ldots,k}\).
By Lemma~\ref{lem:order-reduction-labeled}, we have the required result.
\end{proof}
Using the lemmas above, we can prove 
 Lemma~\ref{lem:parys5}. %

\begin{proof}[Proof of Lemma~\ref{lem:parys5}]
Let \(C', D', t'\) be the direction-annotated contexts and term obtained from
\(C, D, t\) respectively by (recursively) replacing each occurrence of
\(a\,t_1\,\cdots\,t_k\) (where \(k=\arity(a) >0\)) with:
\[(\directedT{1}{a}\,t_1\,\cdots\,t_k) + \cdots + (\directedT{k}{a}\,t_1\,\cdots\,t_k),\]
and replacing each occurrence of a nullary non-terminal
\(a\) with \(\directedT{0}{a}\). By the construction,
for any \(\pi'\in \Lang(C'[D'^i[t']])\), \(\rmd(\pi') = \tree(C[D^i[t]])\).
By the assumption that 
\(\set{\tree(C[D^k[t]]) \mid k \ge 1}\) is infinite, for any \(d\)
there exist \(k_d\) and \(\pi' \in \Lang(C'[D'^{k_d}[t']])\) such that
\(\plen{\pi'}\geq d\).
Then for any \(c_0\), by this fact with \(d = \expn{n}{c_0}\) and by Lemma~\ref{th:completeness-tr2-linear}, 
there exist \(\dt\), \(c\), and \(s\) such that
\(\dt \pdt \emptyset \PMD C'[D'^{k_{\expn{n}{c_0}}}[t']]:(\emptyset,\set{0,\ldots,n-1},\T)\hascost c\tr s\) holds,
\(c\geq c_0\), and
\(\toLty{\dt}\) is admissible.
Furthermore, for each \(i\),
a subderivation for \(D'^i[t']\) occurs at most once except that we ignore a subderivation
whose last rule is derived by \rname{PDT-Mark} or \rname{PDT-Weak}.
(To observe this, let \(u\) (\(u'\), resp.) be the term obtained from \(C[D^{k_{\expn{n}{c_0}}}[t]]\) (\(C'[D'^{{k_{\expn{n}{c_0}}}}[t']]\), resp.)
by adding a label \(\Lab_i\) to \(D^{i}[t]\) (\(D'^{i}[t']\), resp.) for each \(i\).
By Lemma~\ref{lem:linear-context-labels}, we have 
\(u \oreds{n} \oreds{n-1} \cdots \oreds{1}\creds \pi\) such that \(\Labs(\pi)(\Lab_i)=1\).
For any \(\pi' \in \Lang(u')\), guided by the direction-annotations in \(\pi'\),
from the reduction sequence from \(u\) to \(\pi\) we obtain a reduction sequence 
\(u'\oreds{n} \oreds{n-1} \cdots \oreds{1}\creds \pi'\) such that \(\Labs(\pi')(\Lab_i)=1\),
which satisfies the assumption for Lemma~\ref{th:completeness-tr2-linear}.)

By a similar reasoning to the proof of Lemma~\ref{lem:pumpable-derivation},
by choosing a sufficiently large number \(c_0\),
we can ensure that there exist 
\(q_1, p, q_2\) 
such that
\(k_{\expn{n}{c_0}} = q_1 + p + q_2\) and
\(\emptyset \PMD C'[D'^{{k_{\expn{n}{c_0}}}}[t']]:(\emptyset,\set{0,\ldots,n-1},\T)\hascost c_1+c_2+c_3\tr G[H[u]]\)
(where \(c_1+c_2+c_3=c\) and \(G[H[u]] = s\)) is derived from
\(\emptyset \PMD D'^{p+q_2}[t']:\pty\hascost c_1+c_2\tr H[u]\), which is in turn derived from
\(\emptyset \PMD D'^{q_2}[t']:\pty\hascost c_1\tr u\), where \(c_1,c_2>0\).
Thus, by ``pumping'' the derivation of 
\(\emptyset \PMD D'^{p+q_2}[t']:\pty\hascost c_1+c_2\tr H[u]\) from
\(\emptyset \PMD D'^{q_2}[t']:\pty\hascost c_1\tr u\), 
we obtain \(\emptyset \PMD C'[D'^{pk+q}[t]]:(\emptyset,\set{0,\ldots,n-1},\T)\hascost c_1+kc_2+c_3\tr
   G[H^k[u]]\), where \(q=k_{\expn{n}{c_0}}-p\).
By Lemma~\ref{lem:soundness-tr2}, we have 
\(\tree(G[H^k[u]])\in \Lang(C'[D'^{pk+q}[t']])\).
Since \(\rmd(\pi')=\tree(C[D^{pk+q}[t]])\)
for every \(\pi'\in \Lang(C'[D'^{pk+q}[t']])\),
we have \(\rmd(\tree(G[H^k[u]]))=\tree(C[D^{pk+q}[t]])\).
By Lemma~\ref{lem:soundness-tr2}, we also have \(\plen{\tree(G[H^k[u]])}\geq c_1+kc_2+c_3\).
Thus, we also have the second condition of the lemma.
Finally, since \(\toLty{\dt}\) is an admissible derivation, in which 
a linear type is assinged to every term containing \(t'\), 
the pair of \(G\) and \(H\) is linear by Lemma~\ref{lem:linearity}.
\end{proof}

\subsection{Examples}

\subsubsection{Example 1}
Consider the second-order tree grammar \(\GRAM_2\) consisting of the following rules.
\[
\begin{array}{l}
S \to R\,A \\
R\,f \to f\,\Te
\qquad 
R\,f \to R\, (T\,f)\\
A\,x\to \Ta\,x\,x\\
T\,f\,x \to f(f\,x).
\end{array}
\]
The language generated by \(\GRAM_2\) is:
\[
\set{ \pi_{2^n} \mid n\geq 0},
\]
where \(\pi_0 = \Te\) and \(\pi_{k+1} = \Ta\,\pi_k\,\pi_k\).
The rules can also be expressed as the set of equations:
\[
\begin{array}{l}
S = R\,A \qquad 
R = \lambda f.(f\,\Te +  R\, (T\,f))\qquad
A = \lambda x.\Ta\,x\,x\qquad
T = \lambda f.\lambda x.f(f\,x).
\end{array}
\]

A good way to construct a type derivation in Parys' type system is to first
reduce a term to a tree, and then construct type derivations for the terms occurring
during the reduction in a backward manner.
Consider the following rewriting sequence:
\[
\begin{array}{l}
S \red R\,A\red (\lambda f.(f\,\Te+R(T\,f)))A\red A\,\Te+R(T\,A) \red R(T\,A)\\
\red (\lambda f.(f\,\Te+R(T\,f)))(T\,A)
\red (T\,A\,\Te + R(T(T\,A)))\red T\,A\,\Te \\
\red (\lambda f.\lambda x.f(f\,x))\,A\,\Te
\red (\lambda x.A(A\,x))\Te\\
\red A(A\,\Te) 
\red A((\lambda x.\Ta\,x\,x)\Te) 
\red A(\Ta\,\Te\,\Te) \red (\lambda x.\Ta\,x\,x)(\Ta\,\Te\,\Te) 
\red \Ta(\Ta\,\Te\,\Te)(\Ta\,\Te\,\Te)
\end{array}
\]
Here, we have reduced order-2 redexes first (on the first three lines), and then order-1 redexes
(on the last line).
Here is the typing for the tree (we omit the outputs of transformations below): 
\[
\infers{\PMm{1}(\Ta(\Ta\,\Te\,\Te)(\Ta\,\Te\,\Te)): (\emptyset,\set{0},\T)\hascost{5}}
 {\infers{\PMm{1} \Ta\,\Te\,\Te:(\set{0},\emptyset,\T)\hascost{0}}
   {\PMm{1} \Te:(\set{0},\emptyset,\T)\hascost{0} & \PMm{1} \Te:(\set{0},\emptyset,\T)\hascost{0}}
  & \infers{\PMm{1} \Ta\,\Te\,\Te:(\emptyset,\set{0},\T)\hascost{3}}
    {\PMm{1}\Te:(\set{0},\emptyset,\T)\hascost{0}
     & \infers[PTr-Mark]{\PMm{1}\Te:(\emptyset,\set{0},\T)\hascost{1}}{\PMm{1}\Te:(\set{0},\emptyset,\T)\hascost{0}}
   }
  }
\]
By (two steps of) subject expansion, we obtain:
\[
\infers{\PMm{1}A(\Ta\,\Te\,\Te):(\emptyset,\set{0},\T)\hascost{5}}
{\infers{\PMm{1}A:(\emptyset,\emptyset,(\set{0},\emptyset,\T)\land(\emptyset,\set{0},\T)\to\T)\hascost{2}}
{\infers{\PMm{1}\lambda x.\Ta\,x\,x:
  (\emptyset,\emptyset,(\set{0},\emptyset,\T)\land(\emptyset,\set{0},\T)\to\T)\hascost{2}}
 {\infers{x\COL (\set{0},\emptyset,\T), x\COL (\emptyset,\set{0},\T) \PMm{1} \Ta\,x\,x:(\emptyset,\set{0},\T)\hascost{2}}
 {x\COL(\set{0},\emptyset,\T)\PMm{1}x\COL(\set{0},\emptyset,\T)\hascost{0}
  & x\COL (\emptyset,\set{0},\T) \PMm{1}x\COL (\emptyset,\set{0},\T) \hascost{0}}}}
 & \hspace*{-1cm} \infers{\PMm{1} \Ta\,\Te\,\Te:(\set{0},\emptyset,\T)\hascost{0}}{\cdots}
 & \infers{\PMm{1} \Ta\,\Te\,\Te:(\emptyset,\set{0},\T)\hascost{3}}{\cdots}}
\hspace*{-10em}\cdots{(*1)}
\]
By applying subject expansion to \(\PMm{1} \Ta\,\Te\,\Te:(\set{0},\emptyset,\T)\hascost{0}\)
and \(\PMm{1} \Ta\,\Te\,\Te:(\emptyset,\set{0},\T)\hascost{3}\), we also obtain:
\[
\infers{\PMm{1}A\,\Te:(\set{0},\emptyset,\T)\hascost{0}}
{\infers{\PMm{1} A:(\set{0},\emptyset, (\set{0},\emptyset,\T)\to \T)\hascost{0}}
{ \infers{\PMm{1}\lambda x.\Ta\,x\,x:(\set{0},\emptyset, (\set{0},\emptyset,\T)\to \T)\hascost{0}}
  {\infers{x\COL (\set{0},\emptyset,\T)\PMm{1} \Ta\,x\,x:(\set{0},\emptyset,\T)\hascost{0}}
   {x\COL (\set{0},\emptyset,\T)\PMm{1} x:(\set{0},\emptyset,\T)
    & x\COL (\set{0},\emptyset,\T)\PMm{1} x:(\set{0},\emptyset,\T)}
    }}
  &\PMm{1}\Te:(\set{0},\emptyset,\T)\hascost{0}}
\]

\[
\infers{\PMm{1}A\,\Te:(\emptyset,\set{0},\T)\hascost{3}}
{\infers{\PMm{1}A: (\emptyset,\emptyset,
    (\set{0},\emptyset,\T)\land (\emptyset,\set{0},\T)\to\T)\hascost{2}}{\cdots}
 & \PMm{1}\Te:(\set{0},\emptyset,\T)\hascost{0}
 & \PMm{1}\Te:(\emptyset,\set{0},\T)\hascost{1}
}
\]

Thus, by replacing the subderivations for \(\Ta\,\Te\,\Te\) in (*1) with the two derivations above,
we obtain:
\[
\infers{\PMm{1}A(A\,\Te):(\emptyset,\set{0},\T)\hascost{5}}
{\infers{\PMm{1}A:
  (\emptyset,\emptyset,(\set{0},\emptyset,\T)\land(\emptyset,\set{0},\T)\to\T)\hascost{2}}{\cdots}
 & \infers{\PMm{1} A\,\Te:(\set{0},\emptyset,\T)\hascost{0}}{\cdots}
 & \infers{\PMm{1} A\,\Te:(\emptyset,\set{0},\T)\hascost{3}}{\cdots}}
\]
Let \(\PTE_1=x\COL (\set{0},\emptyset,\T)\), \(\PTE_2 = x\COL(\emptyset,\set{0},\T)\),
and \(\PTE_3 = \PTE_1\PTEcup \PTE_2\).
Then we have:
\[\hspace*{-5em}
\infers{\PMm{1}(\lambda x.A(A\,x))\Te:(\emptyset,\set{0},\T)\hascost{5}}
{
\infers{\PMm{1}\lambda x.A(A\,x):(\emptyset,\emptyset,\prty_1)\hascost{4}}
{
\infers{\PTE_3\PMm{1}A(A\,x):(\emptyset,\set{0},\T)\hascost{4}}
{\infers{\PMm{1}A:
  (\emptyset,\emptyset,\prty_1)\hascost{2}}{\cdots}
 & \infers{\PTE_1\PMm{1} A\,x:(\set{0},\emptyset,\T)\hascost{0}}{\cdots}
 & \infers{\PTE_3\PMm{1} A\,x:(\emptyset,\set{0},\T)\hascost{2}}{
  \infers{\PMm{1}A:(\emptyset,\emptyset,\prty_1)\hascost{2}}{\cdots}
&  \PTE_1\PMm{1}x:(\set{0},\emptyset,\T)\hascost{0} & \PTE_2\PMm{1}x:(\emptyset,\set{0},\T)\hascost{0}}}}
 & \hspace*{-20em}\PMm{1}\Te:(\set{0},\emptyset,\T)\hascost{0}
 & \PMm{1}\Te:(\emptyset,\set{0},\T)\hascost{1}
}
\]
where \(\prty_1 = (\set{0},\emptyset,\T)\land (\emptyset,\set{0},\T)\to\T\).
Now, by increasing the order of the type system to \(2\), we obtain:
\[\hspace*{-5em}
\infers{\PMm{2}(\lambda x.A(A\,x))\Te:(\set{1},\set{0},\T)\hascost{0}}
{
\infers{\PMm{2}\lambda x.A(A\,x):(\set{1},\emptyset,\prty_1)\hascost{0}}
{
\infers{\PTE_3\PMm{2}A(A\,x):(\set{1},\set{0},\T)\hascost{0}}
{\infers{\PMm{2}A:
  (\set{1},\emptyset,\prty_1)\hascost{0}}{\cdots}
 & \infers{\PTE_1\PMm{2} A\,x:(\set{0},\emptyset,\T)\hascost{0}}{\cdots}
 & \infers{\PTE_3\PMm{2} A\,x:(\set{1},\set{0},\T)\hascost{0}}{
  \infers{\PMm{2}A:(\set{1},\emptyset,\prty_1)\hascost{0}}{\cdots}
 &  \PTE_1\PMm{2}x:(\set{0},\emptyset,\T)\hascost{0} & \PTE_2\PMm{2}x:(\emptyset,\set{0},\T)\hascost{0}}}}
 & \hspace*{-20em}\PMm{2}\Te:(\set{0},\emptyset,\T)\hascost{0}
 & \PMm{2}\Te:(\set{1},\set{0},\T)\hascost{0}
}
\]
Note that now every counter (which now represent the numer of order-2 flags) has become \(0\),
and instead order-1 flags have been set in the places where there were non-zero counters.

Let us now put an order-1 marker to the rightmost occurrence of \(A\):
\[\small\hspace*{-5em}
\infers{\PMm{2}(\lambda x.A(A\,x))\Te:(\emptyset,\set{0,1},\T)\hascost{3}}
{
\infers{\PMm{2}\lambda x.A(A\,x):(\emptyset,\set{1},\prty_1)\hascost{2}}
{
\infers{\PTE_3\PMm{2}A(A\,x):(\emptyset,\set{0,1},\T)\hascost{2}}
{{
\infers{\PMm{2}A:
  (\set{1},\emptyset,\prty_1)\hascost{0}}{\cdots}}
 & \infers{\PTE_1\PMm{2} A\,x:(\set{0},\emptyset,\T)\hascost{0}}{\cdots}
 & \infers{\PTE_3\PMm{2} A\,x:(\emptyset,\set{0,1},\T)\hascost{1}}{
  \infers{\PMm{2}A:
  (\emptyset,\set{1},\prty_1)\hascost{1}}{\infers{\PMm{2}A:(\set{1},\emptyset,\prty_1)\hascost{0}}{\cdots}}
&  \PTE_1\PMm{2}x:(\set{0},\emptyset,\T)\hascost{0} & \PTE_2\PMm{2}x:(\emptyset,\set{0},\T)\hascost{0}}}}
 & \hspace*{-20em}\PMm{2}\Te:(\set{0},\emptyset,\T)\hascost{0}
 & \PMm{2}\Te:(\set{1},\set{0},\T)\hascost{0}
}
\]

Let \(\PTE_4 = f\COL (\set{1},\emptyset,\prty_1), f\COL(\set{0},\emptyset,\prty_2), f\COL(\emptyset,\set{1},\prty_1)\)
where \(\prty_2 = (\set{0},\emptyset,\T)\to\T\).
Then, we have:
\[\hspace*{-3em}
\infers{\PMm{2} T\,A\,\Te+R(T(T A)):(\emptyset,\set{0,1},\T)\hascost{3}}{
\infers{\PMm{2} T\,A\,\Te:(\emptyset,\set{0,1},\T)\hascost{3}}{
\infers{\PMm{2} T\,A:(\emptyset,\set{1},\prty_1)\hascost{2}}{
\infers{
\PMm{2} T:(\emptyset,\emptyset,\prty_3)\hascost{1}}{\infers{
\PMm{2} \lambda f.\lambda x.f(f\,x):(\emptyset,\emptyset,\prty_3)\hascost{1}}{
\infers{\PTE_4\PMm{2}\lambda x.f(f\,x):(\emptyset,\set{1},\prty_1)\hascost{1}}{\cdots}}}
& \hspace*{-2em}\PMm{2} A: (\set{1},\emptyset,\prty_1)\hascost{0}
& \PMm{2} A: (\set{0},\emptyset,\prty_2)\hascost{0}
& \PMm{2} A: (\emptyset,\set{1},\prty_1)\hascost{1}
}
&\hspace*{-13em}\PMm{2}\Te:(\set{0},\emptyset,\T)\hascost{0}
 & \PMm{2}\Te:(\set{1},\set{0},\T)\hascost{0}
}}
\]
where \(\prty_3 = 
(\set{1},\emptyset,\prty_1)\land(\set{0},\emptyset,\prty_2)\land (\emptyset,\set{1},\prty_1)\to\prty_1\).

Let \(\PTE_5 = f\COL (\emptyset,\set{1},\prty_1)\).
Then we have:
\[
\infers{\PMm{2} A\,\Te + R(T A):(\emptyset,\set{0,1},\T)\hascost{3}}
{
\infers{\PMm{2}R(T A):(\emptyset,\set{0,1},\T)\hascost{3}}
{
\infers{\PMm{2}R:(\emptyset,\set{0},(\emptyset,\set{1},\prty_1)\to\T)\hascost{1}}{
\infers{\PMm{2}\lambda f.f\,\Te+R(T\,f):(\emptyset,\set{0},(\emptyset,\set{1},\prty_1)\to\T)\hascost{1}}{
\infers{\PTE_5\PMm{2} f\,\Te+R(T\,f):(\emptyset,\set{0,1},\T)\hascost{1}}{
\infers{\PTE_5\PMm{2} f\,\Te:(\emptyset,\set{0,1},\T)\hascost{1}}{
\PTE_5\PMm{2} f:(\emptyset,\set{1},\prty_1)\hascost{0}
& \PMm{2} \Te: (\set{0},\emptyset,\T)\hascost{0}
& \PMm{2} \Te: (\emptyset,\set{0},\T)\hascost{0}
}}}}
& \infers{\PMm{2}T A:(\emptyset,\set{1},\prty_1)\hascost 2}
{\cdots}
}}
\]
Let \(\PTE_6 = f\COL (\set{1},\emptyset,\prty_1),f\COL(\set{0},\emptyset,\prty_2),f\COL (\emptyset,\set{1},\prty_1)\).
Then we have:
\[\small\hspace*{-2em}
\infers{\PMm{2} S:(\emptyset,\set{0,1},\T)\hascost{3}}
{
\infers{\PMm{2} R\,A:(\emptyset,\set{0,1},\T)\hascost{3}}
{
\infers{\PMm{2} R:(\emptyset,\set{0},\prty_5)\hascost{2}}
{
\infers{\PMm{2} \lambda f.f\,\Te + R(T\,f):(\emptyset,\set{0},\prty_5)\hascost{2}}
{
\infers{\PTE_6\PMm{2} f\,\Te + R(T\,f):(\emptyset,\set{0,1},\T)\hascost{2}}
{
\infers{\PTE_6\PMm{2}R(T\,f):(\emptyset,\set{0,1},\T)\hascost{2}}
{
\infers{\PMm{2}R:(\emptyset,\set{0},(\emptyset,\set{1},\prty_1)\to\T)\hascost{1}}{
\cdots}
& \infers{\PTE_6\PMm{2}T\,f:(\emptyset,\set{1},\prty_1)\hascost 1}
{\cdots}
}}}}
& \hspace*{-7em}\PMm{2} A: (\set{1},\emptyset,\prty_1)\hascost{0}
& \PMm{2} A: (\set{0},\emptyset,\prty_2)\hascost{0}
& \PMm{2} A: (\emptyset,\set{1},\prty_1)\hascost{1}
}}
\]
where \(\prty_5 = 
 (\set{1},\emptyset,\prty_1)\land(\set{0},\emptyset,\prty_2)\land (\emptyset,\set{1},\prty_1)\to\T\).
This completes the construction of a type derivation corresponding to the reduction sequence
\(S\reds \Ta(\Ta\,\Te\,\Te)(\Ta\,\Te\,\Te)\).

The derivation above does not contain a pumpable part. 
To obtain a pumpable derivation, let us prepare the following derivations:
\[\small
\infers{\PMm{2} T: (\set{0},\emptyset, (\set{0},\emptyset,\prty_2)\to\prty_2)\hascost{0}}
{\infers
{\PMm{2} \lambda f.\lambda x.f(f\,x): (\set{0},\emptyset, (\set{0},\emptyset,\prty_2)\to\prty_2)\hascost{0}}
{\infers{f\COL(\set{0},\emptyset,\prty_2)\PMm{2}\lambda x.f(f\,x):(\set{0},\emptyset,\prty_2)\hascost{0}}
{\infers{f\COL(\set{0},\emptyset,\prty_2),x\COL(\set{0},\emptyset,\T)\PMm{2} f(f\,x):(\set{0},\emptyset,\T)\hascost{0}}
{f\COL(\set{0},\emptyset,\prty_2)\PMm{2}f:(\set{0},\emptyset,\prty_2)\hascost{0}
 & \infers{f\COL(\set{0},\emptyset,\prty_2),x\COL(\set{0},\emptyset,\T)\PMm{2} f\,x:(\set{0},\emptyset,\T)\hascost{0}}{\cdots}
}}}}
\]

\[\small
\infers{\PMm{2} T: (\set{1},\emptyset, (\set{1},\emptyset,\prty_1)\land (\set{0},\emptyset,\prty_2)\to\prty_1)\hascost{0}}
{\infers
{\PMm{2} \lambda f.\lambda x.f(f\,x): 
(\set{1},\emptyset, (\set{1},\emptyset,\prty_1)\land (\set{0},\emptyset,\prty_2)\to\prty_1)\hascost{0}}
{\infers{\PTE_9\PMm{2}\lambda x.f(f\,x):(\set{1},\emptyset,\prty_1)\hascost{0}}
{\infers{\PTE_9,\PTE_3\PMm{2} f(f\,x):(\set{1},\set{0},\T)\hascost{0}}
{\PTE_7\PMm{2}f:(\set{1},\emptyset,\prty_1)\hascost{0}
 & \infers{\PTE_8,x\COL(\set{0},\emptyset,\T)\PMm{2} f\,x:(\set{0},\emptyset,\T)\hascost{0}}{\cdots}
 & \infers{\PTE_7,\PTE_3\PMm{2} f\,x:(\set{1},\set{0},\T)\hascost{0}}{\cdots}
}}}}
\]
where \(\PTE_9 = 
f\COL (\set{1},\emptyset,\prty_1),f\COL(\set{0},\emptyset,\prty_2)\),
 \(\PTE_7 = f\COL (\set{1},\emptyset,\prty_1)\),
and \(\PTE_8 = f\COL(\set{0},\emptyset,\prty_2)\).

Thus, we can also construct the following derivation:
\[\small
\hspace*{-8em}
\infers{\PMm{2} S:(\emptyset,\set{0,1},\T)\hascost{4}}
{
\infers{\PMm{2} R\,A:(\emptyset,\set{0,1},\T)\hascost{4}}
{
\infers{\PMm{2} R:(\emptyset,\set{0},\prty_5)\hascost{3}}
{
\infers{\PMm{2} \lambda f.f\,\Te + R(T\,f):(\emptyset,\set{0},\prty_5)\hascost{3}}
{
\infers{\PTE_6\PMm{2} f\,\Te + R(T\,f):(\emptyset,\set{0,1},\T)\hascost{3}}
{
\infers{\PTE_6\PMm{2}R(T\,f):(\emptyset,\set{0,1},\T)\hascost{3}}
{
\infers{\PMm{2} R:(\emptyset,\set{0},\prty_5)\hascost{2}}
{\cdots}
& \infers{\PTE_9\PMm{2}T\,f:(\set{1},\emptyset,\prty_1)\hascost 0}{\cdots}
& \infers{\PTE_8\PMm{2}T\,f:(\emptyset,\set{0},\prty_2)\hascost 0}{\cdots}
& \infers{\PTE_6\PMm{2}T\,f:(\emptyset,\set{1},\prty_1)\hascost 1}{\cdots}
}}}}
& \hspace*{-18em}\PMm{2} A: (\set{1},\emptyset,\prty_1)\hascost{0}
& \PMm{2} A: (\set{0},\emptyset,\prty_2)\hascost{0}
& \PMm{2} A: (\emptyset,\set{1},\prty_1)\hascost{1}
}
}
\]
Thus, the subderivation of 
\(\PMm{2} R:(\emptyset,\set{0},\prty_5)\hascost{3}\) from 
\(\PMm{2} R:(\emptyset,\set{0},\prty_5)\hascost{2}\) is pumpable.

From the derivation, we obtain the following triples \((C,D,t)\):
\[
\begin{array}{ll}
C = \Hole\,t_{A,{(\set{1},\emptyset,\prty_1)}}\,t_{A,{(\set{0},\emptyset,\prty_2)}}\,t_{A,{(\emptyset,\set{1},\prty_1)}}\\
t_{A,{(\set{1},\emptyset,\prty_1)}} = 
\lambda x_{(\set{0},\emptyset,\T)}.\lambda x_{(\emptyset,\set{0},\T)}.\Ta\,x_{(\set{0},\emptyset,\T)} x_{(\emptyset,\set{0},\T)}\\
t_{A,{(\set{0},\emptyset,\prty_2)}} = 
\lambda x_{(\set{0},\emptyset,\T)}.\Ta\,x_{(\set{0},\emptyset,\T)}x_{(\set{0},\emptyset,\T)}\\
t_{A,{(\emptyset,\set{1},\prty_1)}} = t_{A,{(\set{1},\emptyset,\prty_1)}} \\
D = \lambda f_{(\set{1},\emptyset,\prty_1)}.\lambda f_{(\set{0},\emptyset,\prty_2)}.\lambda f_{(\emptyset,\set{1},\prty_1)}.
\Hole\,t_{T f, {(\set{1},\emptyset,\prty_1)}}\,t_{T f,{(\set{0},\emptyset,\prty_2)}}\,t_{T f,{(\emptyset,\set{1},\prty_1)}}\\
t_{T f, {(\set{1},\emptyset,\prty_1)}} = \lambda x_{(\set{0},\emptyset,\T)}.\lambda x_{(\emptyset,\set{0},\T)}.
f_{(\set{1},\emptyset,\prty_1)} 
(f_{(\set{0},\emptyset,\prty_2)}x_{(\set{0},\emptyset,\T)})(f_{(\set{1},\emptyset,\prty_1)}\,x_{(\set{0},\emptyset,\T)} x_{(\emptyset,\set{0},\T)})\\
t_{T f,{(\set{0},\emptyset,\prty_2)}} = \lambda x_{(\set{0},\emptyset,\T)}.
f_{(\set{0},\emptyset,\prty_2)}(f_{(\set{0},\emptyset,\prty_2)}x_{(\set{0},\emptyset,\T)})\\
t_{T f,{(\emptyset,\set{1},\prty_1)}} = \lambda x_{(\set{0},\emptyset,\T)}.\lambda x_{(\emptyset,\set{0},\T)}.
f_{(\set{1},\emptyset,\prty_1)}(f_{(\set{0},\emptyset,\prty_2)}x_{(\set{0},\emptyset,\T)})(f_{(\emptyset,\set{1},\prty_1)}\,x_{(\set{0},\emptyset,\T)} x_{(\emptyset,\set{0},\T)})\\
t = \lambda f_{(\set{1},\emptyset,\prty_1)}.\lambda f_{(\set{0},\emptyset,\prty_2)}.\lambda f_{(\emptyset,\set{1},\prty_1)}.
   t_{R,(\emptyset,\set{0},(\emptyset,\set{1},\prty_1)\to\T)} (t_{T f,{(\emptyset,\set{1},\prty_1)}})\\
t_{R,(\emptyset,\set{0},(\emptyset,\set{1},\prty_1)\to\T)} = 
\lambda f_{(\emptyset,\set{1},\prty_1)}. f_{(\emptyset,\set{1},\prty_1)}\,\Te\,\Te
\end{array}
\]
For the readability, let us rename variables and terms.
\[
\begin{array}{ll}
C = \Hole\,t_{A,0}\,t_{A,{1}}\,t_{A,{2}}\\
t_{A,0} = t_{A,2} =
\lambda x_0.\lambda x_1.\Ta\,x_0\, x_1\\
t_{A,1} = 
\lambda x_0.\Ta\,x_0\,x_0\\
D = \lambda f_0.\lambda f_1.\lambda f_2.
\Hole\,t_{T f, 0}\,t_{T f,{1}}\,t_{T f,{2}}\\
t_{T f, {0}} = \lambda x_0.\lambda x_1.
f_0
(f_{1}x_{0})(f_{0}\,x_{0} x_{1})\\
t_{T f,1} = \lambda x_0.
f_1(f_1\,x_0)\\
t_{T f,2} = \lambda x_0.\lambda x_1.
f_0(f_1\,x_0)(f_2\,x_0\, x_1)\\
t = \lambda f_0.\lambda f_1.\lambda f_2.
   t_R (t_{T f,2})\\
t_R = \lambda f_2. f_2\,\Te\,\Te
\end{array}
\]
Note that \(C\) and \(D\) are linear.

The corresponding OCaml code, and sample runs are given below.
\begin{quote}
\small
\begin{verbatim}
(* f0 = f_{({1},{},ro_1)}, f1 = f_{({0},{},ro_1)}, f2 = f_{({},{1},ro_1)} *)
(* x0 = x_{({0},{},T)}, x1 = x_{({},{0},T)} *)

(* tree constructors *)
type tree = A of tree * tree | E

(* context C *)
let tA0 = fun x0 x1 -> A(x0, x1)
let tA1 = fun x0 -> A(x0, x0)
let tA2 = tA0
let c g = g tA0 tA1 tA2

(* context D *)
let tT0 f0 f1 = fun x0 x1 -> f0 (f1 x0) (f0 x0 x1)
let tT1 f1 =  fun x0 -> f1 (f1 x0)
let tT2 f0 f1 f2 =  fun x0 x1 -> f0 (f1 x0) (f2 x0 x1)
let d g = fun f0 f1 f2 -> g (tT0 f0 f1) (tT1 f1) (tT2 f0 f1 f2)

(* term t *)
let tR = fun f2 -> f2 E E
let t = fun f0 f1 f2 -> tR(tT2 f0 f1 f2)

(*  sample execution 
# let t0 = c t;;
val t0 : tree = A (A (E, E), A (E, E))
# let t1 = c (d t);;
val t1 : tree =
  A (A (A (A (E, E), A (E, E)), A (A (E, E), A (E, E))),
   A (A (A (E, E), A (E, E)), A (A (E, E), A (E, E))))
# let t2 = c (d (d t));;
val t2 : tree =
  A
   (A
     (A
       (A
         (A (A (A (A (E, E), A (E, E)), A (A (E, E), A (E, E))),
  ...
*)
\end{verbatim}
\end{quote}

The following is a direction-annotated version of \((C, D,t)\):
\[
\begin{array}{ll}
C' = \Hole\,t'_{A,0}\,t'_{A,{1}}\,t'_{A,{2}}\\
t'_{A,0} = t'_{A,2} =
\lambda x_0.\lambda x_1.\directedT{1}{\Ta}\,x_0\, x_1 + \directedT{2}{\Ta}\,x_0\, x_1\\
t'_{A,1} = 
\lambda x_0.
\directedT{1}{\Ta}\,x_0 x_0 + \directedT{2}{\Ta}\,x_0\, x_0\\
D' = \lambda f_0.\lambda f_1.\lambda f_2.
\Hole\,t'_{T f, 0}\,t'_{T f,{1}}\,t'_{T f,{2}}\\
t'_{T f, {0}} = \lambda x_0.\lambda x_1.
f_0
(f_{1}x_{0})(f_{0}\,x_{0} x_{1})\\
t'_{T f,1} = \lambda x_0.
f_1(f_1\,x_0)\\
t'_{T f,2} = \lambda x_0.\lambda x_1.
f_0(f_1\,x_0)(f_2\,x_0\, x_1)\\
t' = \lambda f_0.\lambda f_1.\lambda f_2.
   t'_R (t'_{T f,2})\\
t'_R = \lambda f_2. f_2\,\directedT{0}{\Te}\,\directedT{0}{\Te}
\end{array}
\]
Let us construct a type derivation for \(C'\), \(D'\), and \(t'\).
A derivation for \(t'\) is:
\[
\infers{\PMm{2} t':
   (\emptyset,\set{0},\prty_5')\hascost{2}}
{\infers
{\PTE_6'\PMm{2} t'_R(t'_{Tf,2}): (\emptyset,\set{0,1},\T)\hascost{2}}
{\PMm{2} t'_R:(\emptyset,\set{0},(\emptyset,\set{1},\prty'_1)\to \T)\hascost{1}
& \PTE_6' \PMm{2} t'_{Tf,2}: (\emptyset,\set{1},\prty'_1)\hascost{1}}
}
\]
\[
\infers{\PMm{2} t'_R:(\emptyset,\set{0},(\emptyset,\set{1},\prty'_1)\to \T)\hascost{1}}
{\infers
 {f_2\COL(\emptyset,\set{1},\prty'_1)\PMm{2} f_2\,\directedT{0}{\Te}\,\directedT{0}{\Te}:(\emptyset,\set{0,1},\T)\hascost{1}}
{\infers{f_2\COL(\emptyset,\set{1},\prty'_1)\PMm{2} f_2\,\directedT{0}{\Te}:
  (\set{0},\set{1},(\emptyset,\set{0},\T)\to\T)\hascost{0}}
 {f_2\COL(\emptyset,\set{1},\prty'_1)\PMm{2} f_2: (\emptyset,\set{1},\prty'_1)\hascost{0}
 & \PMm{2}\directedT{0}{\Te}:(\set{0},\emptyset,\T)\hascost{0}}
 & \infers{\PMm{2}\directedT{0}{\Te}:(\emptyset,\set{0},\T)\hascost{0}}
   {\PMm{2}\directedT{0}{\Te}:(\set{0},\emptyset,\T)\hascost{0}}
}}
\]
\[\small
\infers{\PTE_6' \PMm{2} t'_{Tf,2}: (\emptyset,\set{1},\prty'_1)\hascost{1}}
{\infers{\PTE_6',\PTE'_3\PMm{2} f_0(f_1x_0)(f_2x_0x_1): (\emptyset,\set{0,1},\T)\hascost{1}}
 {{\PTE_9',\PTE'_1\PMm{2} f_0(f_1x_0):(\set{0,1},\emptyset,(\emptyset,\set{0},\T)\to\T)\hascost{0}}
&\hspace*{-5em}
\infers{\PTE_{10}',\PTE'_3\PMm{2} f_2x_0x_1: (\emptyset,\set{0,1},\T)\hascost{0}}
{\infers{\PTE_{10}',\PTE'_1\PMm{2}f_2x_0:(\emptyset,\set{1},(\emptyset,\set{0},\T)\to\T)\hascost{0}}
 {\PTE_{10}'\PMm{2}f_2: (\emptyset,\set{1},\prty'_1)\hascost{0}
  & \PTE'_1\PMm{2}x_0:(\set{0},\emptyset,\T)\hascost{0}}
& \PTE'_2\PMm{2}x_1:(\emptyset,\set{0},\T)\hascost{0}}
}}
\]
\[
\infers{\PTE_9',\PTE'_1\PMm{2} f_0(f_1x_0):(\set{0,1},\emptyset,(\emptyset,\set{0},\T)\to\T)\hascost{0}}
  {\PTE_7'\PMm{2}f_0:(\set{1},\emptyset,\prty'_1)\hascost{0}
  & \infers{\PTE_8',\PTE'_1\PMm{2} f_1x_0:(\set{0},\emptyset,\T)\hascost{0}}
    {\PTE_8'\PMm{2}f_1:(\set{0},\emptyset,\prty'_2\hascost{0})
     & \PTE_1\PMm{2}x_0:(\set{0},\emptyset,\T)\hascost{0}}}
\]
where
\[
\begin{array}{l}
\PTE_1' = x_0\COL (\set{0},\emptyset,\T)\\
\PTE_2' = x_1\COL(\emptyset,\set{0},\T)\\
\PTE_3' = x_0\COL (\set{0},\emptyset,\T), x_1\COL(\emptyset,\set{0},\T)\\
\PTE_6' = 
f_0\COL (\set{1},\emptyset,\prty'_1),f_1\COL(\set{0},\emptyset,\prty'_2),f_2\COL (\emptyset,\set{1},\prty'_1)\\
\PTE_7' = f_0\COL (\set{1},\emptyset,\prty'_1)\\
\PTE_8' = f_1\COL(\set{0},\emptyset,\prty'_2)\\
\PTE_9' = f_0\COL (\set{1},\emptyset,\prty'_1),f_1\COL(\set{0},\emptyset,\prty'_2)\\
\PTE_{10}' = f_2\COL (\emptyset,\set{1},\prty'_1)\\
\prty'_1 = (\set{0},\emptyset,\T)\to (\emptyset,\set{0},\T)\to\T\\
\prty'_2 = (\set{0},\emptyset,\T)\to\T\\
\prty_5' = (\set{1},\emptyset,\prty'_1)\to (\set{0},\emptyset,\prty'_2)\to (\emptyset,\set{1},\prty'_1)\to\T
\end{array}
\]

A derivation for \(C'\) is:
\[\small
\infers{\PMm{2}C': (\emptyset,\set{0,1},\T)\hascost{4}}
 {\PMm{2}\Hole: (\emptyset,\set{0},\prty_5)\hascost{3}
 & \PMm{2}t'_{A,0}:(\set{1},\emptyset,\prty_1')\hascost{0}
 & \PMm{2}t'_{A,1}:(\set{0},\emptyset,\prty_2')\hascost{0}
 & \PMm{2}t'_{A,2}:(\emptyset,\set{1},\prty_1')\hascost{1}}
\]
where
\[\small
\infers{\PMm{2}t'_{A,0}:(\set{1},\emptyset,\prty_1')\hascost{0}}
{\infers
{\PTE'_3\PMm{2} \directedT{1}{\Ta}\,x_0\, x_1 + \directedT{2}{\Ta}\,x_0\, x_1:(\set{1},\set{0},\T)\hascost{0}}
{\infers{\PTE'_3\PMm{2} \directedT{2}{\Ta}\,x_0\, x_1:(\set{1},\set{0},\T)\hascost{0}}
{\PTE_1'\PMm{2} x_0:(\set{0},\emptyset,\T) &
\PTE_2'\PMm{2} x_1:(\emptyset,\set{0},\T) 
}
}}
\]

\[\small
\infers{\PMm{2}t'_{A,1}:(\set{0},\emptyset,\prty_2')\hascost{0}}
{\infers
{\PTE'_1\PMm{2} \directedT{1}{\Ta}\,x_0\, x_0 + \directedT{2}{\Ta}\,x_0\, x_0:(\set{0},\emptyset,\T)\hascost{0}}
{\infers{\PTE'_1\PMm{2} \directedT{2}{\Ta}\,x_0\, x_0:(\set{0},\emptyset,\T)\hascost{0}}
{\PTE_1'\PMm{2} x_0:(\set{0},\emptyset,\T) 
}
}}
\]

\[\small
\infers{\PMm{2}t'_{A,2}:(\emptyset,\set{1},\prty_1')\hascost{1}}
{\infers
{\PTE'_3\PMm{2} \directedT{1}{\Ta}\,x_0\, x_1 + \directedT{2}{\Ta}\,x_0\, x_1:(\emptyset,\set{0,1},\T)\hascost{1}}
{\infers{\PTE'_3\PMm{2} \directedT{2}{\Ta}\,x_0\, x_1:(\emptyset,\set{0,1},\T)\hascost{1}}
{\PTE_1'\PMm{2} x_0:(\set{0},\emptyset,\T) &
\infers{\PTE_2'\PMm{2} x_1:(\emptyset,\set{0,1},\T)\hascost{0} }{\PTE_2'\PMm{2} x_1:(\emptyset,\set{0},\T)\hascost{0} }
}
}}
\]

A derivation for \(D'\) is:
\[\small
\infers{\PMm{2}D': (\emptyset,\set{0},\prty_5)\hascost{3}}
{\infers{\PTE_6'\PMm{2} \Hole\,t'_{Tf,0}t'_{Tf,1}t'_{Tf,2}:(\emptyset,\set{0,1},\T)\hascost{3}}
{\PMm{2}\Hole: (\emptyset,\set{0},\prty_5)\hascost{2}
 & \PTE_9'\PMm{2}t'_{Tf,0}:(\set{1},\emptyset,\prty_1')\hascost{0}
 & \PTE_8'\PMm{2}t'_{Tf,1}:(\set{0},\emptyset,\prty_2')\hascost{0}
 & \infers{\PTE_6'\PMm{2}t'_{Tf,2}:(\emptyset,\set{1},\prty_1')\hascost{1}}{\cdots}
}}
\]
where:
\[\small
\infers{\PTE_9'\PMm{2}t'_{Tf,0}:(\set{1},\emptyset,\prty_1')\hascost{0}}
{\infers{\PTE_9',\PTE'_3\PMm{2} f_0(f_1x_0)(f_0x_0x_1): (\set{1},\set{0},\T)\hascost{0}}
 {\PTE_9',\PTE'_1\PMm{2} f_0(f_1x_0):(\set{0,1},\emptyset,(\emptyset,\set{0},\T)\to\T)\hascost{0}
& \hspace*{-6em}
\infers{\PTE_{7}',\PTE'_3\PMm{2} f_0x_0x_1: (\set{1},\set{0},\T)\hascost{0}}
{\infers{\PTE_7',\PTE'_1\PMm{2}f_0x_0:(\set{0,1},\emptyset,(\emptyset,\set{0},\T)\to\T)\hascost{0}}
 {\PTE_7'\PMm{2}f_0:(\set{1},\emptyset,\prty_1')\hascost{0} & \PTE'_1\PMm{2}x_0:(\set{0},\emptyset,\T)\hascost{0}}
 & \PTE'_2\PMm{2}x_1:(\emptyset,\set{0},\T)\hascost{0}}
}}
\]

\[\small
\infers{\PTE_8'\PMm{2}t'_{Tf,1}:(\set{0},\emptyset,\prty_2')\hascost{0}}
{\infers{\PTE_8',\PTE'_1\PMm{2} f_1(f_1x_0): (\set{0},\emptyset,\T)\hascost{0}}
 {\PTE_8'\PMm{2} f_1:(\set{0},\emptyset,\prty_2')\hascost{0}
\infers{\PTE_{8}',\PTE'_1\PMm{2} f_1x_0: (\set{0},\emptyset,\T)\hascost{0}}
 {\PTE_8'\PMm{2}f_1:(\set{0},\emptyset,\prty_2')\hascost{0}
  & \PTE'_1\PMm{2}x_0:(\set{0},\emptyset,\T)\hascost{0}}}}
\]

It turns out that the derivation for \(C'[D'[t']]\) is pumpable. We obtain
the following triple \((G,H,u)\) that satisfies the condition of Lemma~\ref{lem:parys5}.
\[
\begin{array}{ll}
G = \Hole\,u_{A,0}\,u_{A,{1}}\,u_{A,{2}}\\
u_{A,0} = u_{A,2} =
\lambda x_0.\lambda x_1.\directedT{2}{\Ta}\,x_0\, x_1\\
u_{A,1} = 
\lambda x_0.\directedT{2}{\Ta}\,x_0\, x_0\\
H = \lambda f_0.\lambda f_1.\lambda f_2.
\Hole\,u_{T f, 0}\,u_{T f,{1}}\,u_{T f,{2}}\\
u_{T f, {0}} = \lambda x_0.\lambda x_1.
f_0
(f_{1}x_{0})(f_{0}\,x_{0} x_{1})\\
u_{T f,1} = \lambda x_0.f_1(f_1\,x_0)\\
u_{T f,2} = \lambda x_0.\lambda x_1.
f_0(f_1\,x_0)(f_2\,x_0\, x_1)\\
u = \lambda f_0.\lambda f_1.\lambda f_2.
   u_R (u_{T f,2})\\
u_R = \lambda f_2. f_2\,\directedT{0}{\Te}\,\directedT{0}{\Te}
\end{array}
\]
In this case, \((G,H,u)\) is the same as \((C,D,t)\) except that \(\Ta\) and \(\Te\) are annotated
with the directions \(2\) and \(0\) respectively.

Now, let us apply the induction in the proof of Lemma~\ref{lem:orderInd}.
The triple \((G_p,H_p,u_p)\) is given  by:
\[
\begin{array}{ll}
G_p = \Hole\,u'_{A,0}\,u'_{A,{1}}\,u'_{A,{2}}\\
u'_{A,0} = u'_{A,2} =
\lambda x_0.\lambda x_1.\directedT{2}{\Ta}x_1\\
u'_{A,1} = 
\lambda x_0.\directedT{2}{\Ta}\,x_0\\
H_p = \lambda f_0.\lambda f_1.\lambda f_2.
\Hole\,u'_{T f, 0}\,u'_{T f,{1}}\,u'_{T f,{2}}\\
u'_{T f, {0}} = \lambda x_0.\lambda x_1.
f_0
(f_{1}x_{0})(f_{0}\,x_{0} x_{1})\\
u'_{T f,1} = \lambda x_0.f_1(f_1\,x_0)\\
u'_{T f,2} = \lambda x_0.\lambda x_1.
f_0(f_1\,x_0)(f_2\,x_0\, x_1)\\
u_p = \lambda f_0.\lambda f_1.\lambda f_2.
   u'_R (u'_{T f,2})\\
u'_R = \lambda f_2. f_2\,\directedT{0}{\Te}\,\directedT{0}{\Te}
\end{array}
\]

By applying the first of the word-to-leaves transformation (in Section~\ref{sec:FstTrans}),
we obtain the following order-1 contexts and terms.
\[
\begin{array}{ll}
G'_p = \Hole\,u''_{A,0}\,u''_{A,{2}}\\
u''_{A,0} = u''_{A,2} =
 \br\,\directedT{2}{\Ta}\,\Te\\
H'_p = \lambda f_0.\lambda f_2.
\Hole\,u''_{T f, 0}\,u''_{T f,{2}}\\
u''_{T f, {0}} = %
\br\,f_0\,(\br\,f_{0}\,\Te)\\
u''_{T f,2} = %
\br\,f_0\,(\br\,f_2\,\Te)\\
u'_p = \lambda f_0.\lambda f_2.
   u''_R (u''_{T f,2})\\
u''_R = \lambda f_2. \br\,f_2\,\directedT{0}{\Te}
\end{array}
\]
By applying the second transformation (in Section~\ref{sec:SndTrans}) to
eliminate redundant occurrences of \(\Te\), we obtain:
\[
\begin{array}{ll}
G_l = \Hole\,v_{A,0}\,v_{A,{2}}\\
v_{A,0} = v_{A,2} = \directedT{2}{\Ta}\\
H_l = \lambda f_0.\lambda f_2.
\Hole\,v_{T f, 0}\,v_{T f,{2}}\\
v_{T f, {0}} = %
\br\,f_0\,f_{0}\\
v_{T f,2} = %
\br\,f_0\,f_2\\
u_l = \lambda f_0.\lambda f_2.
   v_R (v_{T f,2})\\
v_R = \lambda f_2. f_2
\end{array}
\]
The corresponding OCaml code, and sample runs are as follows.
\begin{quote}
\begin{verbatim}
type tree1 = Br of tree1*tree1 | LeafE | LeafA
(* G_l *)
let vA0 = LeafA
let vA2 = vA0
let gl g = g vA0 vA2
(* H_l *)
let vTf0 f0 = Br(f0,f0)
let vTf2 f0 f2 = Br(f0,f2)
let hl g = fun f0 f2 -> g (vTf0 f0) (vTf2 f0 f2)
(* u_l *)
let vR = fun f2 -> f2
let ul = fun f0 f2 -> vR (vTf2 f0 f2)

(* sample runs
# gl ul;;
- : tree1 = Br (LeafA, LeafA)
# gl(hl ul);;
- : tree1 = Br (Br (LeafA, LeafA), Br (LeafA, LeafA))
# gl(hl (hl ul));;
- : tree1 =
Br (Br (Br (LeafA, LeafA), Br (LeafA, LeafA)),
 Br (Br (LeafA, LeafA), Br (LeafA, LeafA)))
*)
\end{verbatim}
\end{quote}

By applying the induction hypothesis of Lemma~\ref{lem:orderInd}, 
we obtain a pumping sequence (where we can choose \(j=0\) and \(k=1\) in this case):
\[
\leaves(\tree(\wl{G}[\wl{H}^{0}[\wl{u}]])) \she 
\leaves(\tree(\wl{G}[\wl{H}^{1}[\wl{u}]])) \she 
\leaves(\tree(\wl{G}[\wl{H}^{2}[\wl{u}]])) \she \cdots,
\]
and hence we also obtain:
\begin{equation*}
\tree(\tp{G}[{\tp{H}}^{0}[\tp{u}]]) \she 
\tree(\tp{G}[{\tp{H}}^{1}[\tp{u}]]) \she 
\tree(\tp{G}[{\tp{H}}^{2}[\tp{u}]]) \she \cdots.
\end{equation*}
We also happen to have a trivial periodic sequence:
\begin{equation*}
{{H}}^{0}[{u}] \he_\sty
{{H}}^{1}[{u}] \he_\sty 
{{H}}^{2}[{u}] \he_\sty \cdots
\end{equation*}
for \(\sty = (\T\to\T\to\T)\to (\T\to\T)\to(\T\to\T\to\T)\to\T\).
Thus, we also have
\[
\tree({G}[{{H}}^{0}[{u}]]) \she 
\tree({G}[{{H}}^{1}[{u}]]) \she 
\tree({G}[{{H}}^{2}[{u}]]) \she \cdots.
\]
Therefore, we finally obtain:
\[ \tree(C[t]) \she \tree(C[D[t]]) \she \tree(C[D^2[t]]) \she \cdots.\]

\subsubsection{Other examples}
One may expect that the triple 
\((C,D,t)\) obtained from a pumpable derivation always satisfies
\[ |\tree(C[t])| < |\tree(C[D[t]])| < |\tree(C[D[D[t]]])| <\cdots\]
because the counter values of the derivations for \(C[D^i[t]]\) strictly 
increase with respect to \(i\).
Below we give examples where that is not the case
(although there do exist \(j,k\)
such that \(|\tree(C[D^j[t]]) < |\tree(C[D^{j+k}[t]])|\)
as guaranteed by Lemma~\ref{lem:orderInd}).

Consider the following order-1 grammar:
\[
\begin{array}{l}
S \to A(a(a(e))).\\
A\ x\to x.\\
A\ x\to a(A(e)).
\end{array}
\]
We have the following pumpable derivation (where we omit the target of the transformation):
\[
\infers{\PMm{1} S: (\emptyset,\set{0},\T)\hascost 2}
{\infers{\PMm{1} A(a(a(e))): (\emptyset,\set{0},\T)\hascost 2}
  {\infers{\PMm{1} A: (\emptyset,\set{0},(\set{0},\emptyset,\T)\to\T)\hascost 2}
     {\infers{\PMm{1} \lambda x.a(A(e)): (\emptyset,\set{0},(\set{0},\emptyset,\T)\to\T)\hascost 2}
      {\infers{x\COL(\set{0},\emptyset,\T)\PMm{1}a(A(e)):(\emptyset,\set{0},\T)\hascost 2}{
       \infers{\PMm{1}a(A(e)):(\emptyset,\set{0},\T)\hascost 2}
       {\infers{\PMm{1}A(e): (\emptyset,\set{0},\T)\hascost 1}
        {\infers{\PMm{1} A:(\emptyset,\set{0},(\set{0},\emptyset,\T)\to\T)\hascost 1}
         {\infers{\PMm{1} \lambda x.x:(\emptyset,\set{0},(\set{0},\emptyset,\T)\to\T)\hascost 1}
          {\infers{x\COL(\set{0},\emptyset,\T)\PMm{1} x:(\emptyset,\set{0},\T)\hascost 1}
           {\infers{x\COL(\set{0},\emptyset,\T)\PMm{1} x:(\set{0},\emptyset,\T)\hascost 0}{}}}
          } 
        &
         \infers{\PMm{1} e:(\set{0},\emptyset,\T)\hascost 0}{} }}
       }
       }
     }
     & \infers{\PMm{1} a(a(e)):(\set{0},\emptyset,\T)\hascost 0}{\cdots}
}}
\]
The triple \((C,D,t)\) obtained from the above derivation is:
\[
C = \Hole\,(a(a(e)))\qquad D = \lambda x.a(\Hole\,e)\qquad t = \lambda x.x.
\]
Notice that \(C[t] \reds a(a(e))\) and \(C[D[t]]\reds a((\lambda x.x)e)\reds a(e)\);
thus, \(|\tree(C[t])|\not<|\tree(C[D[t]])|\).
Note, however, that since \(\tree(C[D^k[t]])= a^k(e)\) for \(k>0\),  we have:
\[ \tree(C[D[t]]) \she \tree(C[D^2[t]])\she \tree(C[D^3[t]])\she \cdots.\]
Thus, for the triple above, Lemma~\ref{lem:orderInd} holds for \(j=k=1\).

Consider the following order-2 grammar:
\[
\begin{array}{l}
S \to A(\lambda xy.y).\\
A\ f\to f\ e\ (a^2\,e).\\
A\ f\to a(A(\lambda xy.f\,y\,x).\\
\end{array}
\]

The following is a pumpable derivation:
\[
\small
\infers{\PMm{2} S: (\emptyset,\set{0,1},\T)\hascost 2}
{\infers{\PMm{2} A(\lambda xy.y): (\emptyset,\set{0,1},\T)\hascost 2}
 {\infers{\PMm{2}A:(\emptyset,\set{0,1},
                   (\set{0},\emptyset,\prty_1)\to \T)\hascost 2}
  {\infers{\PMm{2} \lambda f.a(A(\lambda xy.f\,y\,x):(\emptyset,\set{0,1},
                   (\set{0},\emptyset,\prty_1)\to \T)\hascost 2}
   {\infers{f\COL(\set{0,\emptyset,\prty_1})\PMm{2}a(A(\lambda xy.f\,y\,x):(\emptyset,\set{0,1},\T)\hascost 2}
   {\infers{f\COL(\set{0,\emptyset,\prty_1})\PMm{2}A(\lambda xy.f\,y\,x):(\emptyset,\set{0,1},\T)\hascost 1}
   {\infers{\PMm{2}A:(\emptyset,\set{0,1},
                   (\set{0},\emptyset,\prty_1)\to \T)\hascost 1}
   {\infers {\PMm{2}\lambda f.f\,e\,(a^2\,e):(\emptyset,\set{0,1},
                   (\set{0},\emptyset,\prty_1)\to \T)\hascost 1}
   {\infers{f\COL(\set{0},\emptyset,\prty_1)\PMm{2}f\,e\,(a^2\,e):(\emptyset,\set{0,1},\T)\hascost 1 }
    {\infers{f\COL(\set{0},\emptyset,\prty_1)\PMm{2}f\,e\,(a^2\,e):(\set{0},\emptyset,\T)\hascost 0}{\cdots} }}
  }
    & \infers{f\COL(\set{0,\emptyset,\prty_1})\PMm{2} \lambda xy.f\,y\,x:(\set{0},\emptyset,\prty_1)\hascost 0}{\cdots}}
  }}}
   & 
  \hspace*{-6em}\infers{\PMm{2}\lambda xy.y:(\set{0},\emptyset,\prty_1)
       \hascost 0}{\cdots}
}}
\]
where
\(\prty_1 = (\set{0},\emptyset,\T)\to (\set{0},\emptyset,\T)\to \T)\).
The triple obtained from the above derivation is \((C,D,t)\) where:
\[
C = \Hole\lambda xy.y
\qquad
D = \lambda f.a(\Hole(\lambda xy.f\,y\,x))
\qquad
t = \lambda f.f\,e\,(a^2\,e).
\]
We have:
\[
\tree(C[D^i[t]]) = \left\{
 \begin{array}{ll}
    a^{i+2}(e) & \mbox{ if $i$ is even}\\
    a^{i}(e) & \mbox{ if $i$ is odd}
\end{array}
\right.
\]
Thus, \(|\tree(C[D^{2i}[t]])| \not<|\tree(C[D^{2i+1}[t]])|\).
We, however, have:
\[ C[D^j[t]] \she C[D^{j+k}[t]] \she C[D^{j+2k}[t]] \she \cdots\]
for \(j=0\), \(k=2\).
\anp
\section{Word-to-leaves Transformations}
\label{sec:wordleaf}

Here we prove Lemma~\ref{lem:iteLinPres}.
As explained in Section~\ref{sec:wordleafbody},
we use a modified version of the transformation given in~\cite{asada_et_al:LIPIcs:2016:6246}.
The transformation in~\cite{asada_et_al:LIPIcs:2016:6246}
consists of two steps, and
we restrict each of the two transformations to \(\stlambda\)-terms so that
the two restricted transformations return again \(\stlambda\)-terms.
We call the two steps \emph{first} and \emph{second} transformations,
and the composite the \emph{whole} transformation.
To distinguish the transformations in this paper and those in~\cite{asada_et_al:LIPIcs:2016:6246},
we call %
the latter \emph{original first/second/whole transformations}.

In Sections~\ref{sec:FstTrans} and~\ref{sec:SndTrans} we give the definitions of the first and second
transformations, respectively;
all the definitions in Sections~\ref{sec:FstTrans} and~\ref{sec:SndTrans}
except for Figures~\ref{fig:FirstTransformation} and~\ref{fig:SecondTransformation} are
taken from~\cite{asada_et_al:LIPIcs:2016:6246}.
Then, we show that these transformations preserve meaning (in Section~\ref{sec:propOfFSTrans}) and linearity (in Section~\ref{sec:linPres}).

\subsection{First Transformation}
\label{sec:FstTrans}

The first transformation is applied to order-\((n+1)\) \(\stlambda\)-terms
of a word alphabet and outputs order-\(n\) \(\stlambda\)-terms of a \(\br\)-alphabet.
Constants of type \(\T \to \T\) before the transformation
have type \(\T\) after the transformation.
This first transformation achieves the purpose of the whole transformation
except that an output term might not be \(\Te\)-free \(\br\)-tree: e.g.,
\((\lambda x . x) (\Ta\,\Te)\) is transformed to \(\br\, \Te\, (\br\, \Ta\, \Te)\),
whose leaves, \(\Te\,\Ta\,\Te\), have extra \(\Te\).
Such extra \(\Te\)'s will be removed by the second transformation.
We write \(\br\,t\,s\) also as \(t \ibr s\) and write \(((t_1 \ibr t_2)\dots \ibr t_m)\) as \(t_1 \ibr \dots \ibr t_m\).

Same as the original first transformation,
for technical convenience, we assume below that, for every type \(\sty\) occurring in
the simple type derivation of an input \(\stlambda\)-term,
if \(\sty\) is of the form \(\T\ra\sty'\), then \(\order(\sty')\leq 1\).
This does not lose generality, since any function \(\lambda x\COL\T.t\) of type
\(\T\ra\sty'\) with \(\order(\sty')>1\) can be replaced by the term \(\lambda x'\COL\T\ra\T.[x'\Te/x]t\)
of type \((\T\to\T)\to\sty'\) (without changing the order of the term),
and any term \(t\) of type \(\T\) can be replaced by the term \((\lambda x,y. x)\,t\) of type 
\(\T\ra\T\) (see~\cite[Appendix~D]{asada_et_al:LIPIcs:2016:6246} for the detail).
This preprocessing transformation preserves order, meaning, and linearity.

For the first transformation, we use the following intersection types:
\[
\begin{array}{l}
 \uty ::= \T \mid \ity \ra \uty
\qquad
 \ity ::= \uty_1\land \cdots \land \uty_k \quad (k \ge 0)
\end{array}
\]
We write \(\top\) for \(\uty_1\land \cdots \land \uty_k\) when \(k=0\).
We assume some total order \(<\) on intersection types,
and require that \(\uty_1<\cdots < \uty_k\) whenever 
\(\uty_1\land \cdots \land \uty_k\) occurs in an intersection type.
Intuitively, if a function \(f\) has type \(\uty_1 \land \cdots \land \uty_k \to \uty\),
then \(f\) uses an argument (in \(k\)-number of different ways), and if \(f\) has \(\top \to \uty\),
then \(f\) does not use an argument.

We introduce two refinement relations \(\brr{\uty}{\sty}\) and
\(\urr{\uty}{\sty}\).
The relations are defined as follows, by mutual induction; \(k\) may be \(0\).\\[1ex]
\InfruleS{0.55}{\urr{\uty_j}{\sty} \andalso j\in\set{1,\ldots,k} \\
\brr{\uty_i}{\sty} \mbox{ (for each $i\in\set{1,\ldots,k}\setminus\set{j}$)}
}{\urr{\uty_1\land \cdots \land \uty_k}{\sty}}
\InfruleS{0.45}{
\brr{\uty_i}{\sty} \mbox{ (for each $i\in\set{1,\ldots,k}$)}
}{\brr{\uty_1\land \cdots \land \uty_k}{\sty}}\\[1ex]
\InfruleS{0.2}{}{\urr{\T}{\T}}
\InfruleS{0.25}{
\brr{\ity}{\sty} \andalso \urr{\uty}{\sty'}
}{
\urr{\ity \ra \uty}{\sty \ra \sty'}
}
\InfruleS{0.25}{
\urr{\ity}{\sty} \andalso \urr{\uty}{\sty'}
}{
\brr{\ity \ra \uty}{\sty \ra \sty'}
}
\InfruleS{0.25}{
\brr{\ity}{\sty} \andalso \brr{\uty}{\sty'}
}{
\brr{\ity \ra \uty}{\sty \ra \sty'}
}\\
\noindent
A type \(\uty\) is called \myemph{balanced} if \(\brr{\uty}{\sty}\) for some \(\sty\),
and called \myemph{unbalanced} if \(\urr{\uty}{\sty}\) for some \(\sty\). 
Intuitively, unbalanced types describe trees or 
closures that contain the end of a word (i.e., symbol \(\Te\)).
Intersection types that are neither balanced nor unbalanced 
are considered ill-formed, and excluded out. 
For example, the type \(\T\to\T\to\T\) (as an intersection type) is ill-formed;
since \(\T\) is unbalanced, \(\T\to\T\) must also be unbalanced according to 
the rules for arrow types, but it is actually balanced.
In fact, no term can have the intersection type \(\T\to\T\to\T\) in a word grammar.
We write \(\uty\DCOL\sty\) if \(\brr{\uty}{\sty}\) or \(\urr{\uty}{\sty}\).

We introduce a type-directed transformation relation \(\UE\pf t:\uty\tr u\) for terms,
where \(\UE\) is a set of type bindings of the form \(x\COL\uty\), called a \emph{type
environment}, \(t\) is a source term, and \(u\) is the image of the transformation.
We write \(\UE_1\cup\UE_2\) for the union of \(\UE_1\) and \(\UE_2\); it is defined only
if, whenever \(x\COL\uty\in\UE_1\cap \UE_2\), \(\uty\) is balanced. In other words,
unbalanced types are treated as linear types, whereas balanced ones as non-linear (or
idempotent) types. 
Intuitively, if a type is unbalanced then an argument of the type is used linearly, i.e.,
used exactly once in the reduction;
while if a type is balanced then an argument of the type may be copied and used in many places.
We write \(\balanced(\TE)\) if \(\uty\) is balanced for every
\(x\COL\uty\in\TE\).

The relation \(\UE \pf t:\uty\tr u\) is defined inductively by the rules in Figure~\ref{fig:FirstTransformation}.
Note that, for a given ground closed term \(t\), a term \(u\) such that \(\pf t : \T \tr u\) is not necessarily unique, while
the original transformation return unique output (gathering all by using non-deterministic choice).
However this does not matter since the result is semantically the same as shown in
Section~\ref{sec:propOfFSTrans}.
By dropping the transformation part ``\(\tr u\)'' from the rules in Figure~\ref{fig:FirstTransformation},
we obtain a standard form of intersection type system.
Though the transformation is not deterministic, it is deterministic as a transformation of
a type derivation tree of the (simplified) intersection type system to a term.
The following is a standard fact on intersection type systems.
\begin{lemma}[subject reduction/expansion]
\label{lem:fstSubConv}
For \(t \red t'\),
\(\TE \pf t : \uty \tr u\) for some \(u\)
iff
\(\TE \pf t' : \uty \tr u\) for some \(u\).
\end{lemma}

For a word \(a_1\cdots a_n\),
we define term \(\wtt{(a_1\cdots a_n)}\) inductively by:
\(\wtt{\epsilon} = \Te\) and \(\wtt{(as)} = \TT{br}\,a\,\wtt{s}\).
\begin{lemma}[\mbox{\cite[Lemma~10]{asada_et_al:LIPIcs:2016:6246}}]
\label{lem:leaf-fwd}
\( \pf a_1(\cdots(a_n\,\Te)\cdots):\T \tr \wtt{(a_1\cdots a_n)}\).
\end{lemma}
By Lemmas~\ref{lem:fstSubConv} and~\ref{lem:leaf-fwd},
for any closed ground term \(t\), we have \(\pf t : \T \tr u\) for some \(u\).

We define \(\utytosty{\uty\DCOL\sty}\) by:
\[
\begin{array}{l}
\utytosty{\uty\DCOL\sty} =
\T \qquad \mbox{ (if $\order(\sty) \le 1$)}
\\
\utytosty{(\uty_1\land\cdots\land\uty_k \ra \uty)\DCOL(\sty_0\to\sty)} =
\utytosty{\uty_1\DCOL\sty_0} \ra \dots \ra \utytosty{\uty_k\DCOL\sty_0} \ra \utytosty{\uty\DCOL\sty}
\\
\mspace{520mu}\mbox{ (if $\order(\sty_0\to\sty)>1$)}
\end{array}
\]
\begin{lemma}
\label{lem:orderDecreaseFst}
For an order-\((n+1)\) term \(t\) with \(x_1 : \sty_1, \dots, x_m : \sty_m \pK t : \sty\),
if
\begin{align*}&
x_1 : \uty^1_{1}, \dots, x_1 : \uty^1_{k_1}, \dots, x_m : \uty^m_{1}, \dots, x_m : \uty^m_{k_m} \pf t : \uty
 \tr u
\\&
(\land_{i \le k_1}\uty^1_i \to \dots \to \land_{i \le k_m} \uty^m_i \to \uty) \DCOL 
(\sty_1 \to \dots \to \sty_m \to \sty)
\end{align*}
then \(u\) is an order-\(n\) term with
\[
\Bigg(
\begin{aligned}
(x_1)_{\uty^1_{1}} &: \utytosty{\uty^1_{1} \DCOL \sty_1}, &&\dots,& (x_1)_{\uty^1_{k_1}} &: \utytosty{\uty^1_{k_1}
 \DCOL \sty_1} ,&& \dots, 
\\
(x_m)_{\uty^m_{1}} &: \utytosty{\uty^m_{1} \DCOL \sty_m}, &&\dots,& (x_m)_{\uty^m_{k_m}} &: \utytosty{\uty^m_{k_m} \DCOL \sty_m} &&
\end{aligned}
\ \Bigg)
\pK u : \utytosty{\uty \DCOL \sty}
\]
\end{lemma}
\begin{proof}
By straightforward induction on 
\(x_1 : \uty^1_{1}, \dots, x_1 : \uty^1_{k_1}, \dots, x_m : \uty^m_{1}, \dots, x_m : \uty^m_{k_m} \pf t : \uty \tr u\).
\end{proof}

\begin{figure}
\typicallabel{FTr-Const}
\InfruleSR{0.3}{FTr-Var}{ %
}{ %
x\COL\uty \pf x\COL\uty\tr x_{\uty}}
\ \ \ \
\InfruleSR{0.285}{FTr-Const0}
{ %
}{ %
\pf \Te\COL\T\tr \Te}
\ \ \ \
\InfruleSR{0.335}{FTr-Const1}
{\TERMS(a)=1 %
}{ %
\pf a\COL\T\ra\T\tr a}

\infrule[FTr-App1]
  {\TE_0 \pf s\COL \uty_1\land\cdots\land \uty_k\ra\uty\tr v \\
\TE_i\pf t\COL\uty_i\tr u_i \text{ and }\uty_i\neq \T
\text{ (for each $i \in \set{1,\dots,k}$)} %
}
{\TE_0\cup\TE_1\cup\cdots\cup\TE_k\pf s\,t:\uty\tr vu_1\cdots u_k}

\infrule[FTr-App2]
{\TE_0\pf s\COL \T\ra\uty\tr v\andalso
\TE_1\pf t\COL\T\tr u}
{\TE_0\cup\TE_1 \pf s\,t\COL\uty\tr \TT{br}\,v\,u}

\infrule[FTr-Abs1]
{\TE,x\COL\uty_1,\ldots,x\COL\uty_k \pf t:\uty\tr u\andalso x \notin \envdom{\TE}\\
\uty_i\neq \T\mbox{ for each $i\in\set{1,\ldots,k}$}}
{\TE\pf \lambda x.t:\uty_1\land\cdots\land\uty_k\ra\uty\tr 
\lambda x_{\uty_1}\cdots\lambda x_{\uty_k}.u}

\infrule[FTr-Abs2]
{\TE,x\COL\T\pf t:\uty\tr u}
{\TE\pf \lambda x.t:\T\ra\uty \tr [\Teps/x_{\T}]u}

\caption{First transformation}
\label{fig:FirstTransformation}
\end{figure}

\subsection{Second Transformation}
\label{sec:SndTrans}

As explained above, the purpose of the second transformation is
to remove extra \(\Te\). 
Inputs and outputs of the second transformation is \(\stlambda\)-terms
of a \(\br\)-alphabet, and
output terms generate \(\Te\)-free \(\br\)-trees or \(\Te\).

For the second transformation,
we use the following intersection types:
\[ \tty ::= \Tempty \mid \Tplus \mid \tty_1\land\cdots\land \tty_k\ra \tty\]
Intuitively, \(\Tempty\) describes trees consisting of only \(\TT{br}\) and
\(\Teps\), and \(\Tplus\) describes trees that have at least one non-\(\Teps\) leaf.
We again assume some total order \(<\) on intersection types, and require that
whenever we write \(\tty_1\land\cdots\land \tty_k\), \(\tty_1<\cdots< \tty_k\) holds.
We define the refinement relation \(\tty\DCOL\sty\) inductively by:
(i) \(\Tempty\DCOL \T\), (ii) \(\Tplus\DCOL\T\), and (iii) \(
(\tty_1\land\cdots\land \tty_k\ra \tty)\DCOL (\sty_1\to\sty_2)\) if
\(\tty\DCOL\sty_2\) and \(\tty_i\DCOL\sty_1\) for every \(i\in\set{1,\ldots,k}\).
We consider only types \(\tty\) such that \(\tty\DCOL\sty\) for some \(\sty\).
For example, we forbid an ill-formed type like \(\Tplus\land (\Tplus\to\Tplus) \to \Tplus\).

We introduce a type-based transformation relation \(\TTE\ps t:\tty \tr u\),
where \(\TTE\) is a type environment (i.e., a set of bindings of the form \(x\COL\tty\)),
 \(t\) is a source term, \(\tty\) is the type of \(t\), and \(u\) is the result of transformation.
The relation is defined inductively by the rules in Figure~\ref{fig:SecondTransformation}.
As the first transformation, the results of the second transformation are not necessarily unique,
but unique semantically as shown in Section~\ref{sec:propOfFSTrans}.
Also, we have the following standard fact for this intersection type system.
\begin{lemma}[subject reduction/expansion]
\label{lem:sndSubConv}
For \(t \red t'\),
\(\TTE \ps t : \tty \tr u\) for some \(u\)
iff
\(\TTE \ps t' : \tty \tr u\) for some \(u\).
\end{lemma}

We define \(\ttytosty{\tty}\) by:
\[
\begin{array}{l}
\ttytosty{\Tempty}=\ttytosty{\Tplus}=\T\qquad
\ttytosty{\tty_1\land\cdots\land\tty_k\ra\tty} =
 \ttytosty{\tty_1}\ra\cdots\ra\ttytosty{\tty_k}\ra\ttytosty{\tty}
\end{array}
\]
\begin{lemma}
\label{lem:orderDecreaseSnd}
For an order-\((n+1)\) term \(t\) with \(x_1 : \sty_1, \dots, x_m : \sty_m \pK t : \sty\),
if
\begin{align*}&
x_1 : \tty^1_{1}, \dots, x_1 : \tty^1_{k_1}, \dots, x_m : \tty^m_{1}, \dots, x_m : \tty^m_{k_m} \ps t : \tty
 \tr u
\\&
(\land_{i \le k_1}\tty^1_i \to \dots \to \land_{i \le k_m} \tty^m_i \to \tty) \DCOL 
(\sty_1 \to \dots \to \sty_m \to \sty)
\end{align*}
then \(u\) is an order-\(n\) term with
\[
(x_1)_{\tty^1_{1}} : \ttytosty{\tty^1_{1}}, \dots, (x_1)_{\tty^1_{k_1}} : \ttytosty{\tty^1_{k_1}} , \dots, 
(x_m)_{\tty^m_{1}} : \ttytosty{\tty^m_{1}}, \dots, (x_m)_{\tty^m_{k_m}} : \ttytosty{\tty^m_{k_m}} 
\pK u : \ttytosty{\tty}
\]
\end{lemma}
\begin{proof}
By straightforward induction on 
\(x_1 : \tty^1_{1}, \dots, x_1 : \tty^1_{k_1}, \dots, x_m : \tty^m_{1}, \dots, x_m : \tty^m_{k_m} \ps t : \tty \tr u\).
\end{proof}

\begin{figure}
\typicallabel{STr-NT}
\InfruleSR{0.3}{STr-Var}{} %
{%
x\COL\tty \ps x\COL\tty\tr x_{\tty}}
\ \ 
\InfruleSR{0.3}{STr-Const0}
{}
{%
\ps \Teps\COL\Tempty\tr \Teps}
\ \ 
\InfruleSR{0.35}{STr-Const1}
{\TERMS(a)=0\andalso a\neq \Teps}
{%
\ps a\COL\Tplus\tr a}\\

\infrule[STr-Const2]
{\TTE_0 \ps t_0\COL\tty_0\tr u_0\andalso \TTE_1 \ps t_1\COL\tty_1\tr u_1\\
  (\TTE,\tty,u) = \left\{
\arraycolsep=1.3pt
  \begin{array}{rllllllllllllllllll}
    (\TTE_0 \cup \TTE_1,& \Tplus,&  \TT{br}\,u_0\,u_1 ) && \mbox{ if \(\tty_0=\tty_1=\Tplus\)}\\
    (\TTE_i,            & \Tplus,&  u_i               ) && \mbox{ if \(\tty_i=\Tplus\) and \(\tty_{1-i}=\Tempty\)}\\
    (\eset,             & \Tempty,& \Teps             ) && \mbox{ if \(\tty_0=\tty_1=\Tempty\)}\\
  \end{array}\right.
}
{\TTE  \ps \TT{br}\,t_0\,t_1\COL\tty\tr u}

\infrule[STr-App]
  {\TTE_0 \ps s\COL \tty_1\land\cdots\land \tty_k\ra\tty\tr v\andalso
  \TTE_i \ps t\COL\tty_i\tr u_i \mbox{ (for each $i\in\set{1,\ldots,k}$)}}
{\TTE_0 \cup \TTE_1 \cup \dots \cup \TTE_k \ps st:\tty\tr vu_1\cdots u_k}

\infrule[STr-Abs]
{\TTE,x\COL\tty_1,\ldots,x\COL\tty_k \ps t:\tty\tr u}
{\TTE\ps \lambda x.t:\tty_1\land\cdots\land\tty_k\ra\tty\tr 
\lambda x_{\tty_1}\cdots\lambda x_{\tty_k}.u}
\caption{Second transformation}
\label{fig:SecondTransformation}
\end{figure}

\subsection{First and Second Transformations Preserve Meaning}
\label{sec:propOfFSTrans}

Here we show that the first and second transformations preserve the meaning of a term.
Before that, let us review the corresponding theorems in~\cite{asada_et_al:LIPIcs:2016:6246};
here for a word \(w\), we write \(\remeps{w}\) for the word obtained
by removing all the occurrences of \(\Te\) in \(w\), and
\(\remeps{\Lang}\) for \(\set{\remeps{w} \mid w \in \Lang}\).
\begin{theorem}[\mbox{\cite[Theorem~7]{asada_et_al:LIPIcs:2016:6246}}]
\label{th:tr1-correctness}
If \(\pof \GRAM\tr \GRAM'\), then \(\Wlang(\GRAM) = \remeps{\LLang(\GRAM')}\).
\end{theorem}
\begin{theorem}[\mbox{\cite[Theorem~9]{asada_et_al:LIPIcs:2016:6246}}]
\label{th:tr2-correctness}
If \(\pos \GRAM\tr \GRAM'\), then 
\(\remeps{\LLang(\GRAM)}=\LLangE(\GRAM')\).
\end{theorem}

\begin{lemma}
\label{lem:fst-crrectness}
If \(\pf t : \T \tr u\), then \(\word(\tree(t)) = \remeps{\leaves(\tree(u))}\).
\end{lemma}
\begin{proof}
Figure~\ref{fig:FirstTransformationHOG} shows the definition of the original first transformation \(\pof\)
given in~\cite{asada_et_al:LIPIcs:2016:6246}.
Here, meta variables \(U\) and \(V\) represent a set \(\set{u_1,\dots,u_k}\) of terms, which is nothing but
the non-deterministic choice \(u_1+ \cdots + u_k\).
Figure~\ref{fig:FirstTransformationDHOG} are obtained from Figure~\ref{fig:FirstTransformationHOG}
by removing non-deterministic choices; the change is just for \rname{Tr1-App1}, \rname{Tr1-App2},
 \rname{Tr1-Set}, and \rname{Tr1-Gram}.
We write \(\pofd\) for the transformation by Figure~\ref{fig:FirstTransformationDHOG}.
When applied to recursion-free higher-order grammars,
Figure~\ref{fig:FirstTransformationDHOG} is essentially the same as Figure~\ref{fig:FirstTransformation};
the condition \((**)\) in \rname{Tr1-Gram-Det} in Figure~\ref{fig:FirstTransformationDHOG}
ensures that \(\pofd\) produces a deterministic higher-order grammar,
and as in the case of \(\pf\), the fact that \(\pofd\) produces at least one deterministic higher-order grammar 
is shown by subject expansion.

For a ground closed \(\stlambda\)-term \(t\), let \(\pf t : \T \tr u\)
and let \(\G\) and \(\G'\) be the deterministic higher-order grammars obtained from \(t\) and \(u\),
respectively.
Then clearly we have \(\pofd \G \tr \G'\); also let \(\pof \G \tr \G''\).
Now, since the set of rules in Figure~\ref{fig:FirstTransformationDHOG} is just a subset of
that in Figure~\ref{fig:FirstTransformationHOG},
and since the rule \rname{Tr1-Gram} in Figure~\ref{fig:FirstTransformationHOG} gathers all
the derived rewriting rules, \(\G'\) is a ``syntactical determinization'' of the choices in \(\G''\);
i.e., \(\G' \sle \G''\) where we write \(\sle\) for the least congruence relation such that \(s \sle s+t\) and \(t \sle
s+t\) and we regard a grammar as a ground closed \(\lambdap\)-term with \(Y\)-combinator.
Hence, we have \((\set{\tree(\G')} =) \Lang(\G') \subseteq \Lang(\G'')\).

Now by Theorem~7 in~\cite{asada_et_al:LIPIcs:2016:6246}, \(\Wlang(\G) = \remeps{\LLang(\G'')}\). Hence,
\[
\set{\word(\tree(t))} = \set{\word(\tree(\G))} = \Wlang(\G) = \remeps{\LLang(\G'')} \supseteq 
\remeps{\set{\leaves(\tree(\G'))}} = \remeps{\set{\leaves(\tree(u))}}
\]
and therefore \(\word(\tree(t)) = \remeps{\leaves(\tree(u))}\).
\end{proof}

\begin{lemma}
\label{lem:snd-crrectness}
If \(\ps t : \T \tr u\), then
\[
\remeps{\leaves(\tree(t))} = 
\begin{cases}
\empword & (\leaves(\tree(u)) = \Te)
\\
\leaves(\tree(u)) & (\leaves(\tree(u)) \neq \Te).
\end{cases}
\]
\end{lemma}
\begin{proof}
The proof is quite similar to that of Lemma~\ref{lem:fst-crrectness}: the points are
(i) Figure~\ref{fig:SecondTransformation} is a subset of the set of rules of the original second
 transformation (up to the embedding of \(\lambda\)-terms to higher-order grammars),
(ii) a result of the second transformation is a deterministic \(\lambda\)-term, and hence its meaning, i.e., the language,
 is a singleton,
(iii) subset relation between singleton sets implies equality between their elements.
\end{proof}

\begin{figure}
\typicallabel{Tr1-Const}
\InfruleSR{0.4}{Tr1-Var}{%
}
{%
x\COL\uty \p x\COL\uty\tr x_{\uty}}
\InfruleSR{0.55}{Tr1-NT}
{\uty\DCOL \NONTERMS(A)%
}
{%
\p A\COL\uty \tr A_\uty}\\[1ex]

\InfruleSR{0.4}{Tr1-Const0}
{%
}
{%
\p \Te\COL\T\tr \Te}
\InfruleSR{0.55}{Tr1-Const1}
{\TERMS(a)=1%
}
{%
\p a\COL\T\ra\T\tr a}

\infrule[Tr1-App1-Det]
  {\TE_0 \p s\COL \uty_1\land\cdots\land \uty_k\ra\uty\tr v \\
\TE_i\p t\COL\uty_i\tr u_i \text{ and }\uty_i\neq \T
\text{ (for each $i \in \set{1,\dots,k}$)} %
}
{\TE_0\cup\TE_1\cup\cdots\cup\TE_k\p st:\uty\tr vu_1\cdots u_k}

\infrule[Tr1-App2-Det]
{\TE_0\p s\COL \T\ra\uty\tr v\andalso
\TE_1\p t\COL\T\tr u}
{\TE_0\cup\TE_1 \p st\COL\uty\tr \TT{br}\,v\,u}

\infrule[Tr1-Abs1]
{\TE,x\COL\uty_1,\ldots,x\COL\uty_k \p t:\uty\tr u\andalso x \notin \envdom{\TE}\\
\uty_i\neq \T\mbox{ for each $i\in\set{1,\ldots,k}$}}
{\TE\p \lambda x.t:\uty_1\land\cdots\land\uty_k\ra\uty\tr 
\lambda x_{\uty_1}\cdots\lambda x_{\uty_k}.u}

\infrule[Tr1-Abs2]
{\TE,x\COL\T\p t:\uty\tr u}
{\TE\p \lambda x.t:\T\ra\uty \tr [\Teps/x_{\T}]u}

\infrule[Tr1-Rule]
{\emptyset \p \lambda x_1.\cdots\lambda x_k.t:\uty
 \tr \lambda x'_1.\cdots\lambda x'_\ell.u \andalso \uty\DCOL\NONTERMS(A)}
{(A\,x_1\,\cdots\,x_k\Hra t)\tr
 (A_{\uty}\,x'_1\,\cdots\,x'_{\ell} \Hra u)}

\infrule[Tr1-Gram-Det]
{
{\TERMS' = \set{\TT{br}\mapsto 2,\Te\mapsto 0}\cup\set{a\mapsto 0\mid \TERMS(a)=1}}
\\
\NONTERMS' = \set{A_\uty\COL \utytosty{\uty\DCOL\sty} \mid \NONTERMS(A)=\sty \land \uty\DCOL\sty}\andalso
\RULES' \subseteq \set{ r' \mid \exists r\in \RULES. r\tr r'}
\\
\forall (A\,x_1\,\cdots\,x_k\Hra t) \in \RULES. 
\forall \uty\DCOL\NONTERMS(A).
\exists ! r' \in \RULES'. (A\,x_1\,\cdots\,x_k\Hra t) \tr r' \quad(**)
}
{\p(\TERMS,\NONTERMS,\RULES,S)\tr 
({\TERMS'},\NONTERMS',\RULES',S_\T)}

\caption{First transformation for deterministic higher-order grammar}
\label{fig:FirstTransformationDHOG}
\end{figure}

\begin{figure}
\typicallabel{Tr1-Const}
\InfruleSR{0.4}{Tr1-Var}{\balanced(\TE)}
{\TE,x\COL\uty \p x\COL\uty\tr x_{\uty}}
\InfruleSR{0.55}{Tr1-NT}
{\uty\DCOL \NONTERMS(A)\andalso \balanced(\TE)}
{\TE\p A\COL\uty \tr A_\uty}\\[1ex]

\InfruleSR{0.4}{Tr1-Const0}
{\balanced(\TE)}
{\TE\p \Te\COL\T\tr \Te}
\InfruleSR{0.55}{Tr1-Const1}
{\TERMS(a)=1\andalso \balanced(\TE)}
{\TE \p a\COL\T\ra\T\tr a}

\infrule[Tr1-App1]
  {\TE_0 \p s\COL \uty_1\land\cdots\land \uty_k\ra\uty\tr v \\
\TE_i\p t\COL\uty_i\tr U_i \text{ and }\uty_i\neq \T
\text{ (for each $i \in \set{1,\dots,k}$)} %
}
{\TE_0\cup\TE_1\cup\cdots\cup\TE_k\p st:\uty\tr vU_1\cdots U_k}

\infrule[Tr1-App2]
{\TE_0\p s\COL \T\ra\uty\tr V\andalso
\TE_1\p t\COL\T\tr U}
{\TE_0\cup\TE_1 \p st\COL\uty\tr \TT{br}\,V\,U}

\infrule[Tr1-Set]
{\TE\p t\COL \uty \tr u_i\mbox{ (for each $i\in\set{1,\ldots,k}$)}
{\andalso k \ge 1}\\
}
{\TE \p t\COL\uty \tr \set{u_1,\ldots,u_k}}

\infrule[Tr1-Abs1]
{\TE,x\COL\uty_1,\ldots,x\COL\uty_k \p t:\uty\tr u\andalso x \notin \envdom{\TE}\\
\uty_i\neq \T\mbox{ for each $i\in\set{1,\ldots,k}$}}
{\TE\p \lambda x.t:\uty_1\land\cdots\land\uty_k\ra\uty\tr 
\lambda x_{\uty_1}\cdots\lambda x_{\uty_k}.u}

\infrule[Tr1-Abs2]
{\TE,x\COL\T\p t:\uty\tr u}
{\TE\p \lambda x.t:\T\ra\uty \tr [\Teps/x_{\T}]u}

\infrule[Tr1-Rule]
{\emptyset \p \lambda x_1.\cdots\lambda x_k.t:\uty
 \tr \lambda x'_1.\cdots\lambda x'_\ell.u \andalso \uty\DCOL\NONTERMS(A)}
{(A\,x_1\,\cdots\,x_k\Hra t)\tr
 (A_{\uty}\,x'_1\,\cdots\,x'_{\ell} \Hra u)}

\infrule[Tr1-Gram]
{
{\TERMS' = \set{\TT{br}\mapsto 2,\Te\mapsto 0}\cup\set{a\mapsto 0\mid \TERMS(a)=1}}
\\
\NONTERMS' = \set{A_\uty\COL \utytosty{\uty\DCOL\sty} \mid \NONTERMS(A)=\sty \land \uty\DCOL\sty}\andalso
\RULES' = \set{ r' \mid \exists r\in \RULES. r\tr r'}
}
{\p (\TERMS,\NONTERMS,\RULES,S)\tr 
({\TERMS'},\NONTERMS',\RULES',S_\T)}

\caption{Original first transformation for higher-order grammar}
\label{fig:FirstTransformationHOG}
\end{figure}

\anp
\subsection{First and Second Transformations Preserve Linearity}
\label{sec:linPres}

Here we prove Lemma~\ref{lem:iteLinPres}, by showing that the first and second transformations preserve linearity.
Roughly speaking, this consists of showing two lemmas:
(i) simulation of a reduction of an input term of the (first or second) transformation by reductions of an output term, and
(ii) any linear normal form is transformed to a linear term.
For \(u \red u'\), if \(u'\) is linear then so is \(u\), and hence
(i) and (ii) imply the preservation of linearity.
However, for the first transformation, actually (i) does not hold due to 
the complicated rules \rname{FTr-App2} and \rname{FTr-Abs2};
so we relax (i) to
(i)': if \(t\red t'\) and \(t\) is transformed to \(u\),
then \(t'\) is transformed to some \(u'\) such that
if \(u'\) is linear so is \(u\).
Here we concentrate on the first transformation; 
the proof for the second transformation is analogous.
In the rest of this section, we call a call-by-name normal form a \emph{normal form}.

\begin{lemma}
\label{lem:env}
For \(\TE \pf t :\uty' \tr u\),
\(x_\uty\) occurs in \(u\) iff \((x:\uty) \in \TE\).
If \(\uty\) is unbalanced and \((x: \uty) \in \TE\), then 
\(x_\uty\) occurs in \(u\) exactly once.
\end{lemma}
\begin{proof}
By straightforward induction on \(\TE \pf t :\uty' \tr u\).
\end{proof}

\begin{lemma}[substitution]\label{lem:substFst}
Given
\(\TE_0, x \COL \uty_1,\dots,x\COL \uty_k \pf s \COL \uty \tr v\) 
where \(x \notin \dom(\TE_0)\) and \(k \ge 0\), and given
\(\TE_i \pf t \COL \uty_i \tr u_i\) for each \(\ind{i}{k}\), 
if \(\TE_0\cup \cdots \cup \TE_k\) is well-defined, then
we have
\[
\TE_0\cup \cdots \cup \TE_k \pf [t/x]s \COL \uty \tr 
[u_i/x_{\uty_i}]_{\ind{i}{k}} v\,.
\]
\end{lemma}
\begin{proof}
By straightforward induction on \(\TE_0, x \COL \uty_1,\dots,x\COL \uty_k \pf s \COL \uty \tr v\).
\end{proof}

For intersection type systems, also de-substitution is standard:
\begin{lemma}[de-substitution]\label{lem:desubstFst}
Given \(\TE \pf [t/x]s \COL \uty \tr v'\), %
there exist \(k \ge 0\), \((\TE_i)_{0\le i \le k}\),
\((\uty_i)_{1\le i \le k}\), \((u_i)_{1 \le i \le k}\), and \(v\) such that
\begin{enumerate}
\item\label{item:desubs}
\(\TE_0, x \COL \uty_1,\dots,x\COL \uty_k \pf s \COL \uty \tr v\)
\item\label{item:desubt}
\(\TE_i\pf t \COL \uty_i \tr u_i  \quad  (1 \le i \le k)\)
\item\label{item:env}
\(\TE = \TE_0 \cup \cdots \cup \TE_k\)
\item\label{item:subs}
\(v' =
[u_i/x_{\uty_i}]_{\ind{i}{k}} v
\).
\end{enumerate}
\end{lemma}
\begin{proof}
By induction on \(s\) and case analysis on the
last rule used for deriving \(\TE \pf [t/x]s \COL \uty \tr v'\).
\end{proof}

We call a type \(\uty\) \emph{inhabited} if there exist \(s\) and \(v\) such that
\(\pf s: \uty \tr v\).
\begin{lemma}
\label{lem:unblin}
If \(z : \uty, x: \T \pf t : \T \tr u\), \(u \reds u'\), and \(\uty\) is inhabited, 
then \(x_\T\) occurs in \(u'\) exactly once.
\end{lemma}
\begin{proof}
Since there exist \(s\) and \(v\) such that
\(\pf s: \uty \tr v\),
we have \(x:\T \pf [s/z]t : \T \tr [v/z_\uty]u\) by Lemma~\ref{lem:substFst} above.
Then we use the results in~\cite{asada_et_al:LIPIcs:2016:6246}:
Lemmas~22,~24-(1), and the linearity of the type \(\rt\).
\end{proof}

The following lemma is a kind of subject reduction as well as a kind of left-to-right simulation.
\begin{lemma}
\label{lem:subRedLin}
If \(z: \uty \pf t : \T \tr u\), 
\(t \red t'\), and
\(\uty\) is inhabited,
then 
there exists \(u'\) such that
\(z: \uty \pf t' : \T \tr u'\) and if 
\(u'\) is linear
then so is \(u\).
\end{lemma}
\begin{proof}
The proof proceeds by induction on term \(t\) and by case analysis on the head of \(t\).
The case \(t= a s\) is clear by induction hypothesis.
We consider the other
case that \(t= (\lambda x . t_0)t_1\,t_2\dots,t_m\) (\(m \ge 1\)).
We further perform case analysis on the transformation rule used
for the application \((\lambda x . t_0)t_1\): \rname{FTr-App1} or \rname{FTr-App2}.
\begin{itemize}
\item Case of \rname{FTr-App1}:
Let the simple type of \(\lambda x . t_0\) be
\[
\sty_1 \to \dots \to \sty_m \to \T
\]
and let \(\sty_\ell\) be the first type of order-0 among \(\sty_i\).
Because of the assumption on simple types, \(\sty_i = \T\) for \(i \ge \ell\).
We have
\[
\begin{aligned}&
z: \uty \pf (\lambda x . t_0)t_1\,t_2\dots t_m : \T \tr 
\big((\lambda x_{\uty_1} \cdots x_{\uty_{k_1}} . u_0)
\appseq{u^1_i}{i\le k_1}\cdots\appseq{u^{\ell-1}_i}{i\le k_{\ell-1}}\big)
\ibr u^{{\ell}} \ibr \cdots \ibr u^m
\\&
(\lambda x . t_0)t_1\,t_2\cdots t_m \red
([t_1/x]t_0)t_2\cdots t_m
\\&
z: \uty \pf ([t_1/x]t_0)t_2\cdots t_m : \T \tr 
\big(([u^1_{i} / x_{\uty_{k_i}}]_{i\le k_1}u_0)
\appseq{u^2_i}{i\le k_2}\cdots\appseq{u^{\ell-1}_i}{i\le k_{\ell-1}}\big)
\ibr u^{{\ell}} \ibr \cdots \ibr u^m
\end{aligned}
\]
Then we have
\begin{align*}
& \big((\lambda x_{\uty_1} \cdots x_{\uty_{k_1}} . u_0)
\appseq{u^1_i}{i\le k_1}\cdots\appseq{u^{\ell-1}_i}{i\le k_{\ell-1}}\big)
\ibr u^{{\ell}} \ibr \cdots \ibr u^m
\\ \reds\ &
\big(([u^1_{i} / x_{\uty_{k_i}}]_{i\le k_1}u_0)
\appseq{u^2_i}{i\le k_2}\cdots\appseq{u^{\ell-1}_i}{i\le k_{\ell-1}}\big)
\ibr u^{{\ell}} \ibr \cdots \ibr u^m
\end{align*}
where if the latter is linear then so is the former.
\item Case of \rname{FTr-App2}:
In this case the type of \((\lambda x . t_0)\) is of the form
\(\T \to \top \to \dots \to \top \to \T\),
and we have
\[
\begin{aligned}&
z: \uty \pf (\lambda x . t_0)t_1\,t_2\dots,t_m : \T \tr 
\br\,([\Te/x_\T]u_0)\, u_1
\\&
(\lambda x . t_0)t_1\,t_2\dots,t_m \red
([t_1/x]t_0)t_2\dots,t_m
\\&
z: \uty \pf ([t_1/x]t_0)t_2\dots,t_m : \T \tr 
[u_1 / x_{\T}]u_0
\end{aligned}
\]
By Lemma~\ref{lem:env},
\(x_o\) occurs in \(u_0\) exactly once.
From now we assume that \([u_1 / x_{\T}]u_0\) is linear,
and show that \(\br\,([\Te/x_\T]u_0)\, u_1\) is also linear.
We further perform case analysis:
\begin{itemize}
\item Case where \(z_\uty\) occurs in \(u_1\):
We have
\[
z: \uty \pf t_1 : \T \tr u_1
\]
and the normal form of \(u_1\) is of the form \(E_1[z_\uty]\).

The normal form of \(u_0\) is of the form either 
\(E_0[z_\uty]\) or \(E_0[x_\T]\).
If the former, then
\[
[u_1/x_\T]u_0 \reds [u_1/x_\T](E_0[z_\uty])
=
([u_1/x_\T]E_0)[z_\uty]
\]
where the last term is normal form since \([u_1/x_\T]E_0\) is an evaluation context.
Since \([u_1/x_\T]u_0\) is linear, \(z_\uty\) occurs in
\(([u_1/x_\T]E_0)[z_\uty]\) exactly once.
Since \(z_\uty\) occurs in \(u_1\), \(x_\T\) must not occur in \(E_0\).
Thus \(x_\T\) does not occur in \(E_0[z_\uty]\) and hence by Lemma~\ref{lem:unblin}, neither in \(u_0\).
This is a contradiction since 
\(x_\T\) occurs in \(u_0\) by Lemma~\ref{lem:env}.
Thus, the normal form of \(u_0\) is of the form \(E_0[x_\T]\).

Then, we have
\[
[u_1/x_\T]u_0 \reds [u_1/x_\T](E_0[x_\T])
=
([u_1/x_\T]E_0)[u_1]
\reds
([u_1/x_\T]E_0)[E_1[z_\uty]]
\]
where note that \([u_1/x_\T]E_0\) is an evaluation context
and the last term is a normal form.
Since
\([u_1/x_\T]u_0 \reds [u_1/x_\T](E_0[x_\T])\) is linear,
\(z_\uty\) does not occur in
\([u_1/x_\T]E_0\) nor in \(E_1\).
Since \(z_\uty\) occurs in \(u_1\), 
\(x_\T\) does not occur in \(E_0\).
Hence \(E_0[x_\T]\) is closed; let \(E_0[x_\T] \reds \pi\).

Now we have
\[
\br\,([\Te/x_\T]u_0)\, u_1
\reds
\br\,([\Te/x_\T](E_0[x_\T]))\, u_1
=
\br\,(E_0[\Te])\, u_1
\reds
\br\,\pi\, E_1[z_\uty]
\]
where the last term is a normal form and contains exactly one \(z_\uty\).
Thus, \(\br\,([\Te/x_\T]u_0)\, u_1\) is linear.
\item Case where \(z_\uty\) does not occur in \(u_1\):
In this case we have
\begin{align*}&
z: \uty, x : \T \pf t_0 : \top \to \dots \to \top \to \T \tr u_0
\\&
\pf t_1 : \T \tr u_1
\end{align*}
and hence \(u_1\) is closed; let \(u_1 \reds \pi\).
The normal form of \(u_0\) is of the form either
\(E_0[x_\T]\) or \(E_0[z_\uty]\).

\begin{itemize}
\item
Case of \(u_0 \reds E_0[x_\T]\):
By Lemma~\ref{lem:unblin}, \(E_0\) does not contain \(x_\T\), and
we have
\[
[u_1/x_\T]u_0 \reds [u_1/x_\T](E_0[x_\T])
= E_0[u_1]
\reds E_0[\pi]
\]
Here we put a label on \(\pi\) and then reduce further to a normal form \(N\):
\[
[u_1/x_\T]u_0 
\reds E_0[\pi^{\lab}]
\reds N
\]
Let \(N'\) be the term obtained by replacing \(\pi^{\lab}\) in \(N\) with \(\Te\).
Then we have
\[
E_0[\Te] \reds N'
\]
and \(N'\) is also a normal form.
Further, by the linearity of \([u_1/x_\T]u_0 \), \(z_\uty\) occurs in \(N\) exactly once,
and so does in \(N'\).
Now we have
\[
\br\,([\Te/x_\T]u_0)\, u_1
\reds
\br\,([\Te/x_\T](E_0[x_\T]))\, u_1
=
\br\,(E_0[\Te])\, u_1
\reds
\br\,N'\, u_1
\]
where the last term is a normal form and contains exactly one \(z_\uty\).
Thus \(\br\,([\Te/x_\T]u_0)\, u_1\) is linear.
\item
Case of \(u_0 \reds E_0[z_\uty]\):
We have
\[
[u_1/x_\T]u_0 \reds [u_1/x_\T](E_0[z_\uty])
=
([u_1/x_\T]E_0)[z_\uty]
\]
where the last term is a normal form, and by the linearity of \([u_1/x_\T]u_0\), 
\(z_\uty\) does not occur in \(E_0\).
Now we have
\[
\br\,([\Te/x_\T]u_0)\, u_1
\reds
\br\,([\Te/x_\T](E_0[z_\uty]))\, u_1
=
\br\,(([\Te/x_\T]E_0)[z_\uty])\, u_1
\]
where the last term is a normal form and contains exactly one \(z_\uty\).
Thus \(\br\,([\Te/x_\T]u_0)\, u_1\) is linear.
\end{itemize}
\end{itemize}
\end{itemize}
\end{proof}

\begin{lemma}
\label{lem:strengthening}
For \(\TE \pf t : \uty \tr u\),
if \(x\) does not occur in \(t\), then \(x \notin \dom(\TE)\).
\end{lemma}
\begin{proof}
By straightforward induction on term \(t\).
\end{proof}

\begin{lemma}
\label{lem:linStrNorm}
For a linear normal form \(N\),
if \(x : \uty_1, \dots, x : \uty_k \pf N : \T \tr u\),
then \(k=1\) and \(u\) is a linear normal form.
\end{lemma}
\begin{proof}
By straightforward induction on open normal form \(N\) and by Lemma~\ref{lem:strengthening}.
\end{proof}

\begin{lemma}
\label{lem:linPres}
If \(x:\uty_1,\dots,x:\uty_k \pf t : \T \tr u\)
and \(t\) is linear,
then \(k=1\).

Further if \(\uty_1\) is inhabited, %
then \(u\) is linear.
\end{lemma}
\begin{proof}
Let \(N\) be the normal form of \(t\).
By Lemma~\ref{lem:fstSubConv}, we have
\(x:\uty_1,\dots,x:\uty_k \pf N : \T \tr u'\)
for some \(u'\).
By Lemma~\ref{lem:linStrNorm}, \(k=1\).

When \(\uty_1\) is inhabited,
for the reduction sequence
\[
t = t_0 \red t_1 \red \dots \red t_\ell = N
\]
by Lemma~\ref{lem:subRedLin}, there exist \(u_0,\dots,u_\ell\) such that
\begin{align*}&
u_0 = u
\\&
x : \uty_1 \pf t_i :\T \tr u_i \quad(i=0,\dots,\ell)
\\&
\text{if \(u_{i+1}\) is linear then so is \(u_i\)} \quad(i=0,\dots,\ell-1)
\end{align*}
By Lemma~\ref{lem:linStrNorm}
\(u_\ell\) is linear, and hence \(u_0 = u\) is linear.
\end{proof}

As a corollary of Lemmas~\ref{lem:linPres} and~\ref{lem:substFst}, we have:
\begin{lemma}
\label{lem:iteLin}
For \(z : \uty \pf t : \T \tr u\) and
\(y:\uty_1 ,\dots, y:\uty_k \pf s : \uty \tr v\),
if \([s/z]t\) is linear, \(k=1\).
\end{lemma}

A \(\stlambda\)-term \(t\) is \emph{relevant} if,
for any subterm \(\lambda x . s\) in \(t\), \(x\) occurs in \(s\).
\begin{lemma}
\label{lem:relevance}
For \(\TE \pf t : \uty \tr u\), \(u\) is relevant, and if \(u\) is linear on some variable \(z\),
then \(z\) occurs in \(u\) exactly once.
\end{lemma}
\begin{proof}
The former part is shown by straightforward induction on \(\TE \pf t : \uty \tr u\), with
using Lemma~\ref{lem:env}.
The latter part is clear from the former part since, during reduction of a relevant term toward the normal form,
the number of free occurrences of each variable does not decrease.
\end{proof}

\begin{lemma}
\label{lem:iteLinPresFst}
Given order-\(n\) \(\stlambda\)-contexts \(C\), \(D\), and order-\(n\) \(\stlambda\)-term \(t\)
such that 
\begin{itemize}
\item
the constants in \(C,\,D,\,t\) are in a word alphabet
\item
\(\set{ \tree(C[D^{\ell_i}[t]]) \mid i \ge 0 }\) is infinite for any strictly increasing sequence
 \((\ell_i)_i\).
\item
\(C\) and \(D\) are %
 linear
\end{itemize}
there exist order-\((n-1)\) \(\stlambda\)-contexts \(G\), \(H\), order-\((n-1)\) \(\stlambda\)-term \(u\),
and some constant numbers \(c,\,d \ge 1\)
such that 
\begin{itemize}
\item
the constants in \(G,\,H,\,u\) are in a \(\br\)-alphabet
\item
\(\word(\tree(C[D^{ci+d}[t]])) = \remeps{\leaves(\tree(G[H^i[u]]))} \quad (i \ge 0)\)
\item
\(G\) and \(H\) are %
 linear.
\end{itemize}
\end{lemma}
\begin{proof}
Since \(C[D^j[t]]\) is a ground closed term, we have \(\pf C[D^j[t]] : \T \tr u_0\) for some \(u_0\).
By Lemmas~\ref{lem:desubstFst},~\ref{lem:linPres},~\ref{lem:iteLin}, and the %
 linearity, we have:
\begin{align*}&
z: \uty_0 \pf C[z] : \T \tr v_0
\\&
z: \uty_{i} \pf D[z] : \uty_{i-1} \tr v_{i} \quad(i=1,\dots,j)
\\&
\pf t : \uty_j \tr u'
\\&
u_0=
[[\dots
[u'/z_{\uty_j}]v_j
\dots
/z_{\uty_1}]v_1
/z_{\uty_0}]v_0
\end{align*}

Now we use a ``pumping'' argument:
Let \(\sty\) be the simple type of \(t\)
and let \(j_0\) be the number of intersection types \(\uty\) such that \(\uty \DCOL \sty\).
When \(j = j_0+1\), among \(\uty_j\) above,
we have \(\uty_{j_1} = \uty_{j_2}\) for some \(1 \le {j_1} < {j_2} \le j\).
Then we define
\begin{align*}
G &\defe 
[\Hole/z_{\uty_{j_1}}]([[
\dots
[v_{j_1}/z_{\uty_{{j_1}-1}}]v_{{j_1}-1}
\dots
/z_{\uty_1}]v_1
/z_{\uty_0}]v_0)
\\
H &\defe
[\Hole/z_{\uty_{{j_2}}}](
[
\dots
[v_{j_2}/z_{\uty_{{j_2}-1}}]v_{{j_2}-1}
\dots
/z_{\uty_{{j_1}+1}}]v_{{j_1}+1})
\\
u &\defe
[
\dots
[u'/z_{\uty_{j}}]v_{j}
\dots
/z_{\uty_{{j_2}+1}}]v_{{j_2}+1}
\end{align*}
where \([\Hole / z_{\uty_{j}}]\) represents the replacement of unique \(z_{\uty_{j}}\) with \(\Hole\),
and this uniqueness follows from Lemma~\ref{lem:relevance}.
Since \(\uty_{j_1} = \uty_{j_2}\), for any \(i \ge 0\), \(G[H^i[u]]\) is well-typed.

Let \(c = j_2 - j_1\) and \(d = (j_0 + 1) - (j_2 - j_1)\).
Let \(i \ge 0\). 
By Lemma~\ref{lem:substFst} we have
\[
\pf C[D^{ci+d}[t]] : \T \tr 
G[H^i[u]]
\]
and hence 
\[
\tree(C[D^{ci+d}[t]]) = \remeps{\leaves(\tree(G[H^i[u]]))}
\]
by Lemma~\ref{lem:fst-crrectness}.
Also, by Lemma~\ref{lem:orderDecreaseFst}, \(G[H^i[u]]\) is order-\((n-1)\).

Finally, by Lemma~\ref{lem:substFst} we have
\[
z : \uty_{j_2} \pf C[D^{ci+j_1}[z]] : \T \tr 
G[H^i[z_{\uty_{j_2}}]]
\]
and hence
\(G[H^i[z_{\uty_{j_2}}]]\) is linear by Lemma~\ref{lem:linPres}; thus, \(G\) and \(H\) are %
 linear.
\end{proof}

For the second transformation, we have the following:
\begin{lemma}
\label{lem:iteLinPresSnd}
Given order-\(n\) \(\stlambda\)-contexts \(C\), \(D\), and order-\(n\) \(\stlambda\)-term \(t\)
such that 
\begin{itemize}
\item
the constants in \(C,\,D,\,t\) are in a \(\br\)-alphabet
\item
\(\set{ \remeps{\leaves(\tree(C[D^{\ell_i}[t]]))} \mid i \ge 0 }\) is infinite for any strictly increasing sequence
 \((\ell_i)_i\).
\item
\(C\) and \(D\) are %
 linear
\end{itemize}
there exist order-\(n\) \(\stlambda\)-contexts \(G\), \(H\), order-\(n\) \(\stlambda\)-term \(u\),
and some constant numbers \(c,\,d \ge 1\)
such that 
\begin{itemize}
\item
the constants in \(G,\,H,\,u\) are in a \(\br\)-alphabet
\item
for \(i \ge 0\),
\[
\remeps{\leaves(\tree(C[D^{ci+d}[t]]))} = 
\begin{cases}
\empword & (\leaves(\tree(G[H^i[u]])) = \Te)
\\
\leaves(\tree(G[H^i[u]])) & (\leaves(\tree(G[H^i[u]])) \neq \Te)
\end{cases}
\]
\item
\(G\) and \(H\) are %
 linear.
\end{itemize}
\end{lemma}

Lemma~\ref{lem:iteLinPres} is a corollary of Lemmas~\ref{lem:iteLinPresFst} and~\ref{lem:iteLinPresSnd}.

\section{Properties of Homeomorphic Embedding on Tree Functions}
\label{sec:treeFunction}

This section shows the well-definedness of \(\tle_{\sty}\) and its reflexivity and transitivity.
We write \(\eqbetaeta\) for the \(\beta\eta\)-equivalence relation on the simply-typed
\(\lambda\)-terms.
\begin{lemma}[well-definedness]
\label{lem:red-preserves-tle}
Let \(t,s,t',s'\) be terms of type \(\sty\).
If \(t\eqbetaeta t'\) and \(s\eqbetaeta s'\), then
\(t\tle_{\sty} s\) if and only if \(t'\tle_{\sty} s'\).
\end{lemma}
\begin{proof}
Let \(\sty=\sty_1\to\cdots\to\sty_k\to\T\).
Then 
\[\begin{array}{l}
t\tle_{\sty} s\\
\IFF t\, s_1\,\cdots\,s_k \tle_{\T}s\, s'_1\,\cdots\,s'_k 
\mbox{ for every \(s_1,\ldots,s_k,s'_1,\ldots,s'_k\) such that
\(s_i\tle_{\sty_i}s_i'\)}\\
\IFF \eval{t\, s_1\,\cdots\,s_k} \tle \eval{s\, s'_1\,\cdots\,s'_k}
\mbox{ for every \(s_1,\ldots,s_k,s'_1,\ldots,s'_k\) such that
\(s_i\tle_{\sty_i}s_i'\)}\\
\IFF \eval{t'\, s_1\,\cdots\,s_k} \tle \eval{s'\, s'_1\,\cdots\,s'_k}
\mbox{ for every \(s_1,\ldots,s_k,s'_1,\ldots,s'_k\) such that
\(s_i\tle_{\sty_i}s_i'\)}\\
\IFF t'\, s_1\,\cdots\,s_k \tle_{\T}s'\, s'_1\,\cdots\,s'_k 
\mbox{ for every \(s_1,\ldots,s_k,s'_1,\ldots,s'_k\) such that
\(s_i\tle_{\sty_i}s_i'\)}\\
\IFF t'\tle_{\sty} s'
\end{array}
\]
\end{proof}

The following is the abstraction lemma of the logical relation.
\begin{lemma}
\label{lem:Gtle-reflexivity}
If \(\TE \p t:\sty\), then \(\TE \models t \tle_{\sty} t\).
\end{lemma}
\begin{proof}
This follows by induction on the derivation of \(\TE\p t:\sty\).
Since the other cases are trivial, we show only the case where
\(t\) is a \(\lambda\)-abstraction. 
In this case, we have
\[
\begin{array}{l}
t=\lambda x_0\COL\sty_0.t' \qquad \sty=\sty_0\to\sty'\qquad
\TE,x_0\COL\sty_0 \p t':\sty'
\end{array}
\]
Let \(\TE = x_1\COL\sty_1,\ldots,x_k\COL\sty_k\).
We need to show that
\([s_1/x_1,\ldots,s_k/x_k]t \tle_{\sty} 
[s'_1/x_1,\ldots,s'_k/x_k]t \)
holds for every \(s_1,\ldots,s_k,s'_1,\ldots,s'_k\) such that \(s_i\tle_{\sty_i}s'_i\)
for each \(i\in\set{1,\ldots,k}\).
i.e.,
\[([s_1/x_1,\ldots,s_k/x_k]t) s_0 \tle_{\sty'} 
([s'_1/x_1,\ldots,s'_k/x_k]t) s_0' \]
holds for every \(s_0,s_1,\ldots,s_k,s'_0,s'_1,\ldots,s'_k\) such that \(s_i\tle_{\sty_i}s'_i\)
for each \(i\in\set{0,1,\ldots,k}\).
By Lemma~\ref{lem:red-preserves-tle}, it suffices to show
\[[s_0/x_0,s_1/x_1,\ldots,s_k/x_k]t' \tle_{\sty'} 
[s'_0/x_0,s'_1/x_1,\ldots,s'_k/x_k]t'. \]
This follows immediately from 
the induction hypothesis and \(\TE,x_0\COL\sty_0\p t':\sty'\).
\end{proof}

We are now ready to show that \(\tle_{\sty}\) is reflexive and transitive.
\begin{lemma}
\label{lem:tle-is-transitive}
\(\tle_{\sty}\) is reflexive and transitive.
\end{lemma}
\begin{proof}
The reflexivity follows immediately as a special case of Lemma~\ref{lem:Gtle-reflexivity},
where \(\Gamma=\emptyset\).
We show the transitivity by induction on \(\sty\).
The base case follows immediately from the definition.
For the induction step, suppose 
\(t_1\tle_{\sty_1\to\sty_2} t_2\) and \(t_2\tle_{\sty_1\to\sty_2} t_3\).
Suppose also \(s_1\tle_{\sty_1}s_3\). We need to show \(t_1s_1\tle_{\sty_2}t_3s_3\).
Since \(s_1\tle_{\sty_1}s_1\), we have 
\(t_1s_1\tle_{\sty_2} t_2s_1\) and \(t_2s_1\tle_{\sty_2} t_3s_3\).
By the induction hypothesis, we obtain 
\(t_1s_1\tle_{\sty_2} t_3s_3\) as required.
\end{proof}

As a corollary, it follows that every tree function represented by the
simply-typed \(\lambda\)-calculus is monotonic.
\begin{corollary}[monotonicity]
\label{lem:monotonicity}
If \(\emptyset\p t :\sty_1\to \sty_2\), then 
\(t\,s_1 \tle_{\sty_2} t\,s'_1\)
for every \(s_1,s_1'\) such that
\(s_1\tle_{\sty_1}s_1'\).
\end{corollary}

\end{document}